\pdfoutput=1
\pdftrailerid{<417869656e742d50617065722d49492d72302e332e332d706466>}
\documentclass[11pt]{article}

\usepackage[T1]{fontenc}
\usepackage[utf8]{inputenc}
\usepackage{lmodern}
\usepackage[a4paper,margin=1in]{geometry}
\usepackage{microtype}
\usepackage{amsmath,amssymb,amsthm,mathtools}
\usepackage{aliascnt}
\usepackage{booktabs,longtable,array,multirow,tabularx}
\usepackage{graphicx}
\usepackage{float}
\usepackage{enumitem}
\usepackage{xcolor}
\usepackage{natbib}
\usepackage{xurl}
\usepackage{xspace}
\usepackage{hyperref}
\usepackage[nameinlink,capitalise,noabbrev]{cleveref}
\usepackage{listings}
\usepackage{algorithm}
\usepackage{algpseudocode}
\usepackage{tikz}
\usetikzlibrary{arrows.meta,positioning,shapes.geometric,fit,calc}

\hypersetup{
  colorlinks=true,
  linkcolor=black,
  citecolor=black,
  urlcolor=black,
  pdftitle={Axient: On-Chain Credit and Loss Allocation for Leveraged Event Markets},
  pdfauthor={Maksym Nechepurenko},
  pdfsubject={A venue-agnostic on-chain credit protocol with endogenous capital supply, strategic execution liquidity, withdrawal dynamics, and agent-based validation},
  pdfkeywords={prediction markets, event contracts, DeFi lending, agent-based computational economics, margin finance, liquidation, credit pools, loss waterfalls, market making, reserves, financial contagion}
}
\newtheorem{definition}{Definition}[section]
\newaliascnt{assumption}{definition}
\newtheorem{assumption}[assumption]{Assumption}
\aliascntresetthe{assumption}
\newaliascnt{proposition}{definition}
\newtheorem{proposition}[proposition]{Proposition}
\aliascntresetthe{proposition}
\newaliascnt{theorem}{definition}
\newtheorem{theorem}[theorem]{Theorem}
\aliascntresetthe{theorem}
\newaliascnt{corollary}{definition}
\newtheorem{corollary}[corollary]{Corollary}
\aliascntresetthe{corollary}
\newaliascnt{lemma}{definition}
\newtheorem{lemma}[lemma]{Lemma}
\aliascntresetthe{lemma}
\newaliascnt{remark}{definition}
\newtheorem{remark}[remark]{Remark}
\aliascntresetthe{remark}
\newaliascnt{example}{definition}

\aliascntresetthe{example}

\crefname{definition}{definition}{definitions}
\Crefname{definition}{Definition}{Definitions}
\crefname{assumption}{assumption}{assumptions}
\Crefname{assumption}{Assumption}{Assumptions}
\crefname{proposition}{proposition}{propositions}
\Crefname{proposition}{Proposition}{Propositions}
\crefname{theorem}{theorem}{theorems}
\Crefname{theorem}{Theorem}{Theorems}
\crefname{corollary}{corollary}{corollaries}
\Crefname{corollary}{Corollary}{Corollaries}
\crefname{lemma}{lemma}{lemmas}
\Crefname{lemma}{Lemma}{Lemmas}
\crefname{remark}{remark}{remarks}
\Crefname{remark}{Remark}{Remarks}
\crefname{example}{example}{examples}
\Crefname{example}{Example}{Examples}

\newcommand{\Axient}{\textnormal{\textsc{Axient}}\xspace}
\newcommand{\R}{\mathbb{R}}
\newcommand{\E}{\mathbb{E}}

\newcommand{\one}{\mathbf{1}}

\newcommand{\Pset}{\mathcal{P}}
\newcommand{\Vset}{\mathcal{V}}
\newcommand{\Eset}{\mathcal{E}}
\newcommand{\Iset}{\mathcal{I}}
\newcommand{\Mset}{\mathcal{M}}

\newcommand{\pospart}[1]{\left[#1\right]^+}
\newcommand{\state}[1]{\texttt{#1}}

\lstdefinestyle{axientcode}{
  basicstyle=\ttfamily\small,
  frame=single,
  breaklines=true,
  columns=fullflexible,
  showstringspaces=false,
  keywordstyle=\bfseries,
  commentstyle=\itshape,
  xleftmargin=0.5em,
  xrightmargin=0.5em
}
\title{\textbf{Axient: On-Chain Credit and Loss Allocation for Leveraged Event Markets}\\[0.35em]\large A Venue-Agnostic Protocol for Traders, Credit Providers, Market Makers, and Liquidation Backstops}
\author{Maksym Nechepurenko\thanks{Founder and Director of Research, ForesightFlow, the Research Department of Devnull FZCO, Dubai, United Arab Emirates. Email: \href{mailto:maksym@devnull.ae}{maksym@devnull.ae}. Research profile: \url{https://www.foresightflow.org/}.}}
\date{July 15, 2026}

\begin{document}
\maketitle

\begin{abstract}
A physically backed leveraged event position requires real credit. If a trader contributes collateral $C$ and receives leverage $L$, the protocol must supply $(L-1)C$ in stablecoin and use the combined amount to acquire recognized event exposure. Once third-party capital supplies that financed leg, the product becomes a credit protocol whose design must specify lender priority, liquidation authority, withdrawal liquidity, loss allocation, and the relationship between executable market liquidity and financial capital.

This paper develops a venue-agnostic on-chain architecture for that capital layer and extends it from static capital accounting to an endogenous capital market. The protocol separates traders; Senior Credit LPs; market makers; liquidators; and Liquidation Backstop Providers (LBPs), which supply junior loss-absorbing provider capital ahead of Senior principal. We formalize pool and debt shares, utilization- and risk-sensitive interest, collateral-locked position accounts, venue capability vectors, collateralized market-maker commitments, withdrawal queues, isolated risk pools, non-redeemable reserves, and a deterministic loss waterfall. We then endogenize Senior and LBP participation, trader leverage demand, market-maker delivery and strategic withdrawal, liquidator entry, reserve replenishment, withdrawal coordination, and common-factor contagion.

The formal results establish balanced real- and integer-unit accounting; settlement-confirmed debt priority; trader ownership of value remaining after admitted debt and fees; cumulative and idempotent partial settlement; non-dilutive share issuance; an exact deterministic waterfall; junior-before-Senior impairment; request-time neutrality for loss-participating withdrawal queues; existence of a utilization fixed point; sufficient conditions for its uniqueness and local stability; conditions for market-maker commitment delivery and liquidator coverage; a capacity-induced-shortfall result showing that added backstop capacity can increase aggregate loss when admission expands too quickly; finite-time reserve replenishment under sustained positive inflow; and persistence of common-factor loss covariance despite strict ledger isolation.

No live Axient or venue data are used. The release preserves the 28 exact fixtures and 31,082 deterministic checks from r0.2.2 and adds a fixed-seed agent-based experiment with 96 two-year paths, three pools, four protocol configurations, three leverage policies, and 636 heterogeneous agents per path. The run records 70,207,488 scalar invariant evaluations with zero failures. Six of seven release-registered hypotheses pass. Under the phased leverage mix, the full configuration reduces two-year path-level Senior-loss incidence from 99.0 to 57.3 percent relative to the Senior-only benchmark, but still admits only 38.0 percent of requested credit, impairs LBP capital in 94.8 percent of paths, and leaves material Senior loss. The bonded-market-maker hypothesis fails: modest commitment capacity expands admitted demand faster than it reduces unit shortfall. A $5\times$-heavy policy raises Senior-loss incidence by 159.1 percent relative to $2\times$ only, while a common stablecoin factor preserves 0.842 cross-pool raw-shortfall correlation despite local isolation. An independently coded daily robustness model reproduces the directional value of layered protection, the leverage and hard-flat trade-off, and the withdrawal-buffer phase transition while producing materially different numerical levels. The disagreement reinforces the model-risk boundary. These are synthetic mechanism comparisons under author-specified behavioral rules, not forecasts of APY, default probability, provider participation, venue liquidity, settlement reliability, or production safety. Empirical calibration, external contract audit, and live adapter validation remain separate tasks.
\end{abstract}


\medskip
\noindent\textit{Research and implementation disclosure.} \Axient is an active research-and-development initiative of the author.\footnote{A non-archival project page is maintained at \url{https://axient.app}.} This creates an interest in the protocol's success. The paper separates accounting identities, theorem-conditional guarantees, registered scenario results, synthetic comparative results, implementation assumptions, and unvalidated deployment claims. It does not promise LP yield, principal protection, production safety, venue liquidity, or deployability beyond the conditions stated in each result.

\section{Introduction}
\label{sec:introduction}

\paragraph{Manuscript status (r0.3.3).}
This reproducible publication build preserves the deterministic r0.3.1 numerical evidence while presenting the agent-based economic revision of the repository-linked r0.2.2 protocol specification with a publication-grade vector visualization layer. The real- and integer-accounting results, tagged reference implementation, 28 exact fixtures, and 31,082 deterministic checks are preserved. The new contribution is a release-registered, fixed-seed agent model in which Senior Credit LPs, Liquidation Backstop Providers (LBPs), traders, market makers, and liquidators respond to rates, losses, queue delay, adverse selection, bounties, and one another. The model adds endogenous capital supply, strategic execution capacity, withdrawal coordination, dynamic reserves, and common-factor contagion. Its behavioral parameters are author-specified and synthetic; the experiment is a mechanism stress test, not empirical calibration.

A leveraged purchase of an event outcome token is a financed transaction. If a trader contributes collateral $C>0$ and receives leverage $L>1$, the gross acquisition notional is
\begin{equation}
N=LC,
\label{eq:gross-notional}
\end{equation}
and the financed amount is
\begin{equation}
D_0=(L-1)C.
\label{eq:initial-debt-intro}
\end{equation}
If the position is described as physically backed, the amount $D_0$ cannot be only an internal number. Stablecoin capital equal to the financed leg must be transferred into a recognized position account and used, together with trader collateral, to acquire the outcome exposure. The protocol must then specify who supplies that capital, who receives interest, who can liquidate the financed asset, which claim has priority over sale or redemption proceeds, how withdrawals are handled while credit is deployed, and who bears loss when realized recovery is insufficient.

These questions change the economic category of the product. A leveraged event-market interface becomes a credit protocol. The protocol studied here has five roles:
\begin{enumerate}[leftmargin=*]
\item \emph{traders}, who provide first-loss collateral and demand event exposure;
\item \emph{Senior Credit LPs}, who supply stablecoin principal and receive priority repayment plus a share of interest;
\item \emph{market makers}, who supply ordinary quotes or collateralized executable commitments;
\item \emph{liquidators}, who perform admissible risk-reducing execution and receive a bounty; and
\item \emph{Liquidation Backstop Providers} (LBPs), who commit junior loss-absorbing provider capital ahead of Senior principal and receive enhanced protocol income.
\end{enumerate}
The last two roles are deliberately distinct. A liquidator supplies an action; an LBP supplies a balance sheet. The same address or institution may perform both, but the mechanism must not depend on that coincidence.

The position-level primitive is inherited from \citet{nechepurenko2026axient}. Debt is treated as repaid only after sale proceeds are settled and applied to the obligation; a matched trade, provisional venue balance, or unconfirmed redemption does not extinguish debt. That work separates leverage maturity from claim maturity and proves Debt-Free Finality conditional on successful pre-finality debt extinction. The present paper studies the protocol-level questions that follow once the financed leg comes from external capital providers.

\subsection{Research question}
The central question is:
\begin{quote}
How can an on-chain protocol finance physically backed leveraged event positions across heterogeneous venues while preserving lender priority, transparent provider returns, deterministic loss allocation, withdrawal honesty, and explicit failure boundaries when capital and execution participants respond endogenously to protocol state?
\end{quote}

The answer cannot be ``put the debt in a database.'' A backend ledger does not create an enforceable lien, prevent collateral escape, segregate reserves, protect pool-share accounting, or constrain governance from changing priority after capital has been deposited. Conversely, moving accounting on-chain does not guarantee that external venues will maintain executable depth, accept liquidations, resolve events correctly, or settle on time. Nor does static accounting show whether providers stay, whether market makers withdraw, whether liquidators enter, or whether queues become self-reinforcing. The paper therefore separates contract guarantees, behavioral mechanisms, and external assumptions.

\subsection{Contributions}
The paper makes thirteen contributions.
\begin{enumerate}[leftmargin=*]
\item \textbf{Protocol category and role separation.} It specifies Axient as a venue-agnostic event-margin credit protocol rather than a leveraged frontend and separates Senior credit, ordinary market-making liquidity, liquidation execution, protocol reserve equity, and LBP provider capital.
\item \textbf{On-chain and discrete accounting model.} It formalizes pool assets, claims, debt shares, Senior and junior pool shares, settlement-confirmed repayment, debt-first trader residual release, partial-settlement idempotency, fee routing, integer rounding, and double-entry transition closure.
\item \textbf{Credit and leverage geometry.} It derives the exact credit fraction $1-1/L$ and gross exposure per unit of credit $L/(L-1)$, showing why $5\times$ requires 80 percent real financing and supports only 1.25 units of gross exposure per unit of credit.
\item \textbf{Deterministic loss allocation.} It defines a finite waterfall across position buffers, local reserves, LBP capital, pool reserves, optional global reserves, and Senior principal, with existence, uniqueness, and conservation results.
\item \textbf{Collateralized market-maker commitments.} It distinguishes revocable displayed depth from enforceable capacity and proves when delivered value plus bond slashing can substitute for reserve.
\item \textbf{Isolation and protection contagion.} It proves no direct local loss contagion under strict pool isolation and identifies liability-independent correlation and protection-contagion channels created by common factors and a shared reserve.
\item \textbf{Withdrawal-liquidity and first-exit discipline.} It models share-denominated queues, requires queued shares to remain exposed to NAV changes until payment, and gives clearance and request-time-neutrality results.
\item \textbf{Endogenous capital-market equilibrium.} It defines heterogeneous Senior, LBP, and trader response functions; proves fixed-point existence; gives contraction conditions for uniqueness and stability; and identifies wrong-way utilization pricing under stress.
\item \textbf{Strategic execution supply.} It formalizes market-maker delivery incentives and permissionless liquidator entry and derives a capacity-induced-shortfall condition under which additional bonded capacity increases aggregate loss by expanding admitted credit faster than it reduces unit shortfall.
\item \textbf{Dynamic reserve and run mechanics.} It gives sufficient conditions for finite-time reserve replenishment and for a unique stable withdrawal state while retaining explicit run and repeated-loss failure boundaries.
\item \textbf{Agent-based capital-market experiments.} It introduces a release-registered weekly model with 636 heterogeneous agents per path and a separately coded daily robustness model with different time scale, aggregation, policy definitions, and scenario grids.
\item \textbf{Negative design findings.} It reports the failure of the registered bonded-MM hypothesis, frequent junior-capital impairment, material residual Senior loss, high sensitivity to $5\times$ exposure, and persistent common-factor contagion despite strict local ledgers.
\item \textbf{Repository-linked reproducibility.} It preserves the tagged implementation-parity release and adds versioned agent parameters, source, outputs, tables, figures, release-registered hypotheses, and 70,207,488 zero-failure scalar invariant evaluations.
\end{enumerate}

\subsection{Evidence discipline}
The mathematical results are conditional statements about the specified transition system. The deterministic verifier and tagged reference harness check arithmetic, transitions, and registered invariants; neither is a proof substitute or an external audit. The earlier stochastic experiments and the new agent model are synthetic, author-specified mechanism comparisons. Release registration records hypotheses before the final run, but it is not external preregistration. The outputs are not observations of Axient, a live venue, or an existing DeFi credit pool and cannot support a public APY, a production default probability, a capital-adequacy promise, or a claim that one configuration is commercially optimal.

This discipline follows the broader ForesightFlow programme's separation of structural results, observed microstructure, counterfactual replay, and policy parameters \citep{nechepurenko2026resolutionaware,nechepurenko2026fillside}. Here the restriction remains stronger: there are no live Axient observations.

\subsection{Roadmap}
\Cref{sec:related-work} positions the protocol relative to prediction markets, DeFi credit, liquidation, bank runs, financial networks, and agent-based computational economics. \Cref{sec:setting} defines roles, pool state, venue capability classes, and the position-level interface. \Cref{sec:accounting}--\ref{sec:interest} formalize accounting, shares, debt, rates, and returns. \Cref{sec:liquidation-roles} and \Cref{sec:mm} separate liquidators, LBPs, and market makers. \Cref{sec:waterfall}--\ref{sec:withdrawals} develop loss allocation, isolation, reserves, and queues. \Cref{sec:leverage} derives credit-capacity and leverage constraints. \Cref{sec:architecture} specifies the contract and agent architecture; \Cref{sec:simulation-design}--\ref{sec:simulation-results} preserve the registered static synthetic evaluation. \Cref{sec:endogenous-capital}--\ref{sec:runs-reserves} develop the endogenous and strategic extensions, and \Cref{sec:agent-design}--\ref{sec:agent-results} report the new agent experiment. The remaining sections state contract guarantees, security invariants, implications, limitations, and conclusions. Detailed proofs, algorithms, parameters, and reproducibility materials appear in the appendices.

\section{Related Work}
\label{sec:related-work}

The paper lies at the intersection of prediction-market design, DeFi lending, liquidation microstructure, liquidity transformation, and financial-network loss allocation.

\subsection{Prediction markets and leveraged event instruments}
Prediction markets have traditionally been studied as information-aggregation mechanisms and as markets for bounded claims \citep{hanson2003,wolfers2004,manski2006}. The ForesightFlow Event-Linked Perpetuals programme establishes why ordinary crypto-perpetual mechanics do not port cleanly to event claims: bounded support, terminal collapse, asymmetric depth, and oracle-mediated finality create risk channels absent from continuous underlyings \citep{nechepurenko2026resolutionaware}. The taxonomy paper extends the instrument space beyond single binary markets \citep{nechepurenko2026taxonomy}; the manipulation paper shows that leverage changes both market-price and real-world outcome-manipulation incentives \citep{nechepurenko2026manipulation}; the fill-side study characterizes the concentrated non-retail participation and quote-attribution constraints of a Polymarket-class venue \citep{nechepurenko2026fillside}.

The first Axient paper responds to one specific terminal-risk problem by separating leverage maturity from claim maturity and extinguishing debt before the event claim enters non-tradable finality \citep{nechepurenko2026axient}. The present work begins where that paper stops: it asks how the financed leg is formed, priced, withdrawn, protected, and written down when supplied by third parties.

\subsection{DeFi money markets and adaptive rates}
Compound and Aave popularized pooled stablecoin lending, interest-bearing shares, utilization-based borrow rates, reserve factors, and permissionless liquidation \citep{compound2019,aave2026}. The general DeFi literature documents how composability, overcollateralization, oracle dependence, and liquidation incentives distinguish protocol credit from conventional banking \citep{werner2022,gudgeon2020}. Recent work studies adaptive rate setting and the interaction between fast market feedback and slower governance or learning loops \citep{bastankhah2024agilerate,bastankhah2024fastslow}.

Axient inherits the accounting primitives but not the collateral geometry. In ordinary overcollateralized lending, the borrower pledges an asset that exists independently of the loan. Here trader equity and borrowed stablecoin jointly create the financed event position. The position may resolve to zero, lose executable liquidity before finality, or depend on an external venue's signing and settlement model. The design is therefore closer to secured margin finance or on-chain prime brokerage than to an ordinary money market.

\subsection{Liquidation and execution externalities}
DeFi liquidation research documents the dependence of protocol solvency on liquidation incentives, gas costs, price impact, and market congestion \citep{perez2020,qin2023miqado,sadeghi2026}. Liquidators can protect lenders while also extracting rents or amplifying price dislocations. The distinction between liquidation execution and loss-bearing capital is not always explicit in deployed systems. This paper makes it central: liquidators are action providers, whereas LBPs are junior capital providers. Market makers form a third class because they can supply executable capacity without supplying either senior or junior credit.

\subsection{Funding liquidity, runs, and networks}
The interaction between market liquidity and funding liquidity is a classic source of financial amplification \citep{brunnermeier2009}. Pooled maturity transformation creates withdrawal-run incentives even when a balance sheet is solvent at book value \citep{diamond1983}. Financial-network models show how shared obligations and protection can transmit or absorb losses depending on priority and clearing rules \citep{eisenberg2001}.

Axient's isolated pools, withdrawal queues, and optional global reserve map these questions into smart-contract form. Strict isolation prevents direct transfer of one pool's liability to another. A shared global reserve can still create \emph{protection contagion}: one pool's draw reduces protection available elsewhere without transferring the original liability.

\subsection{Boundary of the contribution}
The paper does not claim an optimal rate curve, welfare-optimal waterfall, complete equilibrium among traders and capital providers, or empirically calibrated loss distribution. Its contribution is a formal and executable reference architecture: explicit roles, enforceable claims, conservation rules, finite-scenario guarantees, and reproducible synthetic comparisons that can later be replaced or calibrated with venue and protocol data.

\subsection{Agent-based computational economics and heterogeneous finance}
The representative-agent abstraction is poorly suited to the present questions because Senior LPs, LBPs, traders, market makers, and liquidators have different objectives, constraints, and timing. Agent-based computational economics treats aggregate outcomes as the result of interacting heterogeneous rules rather than an imposed single equilibrium agent \citep{tesfatsion2006,farmer2009}. A recent DeFi-focused model applies the same approach to multi-asset lending protocols and motivates agent-level treatment of deposits, borrowing, collateral, and market feedback \citep{chaudhary2022abm}. Heterogeneous-belief asset-pricing models show how boundedly rational strategy switching can create endogenous instability and multiple regimes \citep{brock1998,lux1999}. The validation literature also warns that an internally coherent agent model is not empirically validated merely because it reproduces plausible stylized outcomes \citep{fagiolo2007}. We therefore use the agent model as a registered synthetic stress laboratory and report negative findings rather than as an estimator of production probabilities.

\subsection{Runs, liquidity feedback, and network contagion}
Classical run models establish that liquidity transformation can create self-fulfilling withdrawal pressure even when long-run assets are valuable \citep{diamond1983,goldstein2005}. In Axient the queue is share-denominated and loss participating, which removes one request-time transfer but not expectations about delay, future impairment, or peer withdrawals. Market and funding liquidity can also reinforce one another \citep{brunnermeier2009}, while network research distinguishes direct contractual contagion from correlated common exposures and shared-protection channels \citep{gai2010,acemoglu2015,eisenberg2001}. These distinctions motivate the separate treatment of local ledger isolation, global reserve depletion, and common stablecoin, venue, and oracle factors.

\subsection{Strategic liquidity provision and adverse selection}
Market-making capacity is endogenous to adverse selection, inventory risk, and latency-sensitive placement incentives \citep{lehalle2016adverse,fodra2012inventory}. Sequential-trade and strategic-information models explain why liquidity suppliers widen or withdraw when order flow becomes more toxic \citep{glosten1985,kyle1985}. An event-credit protocol cannot treat visible book depth as committed capacity: ordinary quotes are revocable, while a bonded commitment has value only if delivery or slashing is enforceable and the bond retains value in the protected state. The agent model therefore separates ordinary MM capacity, bonded capacity, strategic withdrawal, and liquidator competition.

\section{Protocol Setting}
\label{sec:setting}

\subsection{Sets and time}
Let $p\in\Pset$ index credit pools, $v\in\Vset$ venues, $e\in\Eset$ events, $i\in\Iset$ financed positions, and $m\in\Mset$ market makers. Contract accounting is discrete in block or checkpoint time $t\in\mathbb{N}_0$. Continuous interest notation is used where convenient, but the reference implementation uses a debt index updated at discrete transitions.

All monetary quantities are denominated in one pool settlement asset. Cross-asset and bridge risk require separate pools or explicit conversion haircuts; the base model does not net heterogeneous stablecoins at par.

\subsection{Economic roles}
\begin{definition}[Trader]
A trader supplies collateral $C_i$, chooses an admissible leverage tier $L_i$, and owns residual position value after the debt-first rule has been satisfied.
\end{definition}

\begin{definition}[Senior credit liquidity provider]
A senior credit liquidity provider supplies lendable stablecoin capital to pool $p$ and receives senior pool shares. Debt repayment has priority over trader residual value. If recognized position proceeds are insufficient, designated reserves and junior LBP capital are impaired before senior principal; senior capital is protected by those layers but is not guaranteed against loss.
\end{definition}

\begin{definition}[Liquidator]
A liquidator is an execution agent that submits an admissible risk-reducing transition and receives a contract-defined bounty. Validation of the Axient-side transition may be permissionless; external venue execution remains constrained by the active adapter capability. A liquidator need not supply pool capital and need not absorb principal loss.
\end{definition}

\begin{definition}[Liquidation backstop provider]
A liquidation backstop provider supplies junior loss-absorbing provider capital $J_{p,t}$ to an isolated pool or risk bucket. LBP capital absorbs shortfall before senior principal, but may stand behind protocol-funded market or position reserves, and receives a designated share of interest or fees. An LBP may also operate a liquidator, but the roles are not identified in protocol state.
\end{definition}

\begin{definition}[Market maker]
A market maker supplies ordinary quotes on an external venue, a protocol RFQ, or a collateralized executable commitment. Only capacity with an objective delivery condition and enforceable failure value can be credited as robust protection under \Cref{sec:mm}.
\end{definition}

\subsection{Pool and position state}
For each pool $p$, define:
\begin{align*}
A_{p,t} &\ge 0 &&\text{aggregate settled cash controlled by the local pool},\\
A^L_{p,t} &\ge 0 &&\text{unencumbered senior cash eligible for lending},\\
B_{p,t} &\ge 0 &&\text{recognized outstanding debt receivable},\\
V^S_{p,t} &\ge 0 &&\text{net asset value allocated to senior shares},\\
V^J_{p,t} &\ge 0 &&\text{net asset value allocated to LBP shares},\\
R^{M}_{p,e,v,t} &\ge 0 &&\text{market/event/venue-specific reserve},\\
R^{P}_{p,t} &\ge 0 &&\text{pool reserve},\\
G_t &\ge 0 &&\text{optional protocol-global reserve},\\
W_{p,t} &\ge 0 &&\text{outstanding withdrawal claims},\\
Z^S_{p,t}, Z^J_{p,t} &\ge 0 &&\text{senior and junior share supplies}.
\end{align*}
The global reserve is optional. A pool mandate may prohibit its use to preserve strict isolation.

For position $i$ assigned to pool $p(i)$, define:
\begin{align*}
C_i &>0 &&\text{trader collateral},\\
L_i &\in [1,5] &&\text{admitted leverage},\\
N_i &=L_iC_i &&\text{gross acquisition notional},\\
D_{i,t} &\ge 0 &&\text{recognized debt},\\
K_{i,t} &\ge 0 &&\text{settled position cash},\\
Q_{i,t} &&\text{recognized assets or venue claims},\\
P^{\mathrm{set}}_{i,t} &\ge 0 &&\text{settled proceeds available for allocation},\\
P^{\mathrm{pend}}_{i,t} &\ge 0 &&\text{unsettled proceeds}.
\end{align*}
Pending proceeds do not reduce $D_{i,t}$.

\subsection{Venue capability vectors and backing modes}
Venue independence is encoded through capability rather than branding.

\begin{definition}[Venue capability vector]
A venue adapter publishes
\begin{equation}
\chi_v=(c_v,\ell_v,a_v,s_v,r_v,h_v)\in\{0,1\}^6,
\label{eq:capability-vector}
\end{equation}
where the coordinates indicate contract-controlled custody, enforceable lender lien, delegated liquidation authority, verifiable trade settlement, verifiable redemption, and withdrawal lock while debt is outstanding.
\end{definition}

The binary vector is deliberately conservative. A production registry may use graded scores and evidence hashes, but the pool mandate must still decide whether a capability class is admissible.

\begin{definition}[Backing mode]
A position belongs to one of four backing modes.
\begin{enumerate}[leftmargin=*]
\item \textbf{Native physically backed:} the recognized outcome asset is held in a contract-controlled account with enforceable lien, withdrawal lock, and delegated liquidation.
\item \textbf{Verifiable venue settlement:} the asset and settlement are on-chain verifiable, but some signing, matching, or liquidation authority remains venue- or operator-dependent.
\item \textbf{Attested or custodial:} ownership or liquidation relies on a custodian, controlled signer, or attestation.
\item \textbf{Protocol-native synthetic:} Axient records the trader claim and hedges net exposure externally; the user claim is not represented one-for-one by an outcome token.
\end{enumerate}
\end{definition}

Modes are not interchangeable. A pool may admit multiple venues only when its mandate specifies allowed capability classes, concentration limits, and loss treatment. A senior pool that represents itself as natively physically backed cannot silently include synthetic or attested positions.

\subsection{Pool mandates}
\begin{definition}[Pool mandate]
A pool mandate is
\begin{equation}
\mu_p=(\mathcal A_p,\mathcal V_p,\mathcal E_p,\mathcal L_p,\mathcal C_p,\mathcal R_p),
\end{equation}
where $\mathcal A_p$ is the allowed settlement-asset set, $\mathcal V_p$ the venue and capability set, $\mathcal E_p$ the event-class set, $\mathcal L_p\subseteq\{1,2,3,5\}$ the leverage tiers, $\mathcal C_p$ concentration limits, and $\mathcal R_p$ the reserve and waterfall policy.
\end{definition}
A mandate is prospective. Governance may alter future admission but cannot retroactively subordinate an outstanding claim or reorder the waterfall for an existing position.

\subsection{Position-level primitive}
The protocol assumes that every financed position has a position-level debt-clearing certificate from \citet{nechepurenko2026axient}. At decision time $u$, the risk controller identifies an admissible sale quantity $\widehat x_{i,u}$ such that a lower envelope of settlement-confirmed proceeds covers an upper envelope of debt and expenses.

\begin{assumption}[Certified debt-clearing interface]
\label{ass:position-certificate}
For each admitted position $i$, the risk layer publishes
\begin{equation}
\mathfrak C_i=(\widehat x_{i,u},\underline B_{i,u,\Delta},\overline H_{i,u,\Delta},m_i,\Delta_i)
\end{equation}
with
\begin{equation}
K_{i,u}+\underline B_{i,u,\Delta}(\widehat x_{i,u})
\ge
\overline H_{i,u,\Delta}+m_i.
\label{eq:position-certificate}
\end{equation}
The certificate is valid only inside its registered operating set and settlement horizon.
\end{assumption}

Protocol-level admission adds pool cash, concentration, withdrawal-liquidity, reserve, provider, and shared-book constraints to \eqref{eq:position-certificate}. A position that lacks a valid position certificate cannot be made acceptable merely by increasing the reserve.

\subsection{Protocol network}
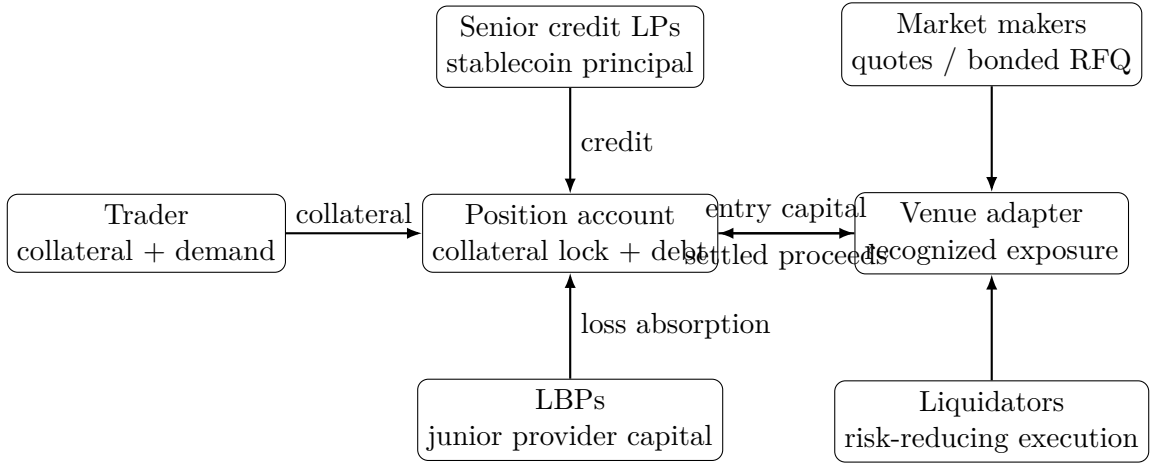
\begin{figure}[H]
\centering
\begin{tikzpicture}[
node distance=1.4cm and 1.8cm,
box/.style={draw, rounded corners, align=center, minimum width=3.0cm, minimum height=0.9cm},
arr/.style={-{Latex[length=2mm]}, thick}
]
\node[box] (trader) {Trader\\collateral + demand};
\node[box, right=of trader] (position) {Position account\\collateral lock + debt};
\node[box, above=of position] (senior) {Senior credit LPs\\stablecoin principal};
\node[box, below=of position] (lbp) {LBPs\\junior provider capital};
\node[box, right=of position] (venue) {Venue adapter\\recognized exposure};
\node[box, above=of venue] (mm) {Market makers\\quotes / bonded RFQ};
\node[box, below=of venue] (liq) {Liquidators\\risk-reducing execution};
\draw[arr] (trader) -- node[above,sloped]{collateral} (position);
\draw[arr] (senior) -- node[right]{credit} (position);
\draw[arr] (position) -- node[above]{entry capital} (venue);
\draw[arr] (venue) -- node[below]{settled proceeds} (position);
\draw[arr] (mm) -- (venue);
\draw[arr] (liq) -- (venue);
\draw[arr] (lbp) -- node[right]{loss absorption} (position);
\end{tikzpicture}
\caption{Axient protocol roles. Credit capital, execution liquidity, liquidation execution, and junior loss absorption are separate functions even when one institution performs more than one role.}
\label{fig:roles}
\end{figure}

\section{On-Chain Balance Sheet and Accounting}
\label{sec:accounting}

\subsection{Pool book value}
The contract ledger treats a credit pool as a balance sheet rather than as a collection of front-end positions. Let
\begin{equation}
\mathcal A_{p,t}=A_{p,t}+B_{p,t}
\label{eq:pool-assets}
\end{equation}
be recognized local pool assets: aggregate settled cash and recognized debt receivables. The cash term includes tagged or segregated sub-accounts for senior, junior, reserve, and protocol claims. Their encumbrance determines whether cash is lendable, not whether it is an asset. Let
\begin{equation}
\mathcal C_{p,t}=V^S_{p,t}+V^J_{p,t}+R^M_{p,t}+R^P_{p,t}+T_{p,t}
\label{eq:pool-claims}
\end{equation}
be local claims and equity layers: senior NAV, LBP NAV, aggregate market reserves assigned to the pool, pool reserve, and protocol fee claim retained inside the accounting boundary. The global reserve $G_t$ is accounted for separately because it may protect more than one pool.

\begin{assumption}[Single-asset book]
\label{ass:single-asset}
All terms in \eqref{eq:pool-assets}--\eqref{eq:pool-claims} are denominated in the same settlement asset. Assets requiring conversion are not recognized at par unless the pool mandate specifies a conversion oracle, haircut, and settlement path.
\end{assumption}

Let $A^L_{p,t}\le A_{p,t}$ denote the unencumbered senior-cash subset eligible for new lending. It is a control-state partition of aggregate cash, not an additional asset. At deployment or after a fully reconciled transition,
\begin{equation}
\mathcal A_{p,t}=\mathcal C_{p,t}.
\label{eq:pool-closure}
\end{equation}
The equality is book-value closure, not a statement that every debt receivable is liquid or riskless.

\subsection{Position balance sheet}
Let $X_{i,t}$ denote settlement-recognized value controlled by position account $i$. It excludes matched but unsettled proceeds. Define trader residual equity and latent credit shortfall by
\begin{align}
E_{i,t}&=\pospart{X_{i,t}-D_{i,t}},\label{eq:position-equity}\\
\ell_{i,t}&=\pospart{D_{i,t}-X_{i,t}}.\label{eq:position-shortfall}
\end{align}
Then
\begin{equation}
X_{i,t}+\ell_{i,t}=D_{i,t}+E_{i,t}.
\label{eq:position-identity}
\end{equation}
Trader equity is first loss by construction. The pool receives at most $D_{i,t}$ from the account; the trader receives only $E_{i,t}$ after repayment. If $X_{i,t}<D_{i,t}$, the difference enters the pool waterfall only after the position transition recognizes the shortfall.

At entry, ignoring unfinanced fees,
\begin{equation}
X_{i,0}=C_i+D_{i,0}=N_i,
\end{equation}
so $E_{i,0}=C_i$ and $\ell_{i,0}=0$. During live trading, risk decisions use the position-level robust execution value. Realized debt reduction still requires settled cash.

\subsection{Debt-first residual ownership}
Let $y\ge0$ be settlement-confirmed cash received by a position account, let $D_i$ be recognized debt immediately before allocation, and let $F_i^{\mathrm{due}}$ be only those fees that are already earned, authorized, and senior to trader residual under the admitted position terms. Define
\begin{align}
y_i^D&=\min\{y,D_i\},\\
y_i^F&=\min\{\pospart{y-y_i^D},F_i^{\mathrm{due}}\},\\
y_i^T&=\pospart{y-y_i^D-y_i^F}.
\label{eq:residual-allocation}
\end{align}
The amount $y_i^T$ is trader residual cash. It is not a reserve contribution, provider yield, or protocol fee merely because it passed through a financed position account.

\begin{proposition}[Trader residual ownership]
\label{prop:trader-residual}
Under the debt-first rule in \eqref{eq:residual-allocation}, any settlement-confirmed value remaining after recognized debt and approved senior fees have been paid belongs to the trader position claim and cannot be credited to a reserve or provider class without a separate, pre-authorized transfer by the trader.
\end{proposition}
\begin{proof}
Debt and approved fees exhaust exactly $y_i^D+y_i^F$. The position identity in \eqref{eq:position-identity} assigns value above debt to trader equity. Crediting $y_i^T$ to another claim would increase pool claims without a corresponding contractual liability of the trader and would violate the admitted priority rule. \qedhere
\end{proof}

\subsection{Debt shares}
The pool uses a borrow index $I^B_{p,t}>0$ and position debt shares $d_i\ge0$:
\begin{equation}
D_{i,t}=d_iI^B_{p,t}.
\label{eq:debt-shares}
\end{equation}
Over an accrual interval $\Delta t$ with effective borrow rate $r^B_{p,t}$,
\begin{equation}
I^B_{p,t+1}=I^B_{p,t}\left(1+r^B_{p,t}\Delta t\right),
\label{eq:borrow-index}
\end{equation}
under the simple discrete reference implementation. A continuous implementation may use an exponential index; the accounting results are unchanged.

A new borrow $b$ mints
\begin{equation}
\Delta d=\frac{b}{I^B_{p,t}}.
\end{equation}
A settlement-confirmed repayment $y\le D_{i,t}$ burns $y/I^B_{p,t}$ debt shares, subject to a final-dust rule that sets shares to zero when full debt is repaid.

\begin{proposition}[Settlement-confirmed debt priority]
\label{prop:settled-only}
Suppose debt shares can be burned only by a transition that receives settled pool-asset cash or an on-chain asset transfer recognized by the pool mandate. Then matched, pending, reverted, or merely attested proceeds cannot reduce $D_{i,t}$.
\end{proposition}
\begin{proof}
By \eqref{eq:debt-shares}, debt changes only if $d_i$ or $I^B_{p,t}$ changes. Interest accrual changes the index weakly upward. The repayment transition is the only transition permitted to reduce $d_i$, and its precondition requires recognized settlement value. Therefore no pending state can reduce debt. A reverted settlement never satisfies the precondition. \qedhere
\end{proof}

\subsection{Partial settlement and replay resistance}
A venue match $j$ has matched capacity $M_j\ge0$. Settlement may arrive in chunks indexed by unique evidence identifiers $e\in\mathcal E_j$. Let $s_{j,e}\ge0$ be the amount recognized under evidence $e$ and
\begin{equation}
S_j=\sum_{e\in\mathcal E_j}s_{j,e}\le M_j.
\label{eq:cumulative-settlement}
\end{equation}
Each evidence identifier is consumable at most once. A partial settlement reduces debt only by the recognized amount and leaves the position in \state{PartiallySettled} or another active risk state until cumulative allocation is complete.

\begin{proposition}[Idempotent settlement allocation]
\label{prop:settlement-idempotency}
If (i) every settlement evidence identifier is unique and marked consumed atomically with allocation, and (ii) cumulative recognized settlement is capped by matched capacity, then replaying an already consumed evidence item cannot reduce debt, increase provider cash, or release trader residual a second time.
\end{proposition}
\begin{proof}
The first successful transition marks the identifier consumed in the same state transition that allocates $s_{j,e}$. A replay fails the uniqueness precondition before any journal entry. The cumulative cap prevents a distinct identifier from allocating more than the unmatched remainder. Hence each settled base unit enters the allocation function at most once. \qedhere
\end{proof}

\subsection{Senior and junior pool shares}
Let $Z^S_{p,t}$ be senior share supply. When $Z^S_{p,t}>0$, the senior share price is
\begin{equation}
\pi^S_{p,t}=\frac{V^S_{p,t}}{Z^S_{p,t}}.
\label{eq:senior-share-price}
\end{equation}
A deposit $a>0$ mints
\begin{equation}
\Delta Z^S=\frac{a}{\pi^S_{p,t}}=a\frac{Z^S_{p,t}}{V^S_{p,t}}.
\label{eq:senior-share-mint}
\end{equation}
If the pool is empty, the initial exchange rate is fixed by deployment specification. Junior shares $Z^J_{p,t}$ and price $\pi^J_{p,t}=V^J_{p,t}/Z^J_{p,t}$ are defined analogously.

\begin{proposition}[Proportional issuance does not dilute incumbents]
\label{prop:no-dilution}
If deposit $a$ is added to senior cash and senior NAV and shares are minted according to \eqref{eq:senior-share-mint}, the senior share price is unchanged by the deposit.
\end{proposition}
\begin{proof}
After deposit, NAV is $V^S+a$ and share supply is $Z^S+aZ^S/V^S=Z^S(V^S+a)/V^S$. Their ratio is $V^S/Z^S=\pi^S$. \qedhere
\end{proof}

\begin{corollary}[Loss-free share-price monotonicity]
\label{cor:share-monotonicity}
Between proportional deposits and withdrawals, if senior-retained income is non-negative and no senior loss is recognized, $\pi^S_{p,t}$ is weakly increasing.
\end{corollary}
This result is conditional on honest valuation and realized income. It does not imply that a pool-share token is riskless or immediately redeemable.

\subsection{Reserve equity is not a provider share class}
Protocol reserves are finite segregated equity accounts. They may receive protocol seed capital, designated fee allocations, slashed bonds, or governance-approved surplus, and may absorb loss according to the versioned waterfall. The reference mechanism issues no ordinary redeemable reserve shares.

\begin{proposition}[Reserve non-redemption]
\label{prop:reserve-nonredemption}
If a reserve vault exposes only contribution, authorized allocation, and loss-absorption transitions, then an LP withdrawal request cannot reduce reserve balance or create a reserve redemption claim.
\end{proposition}
\begin{proof}
No transition maps an LP share balance to a reserve claim or reserve cash outflow. Reserve changes are therefore confined to the enumerated contribution and allocation transitions. \qedhere
\end{proof}

\subsection{Integer base units, rounding, and initial-share protection}
Smart contracts operate on integers. Let $u_A$ and $u_Z$ denote one asset and one share base unit. For non-negative integers, define $\lfloor ab/c\rfloor$ and $\lceil ab/c\rceil$ as explicit multiply-divide-down and multiply-divide-up operators. The reference policy rounds provider share minting and withdrawal payments down, debt-share issuance up, and routes fee-allocation remainders to an explicit dust account. Deterministic pro-rata reserve allocation uses a largest-remainder rule with a fixed tie-break.

For a share vault with virtual offsets $v_A,v_Z>0$, the integer share mint for deposit $a$ is
\begin{equation}
\operatorname{mint}(a)=
\left\lfloor
 a\frac{Z+v_Z}{V+v_A}
\right\rfloor.
\label{eq:virtual-share-mint}
\end{equation}
A pool also sets a public minimum deposit $a_{\min}$ and may lock seed shares permanently.

\begin{proposition}[Minimum-share guard]
\label{prop:minimum-share-guard}
If $a_{\min}(Z+v_Z)\ge V+v_A$ at admission, then every admitted deposit $a\ge a_{\min}$ mints at least one share base unit under \eqref{eq:virtual-share-mint}. Unsolicited asset donations may change the exchange rate but cannot make an admitted deposit mint zero shares while the inequality is enforced.
\end{proposition}
\begin{proof}
The inequality implies $a(Z+v_Z)/(V+v_A)\ge1$. Its floor is therefore at least one. \qedhere
\end{proof}

\begin{theorem}[Integer accounting closure with explicit dust]
\label{thm:integer-closure}
Suppose every integer transition is the rounded form of a balanced real-valued journal, each rounding residual is posted to a recognized dust account, and a transition contains at most $n_j$ downward recipient allocations in asset units. Then accounting closure is exact when the dust account is included. If dust is omitted from a displayed summary after $T$ transitions, the absolute displayed mismatch is bounded by
\begin{equation}
0\le \Delta_T < u_A\sum_{j=1}^{T}n_j
\label{eq:dust-bound}
\end{equation}
plus any separately bounded debt-share conversion residual specified by the implementation.
\end{theorem}
\begin{proof}
Each downward allocation leaves a residual strictly smaller than one asset base unit. Posting that residual to dust preserves the balanced journal exactly. Omitting dust can hide at most the sum of the per-allocation residual bounds; summing over transitions gives \eqref{eq:dust-bound}. \qedhere
\end{proof}

\subsection{Balanced transitions}
The reference accounting uses double-entry journals.
\begin{table}[H]
\centering
\small
\caption{Balanced local-pool transitions. A plus sign increases the account; a minus sign decreases it.}
\label{tab:balanced-transitions}
\begin{tabularx}{\textwidth}{lXX}
\toprule
Transition & Asset-side journal & Claim-side journal\\
\midrule
Senior deposit $a$ & $A_p:+a$ & $V^S_p:+a$\\
Junior deposit $j$ & $A_p:+j$, junior-tagged & $V^J_p:+j$\\
Reserve contribution $r$ & $A_p:+r$, reserve-tagged & $R_p:+r$\\
Borrow $b$ & $A_p:-b$, $B_p:+b$ & none\\
Interest accrual $h$ & $B_p:+h$ & allocate $h$ across senior, junior, reserve, protocol, and MM claims\\
Repayment $y$ & $A_p:+y$, $B_p:-y$ & none\\
Recognized loss $\ell$ & $B_p:-\ell$ & reduce claims by waterfall allocation\\
Withdrawal payment $w$ & $A_p:-w$ & $V^S_p:-w$ or $V^J_p:-w$\\
Fee transfer out $f$ & $A_p:-f$ & $T_p:-f$\\
\bottomrule
\end{tabularx}
\end{table}

Junior deposits, reserve balances, and protocol-fee balances are held in separate vaults or tagged sub-accounts even though they are aggregated in $A_p$ for balance-sheet closure. They are excluded from $A^L_p$. Junior capital is not lendable as senior principal unless the mandate changes prospectively; otherwise the same capital would be counted as both junior provider protection and loanable cash.

\begin{theorem}[Accounting closure]
\label{thm:accounting-closure}
If \eqref{eq:pool-closure} holds at time $t$ and every protocol transition applies a balanced journal in \Cref{tab:balanced-transitions}, then \eqref{eq:pool-closure} holds after any finite sequence of transitions, excluding external asset loss not yet recognized.
\end{theorem}
\begin{proof}
Each transition changes $\mathcal A_p$ and $\mathcal C_p$ by the same amount. Borrow and repayment reclassify assets and have zero net effect on both sides. Claim reduction in a recognized loss equals asset reduction by the conservation property of the waterfall proved in \Cref{thm:waterfall}. Induction over a finite sequence gives the result. Unrecognized external loss lies outside the hypothesis because the book asset has not yet been written down. \qedhere
\end{proof}

\subsection{Pending settlement}
Pending fills are recorded in a memorandum or suspense account. They may affect operational exposure and retry logic but not pool NAV, debt reduction, or withdrawable cash. When settlement confirms, the suspense entry is replaced by settled cash and repayment is applied. When settlement fails, the suspense entry is reversed. This separation prevents optimistic accounting from creating artificial liquidity.

\section{Interest, Utilization, and Participant Returns}
\label{sec:interest}

\subsection{Utilization}
Let lendable senior resources be
\begin{equation}
S^{\mathrm{lend}}_{p,t}=A^L_{p,t}+B_{p,t},
\end{equation}
and define utilization
\begin{equation}
U_{p,t}=\frac{B_{p,t}}{A^L_{p,t}+B_{p,t}}\in[0,1],
\label{eq:utilization}
\end{equation}
when the denominator is positive. Aggregate cash $A_p$ can exceed $A^L_p$ because reserves, LBP capital, protocol claims, withdrawal encumbrances, and required buffers are excluded. If any such balance becomes lendable, it ceases to be simultaneously countable as protection or immediately serviceable cash.

The reference borrow-rate curve is
\begin{equation}
r^B(U)=
\begin{cases}
r_0+s_1\dfrac{U}{U^\star}, & U\le U^\star,\\[7pt]
r_0+s_1+s_2\dfrac{U-U^\star}{1-U^\star}, & U>U^\star,
\end{cases}
\label{eq:rate-curve}
\end{equation}
with $r_0,s_1,s_2\ge0$ and kink $U^\star\in(0,1)$. The curve is a transparent policy benchmark, not an empirical estimate of an optimal Axient rate.

\begin{figure}[H]
\centering
\includegraphics[width=.80\textwidth]{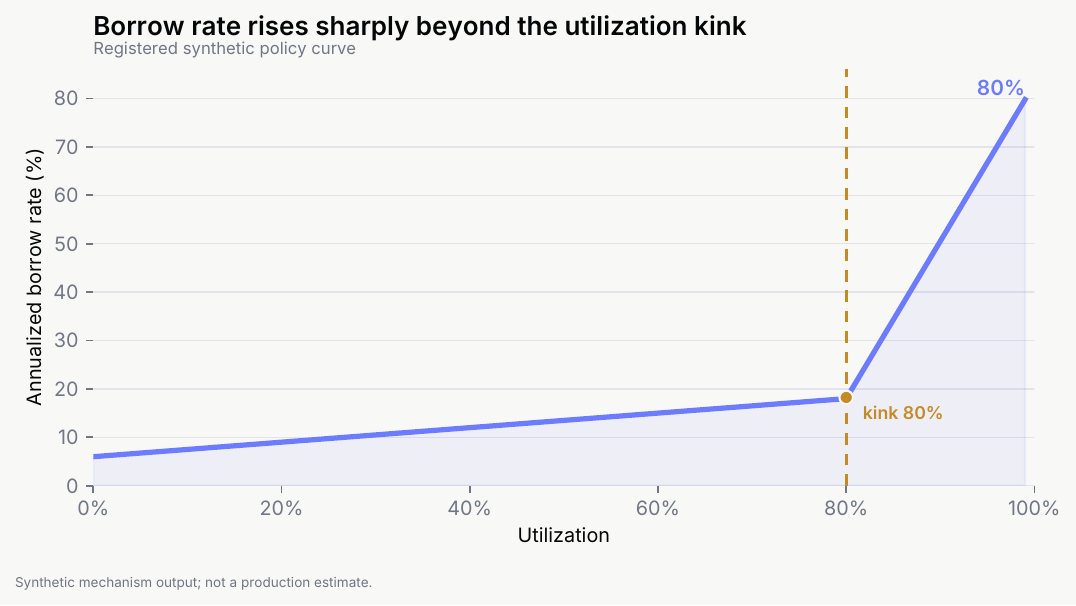}
\caption{Registered synthetic utilization-based borrow-rate curve. Parameters are policy inputs for mechanism comparison, not a live rate recommendation.}
\label{fig:rate-curve}
\end{figure}

\begin{proposition}[Monotonicity of the reference rate]
\label{prop:rate-monotone}
If $s_1,s_2\ge0$, $r^B(U)$ in \eqref{eq:rate-curve} is continuous and weakly increasing on $[0,1]$, and strictly increasing if $s_1,s_2>0$.
\end{proposition}

\subsection{Conditional utilization equilibrium}
Let $b(r)$ denote aggregate desired borrowing at rate $r$. The protocol does not assume it knows this demand curve in production, but a conditional result clarifies the role of the utilization curve.

\begin{theorem}[Unique utilization equilibrium under monotone demand]
\label{thm:utilization-equilibrium}
Suppose lendable supply $S>0$ is fixed, $b:[r^B(0),r^B(1)]\to[0,S]$ is continuous and strictly decreasing, and $b(r^B(0))>0$ while $b(r^B(1))<S$. Then a unique $U^\dagger\in(0,1)$ satisfies
\begin{equation}
U^\dagger=\frac{b(r^B(U^\dagger))}{S}.
\label{eq:utilization-fixed-point}
\end{equation}
\end{theorem}
\begin{proof}
Define $f(U)=U-b(r^B(U))/S$. It is continuous. Endpoint assumptions give $f(0)<0$ and $f(1)>0$, so a root exists. Since $r^B$ is weakly increasing and $b$ strictly decreasing, $g(U)=b(r^B(U))/S$ is weakly decreasing. For $U_2>U_1$,
\[
f(U_2)-f(U_1)=(U_2-U_1)-(g(U_2)-g(U_1))\ge U_2-U_1>0.
\]
Thus $f$ is strictly increasing and the root is unique. \qedhere
\end{proof}

The theorem does not validate rate parameters. It states only that a monotone curve and monotone static demand do not generate multiple utilization equilibria under the assumptions. Dynamic stability, strategic borrowing, and rate manipulation remain empirical questions \citep{bastankhah2024agilerate,bastankhah2024fastslow}.

\subsection{Interest allocation}
Let gross accrued interest over an interval be $H_{p,t}=H^{\mathrm{acc}}_{p,t}\ge0$. The protocol allocates it through non-negative weights
\begin{equation}
\alpha_S+\alpha_J+\alpha_R+\alpha_P+\alpha_M=1,
\label{eq:interest-weights}
\end{equation}
for senior LPs, LBPs, reserves, protocol treasury, and contracted market-maker premiums. A configuration without a role sets its weight to zero.

The claim increments are
\begin{align}
\Delta V^S &= \alpha_SH, &
\Delta V^J &= \alpha_JH,\\
\Delta R &= \alpha_RH, &
\Delta T &= \alpha_PH, &
\Delta M &= \alpha_MH.
\end{align}
The MM term is payable only when a commitment contract or service agreement has earned the premium. Ordinary spread revenue earned directly on an external venue is not protocol interest.

\subsection{Accrued receivables are not withdrawable cash}
Let $H^{\mathrm{acc}}_{p,t}$ denote borrower interest accrued into recognized debt receivables and let $H^{\mathrm{cash}}_{p,t}$ denote interest actually received in settlement cash. Accrual may increase a claim or risk-adjusted NAV according to the valuation policy, but it does not increase $A^W_{p,t}$. Provider, reserve, MM, and protocol cash distributions are funded only from $H^{\mathrm{cash}}$ or other settled revenue.

\begin{proposition}[Accrual--cash separation]
\label{prop:accrual-cash-separation}
If withdrawal-eligible cash excludes unpaid receivables, then increasing $H^{\mathrm{acc}}$ without settlement cannot increase immediate provider withdrawal capacity. When settlement arrives, only the settled component can be routed as cash, with any unpaid remainder remaining a receivable.
\end{proposition}
\begin{proof}
By definition $A^W$ contains settled unencumbered cash and excludes debt receivables. Accrual changes the receivable and claim ledgers but not cash. Therefore the immediate cash bound is unchanged until a settlement transition reclassifies receivable value into cash. \qedhere
\end{proof}

\begin{proposition}[Fee-routing conservation]
\label{prop:fee-conservation}
If \eqref{eq:interest-weights} holds, total allocated claim accrual equals gross accrued interest; realized cash distributions remain bounded by settled interest cash.
\end{proposition}

\subsection{Senior and LBP returns}
Over horizon $[0,T]$, define realized senior net return
\begin{equation}
R^S_T=\frac{I^S_T+F^S_T-L^S_T}{V^S_0},
\label{eq:senior-return}
\end{equation}
where $I^S_T$ and $F^S_T$ are realized interest and fees allocated to senior NAV and $L^S_T$ is recognized senior principal loss. LBP return is
\begin{equation}
R^J_T=\frac{I^J_T+F^J_T-L^J_T}{V^J_0}.
\label{eq:lbp-return}
\end{equation}
LBP capital is junior, so $R^J_T$ should generally have a wider loss distribution and can be negative even when senior return remains positive.

\begin{definition}[LBP participation threshold]
Let $k_J$ be the LBP's required expected return over a chosen horizon. A risk-neutral participation condition is
\begin{equation}
\E[I^J_T+F^J_T-L^J_T]\ge k_JV^J_0.
\label{eq:lbp-participation}
\end{equation}
\end{definition}
The expression is not a protocol promise. Risk aversion, capital lock-up, model uncertainty, operational exposure, and legal cost raise the required premium.

\subsection{Protocol contribution margin}
For the protocol treasury, reserve accrual and provider income are not revenue. A reference contribution margin is
\begin{equation}
\mathrm{CM}_T=
F^{\mathrm{orig}}_T+F^{\mathrm{exec}}_T+\alpha_P I_T+R^{\mathrm{partner}}_T
-C^{\mathrm{ops}}_T-C^{\mathrm{security}}_T-C^{\mathrm{incentive}}_T-L^{\mathrm{protocol}}_T.
\label{eq:contribution-margin}
\end{equation}
This separates protocol economics from LP economics. Senior providers receive contractual pool returns, not an equity-like share of Axient corporate profit merely by supplying credit.

\subsection{No free protection}
Allocating more income to reserves, LBPs, and bonded MMs reduces income retained by senior LPs or the protocol. The protection stack has an explicit price. A comparison should report both senior loss and senior carry; minimizing loss incidence by transferring nearly all interest to junior capital would not constitute a meaningful economic improvement.

\section{Liquidation Backstops and Liquidators}
\label{sec:liquidation-roles}

\subsection{Role separation}
A liquidation event has two economically different inputs:
\begin{enumerate}[leftmargin=*]
\item an execution action that converts or reduces the position; and
\item loss-absorbing capital if realized proceeds do not cover debt.
\end{enumerate}
The first is supplied by a liquidator. The second may be supplied by reserves, an LBP, or senior capital. Combining both under the word ``liquidator'' obscures who is paid for latency and who is paid for bearing principal risk.

\begin{definition}[Admissible liquidation]
A liquidation action $a$ is admissible for position $i$ at time $t$ if it is authorized by the risk state, weakly reduces recognized debt or risk exposure, satisfies venue and price bounds, respects aggregate execution capacity, and cannot transfer trader residual value before debt priority is satisfied.
\end{definition}

At the Axient contract boundary, liquidation submission may be permissionless: any address can submit an admissible action and receive a bounty. Whether the action can be executed on an external venue is determined by adapter capabilities, signer authority, and venue access. The contract, not executor identity, determines validity of the Axient-side transition.

\begin{definition}[LBP vault]
An LBP vault is a junior-capital account with balance $J_{p,t}$, locked withdrawal rules, income share $\alpha_J$, and loss-allocation priority ahead of senior principal. LBP shares represent a residual claim after recognized loss.
\end{definition}

\subsection{Execution bounty}
Let $g(a)$ denote verified gas or execution cost and $\beta(a)$ protocol bounty. The liquidator receives
\begin{equation}
F^{\mathrm{liq}}(a)=g(a)+\beta(a),
\label{eq:liquidator-fee}
\end{equation}
subject to a cap and successful settlement. A bounty may depend on size, urgency, and execution quality, but it must not reward a transition that increases expected shortfall.

An execution-only liquidator has no claim on LBP income unless it also owns LBP shares. Conversely, an LBP receives junior-capital income even when another keeper performs the liquidation.

\begin{theorem}[Role separability]
\label{thm:role-separability}
Suppose (i) liquidation eligibility and transition validity are enforced by contract, (ii) liquidator compensation depends only on a valid settled action, and (iii) LBP loss and income are allocated by junior share ownership. Then liquidation execution can be permissionless independently of the ownership concentration of LBP capital.
\end{theorem}
\begin{proof}
Condition (i) makes admissibility independent of the executor's capital position. Condition (ii) pays for the action rather than LBP ownership. Condition (iii) assigns capital gains and losses by share balance rather than executor identity. Therefore any qualified address may execute while LBP ownership remains arbitrary. \qedhere
\end{proof}

\subsection{Junior capital and lock-up}
LBP capital is not immediately withdrawable while it supports active debt. Let $J^{\mathrm{enc}}_{p,t}$ be the amount encumbered by pool protection policy. Freely withdrawable capital satisfies
\begin{equation}
J^{\mathrm{free}}_{p,t}=\pospart{J_{p,t}-J^{\mathrm{enc}}_{p,t}-b^J_p},
\end{equation}
where $b^J_p$ is a minimum retained buffer. A request above $J^{\mathrm{free}}$ enters a junior queue or is rejected by mandate. Without this rule, junior capital could leave immediately before a known hard-flat window and would not be credible protection.

\subsection{LBP loss boundary}
LBP loss is limited to vault balance committed to the pool or risk bucket. The protocol does not create uncapped recourse against the provider. If a provider also posts a market-maker bond, the bond is a separate account and may be slashed under \Cref{sec:mm}.

\begin{remark}[Senior loss remains possible]
The existence of an LBP does not insure senior capital. If recognized shortfall exceeds position value, reserves, LBP capital, and other protective layers, the remainder impairs senior net asset value. The protocol must expose this boundary.
\end{remark}

\subsection{Liquidation capital versus liquidation liquidity}
Junior capital is useful even when an external venue provides execution. Execution liquidity is useful even when no junior capital exists. The protocol therefore tracks two distinct constraints:
\begin{align}
\text{executable capacity} &\ge \text{required sale quantity},\\
\text{loss-absorbing capacity} &\ge \text{recognized residual shortfall}.
\end{align}
A market maker can improve the first. An LBP improves the second. A bonded MM commitment can affect both only to the extent that its bond is transferable stablecoin capital.

\section{Market Makers and Collateralized Commitments}
\label{sec:mm}

\subsection{Open quotes are not guarantees}
A visible bid contributes to a position-level book curve, but it is revocable until matched and does not by itself reduce reserve requirements. Counting the same quote as both ordinary executable depth and guaranteed backstop capacity would double-count liquidity.

\begin{definition}[Collateralized executable commitment]
A market-maker commitment is
\begin{equation}
\kappa_m=(e,v,[t_0,t_1],Q_m,\underline P_m,\Gamma_m,b_m),
\label{eq:mm-commitment}
\end{equation}
where $Q_m$ is maximum quantity, $\underline P_m$ a minimum net price or proceeds curve, $\Gamma_m$ admissible execution conditions, and $b_m$ a stablecoin bond transferable to the protocol on defined underperformance.
\end{definition}

The contract may pay a commitment premium when no trade occurs and an execution fee when capacity is used. The premium compensates capital lock-up and adverse-selection exposure; the execution fee compensates the trade.

\subsection{Credited capacity and partial performance}
Let $c_m(x)$ be incremental settled recovery credited to commitment $m$ for sale quantity $x$, relative to the uncommitted robust book envelope. The credited amount must satisfy
\begin{equation}
0\le c_m(x)\le b_m
\label{eq:mm-bond-cap}
\end{equation}
for every admissible $x$, unless additional on-chain collateral or atomic delivery guarantees cover the excess. If realized incremental delivery is $y_m(x)\in[0,c_m(x)]$, failure transfer $s_m(x,y)$ must satisfy
\begin{equation}
s_m(x,y)\ge \pospart{c_m(x)-y_m(x)}.
\label{eq:mm-shortfall-slash}
\end{equation}

\begin{theorem}[Bonded commitment reserve substitution]
\label{thm:mm-substitution}
If realized incremental delivery is $y_m(x)\ge0$ and enforceable failure transfer satisfies \eqref{eq:mm-shortfall-slash}, then replacing reserve capacity $c_m(x)$ with the commitment does not reduce the pool's lower realized recovery bound for the committed quantity.
\end{theorem}
\begin{proof}
Total incremental recovery is $y_m(x)+s_m(x,y)$. By \eqref{eq:mm-shortfall-slash},
\[
y_m(x)+s_m(x,y)\ge y_m(x)+\pospart{c_m(x)-y_m(x)}\ge c_m(x).
\]
The result covers full performance, partial performance with slashing, and complete failure. The source of recovery changes, but its lower bound does not. \qedhere
\end{proof}

The theorem is narrow. It assumes enforceable value transfer and does not treat litigation, reputation, or an off-chain letter of intent as stablecoin collateral.

\subsection{Commitment allocation and double counting}
A commitment has finite capacity and cannot be promised independently to several positions. For positions $I_m$ using commitment $m$,
\begin{equation}
\sum_{i\in I_m}x_i\le Q_m
\end{equation}
must hold under the aggregate hard-flat schedule. Allocation can be pro rata, earliest-deadline first, or lowest-coverage first, but it must be deterministic and encoded in the registry.

Ordinary venue depth and bonded capacity also cannot overlap. If the committed order is already displayed in a book used to construct the uncommitted proceeds envelope, the credited increment is only the enforceable value beyond what is already counted.

\subsection{Correlated roles}
A market maker may also supply senior credit or LBP capital. The protocol records roles separately because common ownership creates correlated failure. If one institution is simultaneously the largest senior lender, sole exit provider, and principal LBP, its failure removes funding, executable capacity, and junior loss-absorbing provider capital at once. Pool mandates therefore cap concentration by linked entity or address cluster where observable.

\subsection{Market-maker economics}
A market maker may earn ordinary spread, venue rebates, commitment premium, execution fees, liquidation bounties, and independent senior or LBP returns. These cash flows should not be collapsed into an opaque ``liquidity incentive.'' Separate accounting permits the protocol to test whether a premium compensates enforceable capacity or merely subsidizes ordinary quoting.

\subsection{Venue-agnostic implementation}
A commitment can terminate in a direct on-chain fill, RFQ settlement, auction, or verified venue order. The core requirement is not a particular API but an objective success condition and enforceable failure transfer. Where neither exists, the commitment cannot enter the robust proceeds certificate and remains an off-chain commercial expectation.

\section{Deterministic Loss Waterfall}
\label{sec:waterfall}

\subsection{Recognized shortfall}
A pool loss is recognized only after the position account has completed the applicable execution, settlement, and reconciliation transitions. Recognition requires a unique evidence identifier, the active adapter version, and a finalized loss state. Replaying the same evidence or presenting conflicting evidence cannot allocate the loss twice; an unresolved conflict moves the position or pool to \state{EvidenceConflict} or \state{Recovery}. For position $i$, let
\begin{equation}
\ell_i=\pospart{D_i-K_i-P^{\mathrm{set}}_i},
\label{eq:recognized-shortfall}
\end{equation}
where $K_i$ is settled position cash not already included in $P_i^{\mathrm{set}}$. Trader collateral and realized position value are consumed inside the position account before \eqref{eq:recognized-shortfall} enters the pool protection stack. Pending fills and expected recoveries are excluded.

The reference waterfall contains the following ordered capacities:
\begin{align*}
H_{i,1}&=\text{position-specific protocol buffer},\\
H_{i,2}&=\text{market/event/venue reserve assigned to }i,\\
H_{i,3}&=\text{LBP capital assigned to the risk bucket},\\
H_{i,4}&=\text{pool reserve},\\
H_{i,5}&=\text{optional allocated global reserve},\\
H_{i,6}&=\text{senior principal available for impairment}.
\end{align*}
A pool mandate may omit a layer by setting its capacity to zero. It may not reorder the layers for an already admitted position.

\begin{definition}[Ordered waterfall operator]
\label{def:waterfall}
For a recognized shortfall $\ell\ge0$ and an ordered vector $H=(H_1,\ldots,H_K)$ with $H_k\ge0$, define the residuals and allocations recursively by
\begin{align}
\rho_0&=\ell,\nonumber\\
a_k&=\min\{H_k,\rho_{k-1}\},\label{eq:waterfall-allocation}\\
\rho_k&=\rho_{k-1}-a_k,\qquad k=1,\ldots,K.\label{eq:waterfall-residual}
\end{align}
The final uncovered loss is $\rho_K$.
\end{definition}

\begin{theorem}[Existence, uniqueness, and conservation]
\label{thm:waterfall}
For every finite $H\in\R_+^K$ and $\ell\ge0$, the waterfall in \Cref{def:waterfall} exists and is unique. It satisfies
\begin{equation}
0\le a_k\le H_k,\qquad
\rho_0\ge \rho_1\ge\ldots\ge \rho_K\ge0,
\label{eq:waterfall-bounds}
\end{equation}
and
\begin{equation}
\sum_{k=1}^{K}a_k+\rho_K=\ell.
\label{eq:waterfall-conservation}
\end{equation}
\end{theorem}
\begin{proof}
The recursion defines $a_1$ and $\rho_1$ uniquely. If $\rho_{k-1}$ is defined, then $a_k=\min\{H_k,\rho_{k-1}\}$ and $\rho_k=\rho_{k-1}-a_k$ are unique. Induction gives existence and uniqueness for every finite $K$. Non-negativity and the capacity bound follow from the minimum operator. Since $a_k\ge0$, residuals are weakly decreasing. Summing $\rho_{k-1}-\rho_k=a_k$ telescopes to \eqref{eq:waterfall-conservation}. \qedhere
\end{proof}

\begin{corollary}[Junior-before-senior impairment]
\label{cor:junior-before-senior}
Suppose LBP capital is layer $j$ and senior principal is a later layer $s>j$ with $H_s>0$. If $a_s>0$, then every preceding layer, including LBP capital, is exhausted: $a_k=H_k$ for all $k<s$.
\end{corollary}
\begin{proof}
Positive senior allocation implies $\rho_{s-1}>0$. If any preceding layer $k<s$ were not exhausted, then $a_k=\rho_{k-1}<H_k$ and \eqref{eq:waterfall-residual} would give $\rho_k=0$, hence every later allocation would be zero, a contradiction. \qedhere
\end{proof}

\begin{remark}[Senior is protected, not guaranteed]
\Cref{cor:junior-before-senior} establishes priority, not solvency. Senior impairment occurs whenever recognized shortfall exceeds the aggregate capacity of all earlier layers. The protocol must expose that residual rather than describe the pool as insured.
\end{remark}

\subsection{Multiple positions and one shared book}
The position-level execution policy of \citet{nechepurenko2026axient} ranks emergency hard-flat orders by lowest robust coverage ratio
\begin{equation}
\Gamma_i(u)=
\frac{K_i+\underline B^{\mathrm{standalone}}_{i,u,\Delta}(q_i)}
{\overline H_{i,u,\Delta}+m_i},
\label{eq:robust-coverage-ratio}
\end{equation}
then by earlier hard-flat deadline and finally by deterministic position identifier. The ratio is an execution-priority statistic; it is not additive and does not replace the aggregate shared-book feasibility test. Once common executable depth has been consumed in that deterministic order, the resulting position shortfalls enter the capital waterfall.

If a shared reserve $R$ must be allocated across residual shortfalls $\ell_1,\ldots,\ell_n$, the reference implementation uses proportional allocation:
\begin{equation}
g_i=
\begin{cases}
\ell_i,& \sum_j\ell_j\le R,\\[4pt]
R\dfrac{\ell_i}{\sum_j\ell_j},& \sum_j\ell_j>R.
\end{cases}
\label{eq:pro-rata-reserve}
\end{equation}
This allocation is deterministic, neutral to transaction-order races, and limited to lender-principal shortfall after all position-specific recovery.

In integer base units, the reference implementation first floors each proportional share, then assigns the remaining base units by descending fractional remainder with a deterministic identifier tie-break. This largest-remainder rule preserves the same economic proportions while conserving the finite reserve exactly.

\begin{proposition}[Exact integer shared-reserve allocation]
\label{prop:integer-pro-rata}
For non-negative integer shortfalls $\ell_i$ and integer reserve capacity $R$, the deterministic largest-remainder implementation allocates integers $g_i$ satisfying
\begin{equation}
0\le g_i\le\ell_i,
\qquad
\sum_i g_i=\min\left\{R,\sum_i\ell_i\right\}.
\label{eq:integer-pro-rata}
\end{equation}
The allocation is unique once the tie-break order is fixed.
\end{proposition}
\begin{proof}
Floor allocations cannot exceed their proportional targets or claimant caps. The number of unallocated base units is smaller than the number of positive claimants. Assigning one unit in deterministic remainder order exhausts the target without exceeding any positive residual claim. Fixed ordering makes the result unique. \qedhere
\end{proof}

\begin{proposition}[Proportional reserve properties]
\label{prop:pro-rata-reserve}
The allocation in \eqref{eq:pro-rata-reserve} satisfies
\begin{equation}
0\le g_i\le\ell_i,
\qquad
\sum_i g_i=\min\left\{R,\sum_i\ell_i\right\}.
\label{eq:pro-rata-properties}
\end{equation}
If $\ell_i>0$ and the reserve is insufficient, every positive shortfall receives the same fractional coverage $g_i/\ell_i=R/\sum_j\ell_j$.
\end{proposition}

\subsection{Reference waterfall}
\Cref{fig:waterfall} separates pre-waterfall position recovery from protocol capital. The senior layer is deliberately visible as the final absorber rather than omitted from the diagram.

\begin{figure}[H]
\centering
\begin{tikzpicture}[
  node distance=0.34cm,
  box/.style={draw,rounded corners,minimum height=0.75cm,text width=2.35cm,align=center,font=\small},
  arr/.style={-{Latex[length=2mm]},thick}
]
\node[box] (pos) {Trader equity and settled position proceeds};
\node[box,right=of pos] (buf) {Position-specific buffer};
\node[box,right=of buf] (mres) {Market / venue reserve};
\node[box,right=of mres] (lbp) {LBP provider capital};
\node[box,below=0.72cm of lbp] (pres) {Pool reserve};
\node[box,left=of pres] (gres) {Optional global reserve};
\node[box,left=of gres] (sen) {Senior principal impairment};
\draw[arr] (pos) -- (buf);
\draw[arr] (buf) -- (mres);
\draw[arr] (mres) -- (lbp);
\draw[arr] (lbp) -- (pres);
\draw[arr] (pres) -- (gres);
\draw[arr] (gres) -- (sen);
\end{tikzpicture}
\caption{Canonical Axient loss path. Capital reaches the senior layer only after all earlier recognized protection has been exhausted.}
\label{fig:waterfall}
\end{figure}
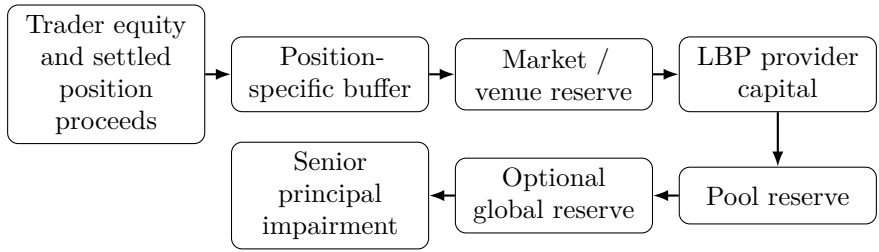

Every allocation emits the source layer, position or bucket identifier, amount, share-price effect, and remaining residual. A loss can be challenged at the data-reconciliation layer, but once finalized it cannot be moved to a different layer by an off-chain operator.

\section{Pool Isolation and Contagion}
\label{sec:isolation}

\subsection{Strict isolation}
\begin{definition}[Strict pool isolation]
\label{def:strict-isolation}
Pools $p$ and $q$ are strictly isolated if they share no lendable cash, reserve, LBP tranche, collateral account, market-maker bond, cross-guarantee, debt receivable, or withdrawal queue, and if neither pool's transition system has write access to the other's claims or assets.
\end{definition}

\begin{theorem}[No direct local loss contagion]
\label{thm:no-local-contagion}
Under \Cref{def:strict-isolation}, a recognized shortfall in pool $p$ cannot reduce the recognized net asset value of pool $q\ne p$.
\end{theorem}
\begin{proof}
Every asset and protection layer eligible for the shortfall belongs to $p$. By isolation, no transition induced by the shortfall can debit an asset or claim account of $q$. Hence the book assets and claims used to compute $q$'s NAV are unchanged. Common external shocks may independently reduce both pools, but a write-down in $p$ does not cause a ledger write-down in $q$. \qedhere
\end{proof}

Strict isolation is therefore an accounting property, not a statistical independence claim. Two isolated pools can still fail simultaneously because they share a venue, stablecoin, oracle, chain, market maker, or event class.

\subsection{Shared protection and protection contagion}
Let a protocol-global reserve have balance $G_t$. If pool $p$ draws $g_p\le G_t$, the remaining balance is
\begin{equation}
G_{t^+}=G_t-g_p.
\label{eq:global-reserve-after-draw}
\end{equation}
The draw does not transfer $p$'s liability to another pool. It does reduce the maximum future protection available to every pool whose mandate admits the same reserve.

\begin{proposition}[Bound on protection contagion]
\label{prop:protection-contagion}
For any other pool $q$, a draw $g_p$ from a shared reserve can reduce $q$'s reserve-backed protection capacity by at most $g_p$, and by exactly $g_p$ whenever $q$ previously had access to the full remaining global reserve and no pool-specific allocation floor applies.
\end{proposition}
\begin{proof}
Before the draw, the most $q$ can obtain from the shared layer is $G_t$ or a smaller mandate cap. After the draw, the layer is $G_t-g_p$. The difference cannot exceed $g_p$. Under full shared access and no smaller cap, the difference is exactly $g_p$. \qedhere
\end{proof}

This channel is called \emph{protection contagion}. It is economically important even though \Cref{thm:no-local-contagion} still holds: pool $q$ has not absorbed $p$'s present loss, but it is less protected against its own future loss.

\subsection{Capability isolation}
A pool mandate includes a minimum venue-capability class and permitted backing modes. Let $\chi_v$ be the venue vector in \eqref{eq:capability-vector}. A position backed through a weaker adapter cannot be silently inserted into a stronger pool by assigning it an optimistic valuation.

\begin{definition}[Capability-homogeneous pool]
A pool is capability-homogeneous if every admitted position satisfies the same minimum custody, lien, liquidation, settlement, redemption, and withdrawal-lock requirements, or is subject to an explicit weaker-mode sublimit and separately reported capital charge.
\end{definition}

Capability isolation prevents a custodial or attested venue position from diluting a pool marketed and risk-priced for contract-controlled collateral. It does not require every venue to have identical APIs; it requires every adapter to make its trust boundary machine-readable.

\subsection{Concentration and service contagion}
Even with separate contracts, linked providers can create service contagion. If one institution is simultaneously a senior LP, LBP, primary market maker, signer operator, and venue counterparty, its failure removes several protective functions at once. Pool mandates therefore apply linked-entity concentration caps to:
\begin{itemize}[leftmargin=*]
\item senior supplied capital;
\item LBP capital;
\item bonded execution capacity;
\item venue and settlement exposure;
\item signer and keeper dependencies.
\end{itemize}
Where beneficial ownership is unobservable, caps can only operate at address or attestation scope; this limitation must be disclosed rather than assumed away.

\section{Reserve Sufficiency and Open-Interest Capacity}
\label{sec:reserve}

\subsection{Scenario-conditional sufficiency}
Let $\Omega_p^\star$ be a finite registered scenario family for pool $p$. For scenario $\omega$, let $L_p(\omega)$ be aggregate recognized shortfall after trader equity and settled position proceeds, and let $M_p(\omega)$ be enforceable incremental recovery from bonded market-maker commitments or other collateralized adapters. Define
\begin{equation}
K_p^\star=
\max_{\omega\in\Omega_p^\star}
\pospart{L_p(\omega)-M_p(\omega)}.
\label{eq:registered-required-capital}
\end{equation}
Let $R^M_p$, $J_p$, $R^P_p$, and $G_p^{\mathrm{alloc}}$ denote market reserves, LBP capital, pool reserve, and the portion of global reserve contractually allocable to $p$.

\begin{theorem}[Finite-scenario senior protection]
\label{thm:finite-scenario-protection}
If
\begin{equation}
R^M_p+J_p+R^P_p+G_p^{\mathrm{alloc}}
\ge K_p^\star,
\label{eq:finite-scenario-capital-condition}
\end{equation}
then senior principal is not impaired under any $\omega\in\Omega_p^\star$, provided the scenario's settlement and enforceability assumptions hold.
\end{theorem}
\begin{proof}
For each registered scenario, residual shortfall after enforceable recovery is at most $K_p^\star$. By \eqref{eq:finite-scenario-capital-condition}, the aggregate capacity of layers preceding senior principal is at least that amount. \Cref{thm:waterfall} therefore allocates the entire shortfall before reaching senior principal. \qedhere
\end{proof}

The result is conditional on the registered scenario family. It is not a proof of solvency on paths outside $\Omega_p^\star$, and it does not transform an off-chain promise into enforceable recovery.

\subsection{Reserve-supported open interest}
Suppose a risk bucket has gross open interest $G$ and registered worst-case residual shortfall rate $\lambda^\star(G)$ after ordinary position recovery and bonded commitments. If $\lambda^\star$ is treated conservatively as non-decreasing in size, a necessary admission condition is
\begin{equation}
G\lambda^\star(G)
\le R^M+J+R^P+G^{\mathrm{alloc}}.
\label{eq:reserve-supported-oi}
\end{equation}
Under a constant registered rate $\bar\lambda>0$, the corresponding cap is
\begin{equation}
G_{\max}^{\mathrm{protect}}
=\frac{R^M+J+R^P+G^{\mathrm{alloc}}}{\bar\lambda}.
\label{eq:reserve-supported-oi-linear}
\end{equation}
The constant-rate expression is a planning approximation. Production admission uses size-dependent shared-book execution and scenario transforms rather than an unconditional percentage.

\subsection{Funding the protection layers}
Reserve capital may be formed from:
\begin{itemize}[leftmargin=*]
\item protocol seed capital;
\item the reserve share $\alpha_R$ of borrower interest;
\item origination or liquidation fees;
\item slashed MM bonds;
\item explicit partner contributions;
\item realized protocol surplus transferred by governance under a timelock.
\end{itemize}
LBP capital is economically different. It has an ownership claim and receives designated income in exchange for standing ahead of senior principal. Protocol reserves are not ordinary yield-bearing deposits, issue no ordinary redemption shares in the reference mechanism, and must not be simultaneously counted as lendable cash.

\subsection{Replenishment and admission response}
A reserve draw changes future admission immediately. If current protected-capital coverage falls below the mandate floor, the risk manager must reduce new-credit capacity, lower leverage tiers, require additional MM or LBP support, or place the bucket in \state{ReduceOnly}. The system may not continue admitting exposure on the assumption that future fees will replenish an already depleted reserve.

\begin{remark}[Why the word insurance is avoided]
The reserve is finite, scenario-limited, and subject to governance and asset risks. The protocol therefore uses \emph{reserve} and \emph{backstop capital}, not an unconditional insurance claim. Senior loss remains possible after exhaustion of the waterfall.
\end{remark}

\subsection{Admission before reserve}
A reserve is a last protective layer, not a substitute for executable liquidity. If ordinary hard-flat is structurally infeasible for a market or venue mode, admission fails. Otherwise the reserve would be used for routine financing losses rather than for departures from the registered operating region, making the product an underpriced unsecured credit exposure.

\section{Withdrawal Queues and Liquidity Transformation}
\label{sec:withdrawals}

Senior pool shares are marked against pool NAV, but financed positions can remain locked until repayment, liquidation, venue settlement, or redemption. A promise of instantaneous redemption at book value would therefore create an unbacked maturity transformation.

\subsection{Share-denominated requests}
An LP submits a request for $z$ senior shares. The reference mechanism escrows the shares without burning them. They remain part of total share supply and continue to participate in income, recognized loss, and NAV haircuts until actual payment. An implementation that burns shares at request time must mint an economically identical queue-claim unit that preserves the same participation; otherwise early requesters can escape later losses. Let total queued shares at time $0$ be
\begin{equation}
Z_Q=\sum_{j=1}^{n}z_j.
\end{equation}
Let $C_T$ be cumulative settled cash made available for queue service by horizon $T$, after mandatory operational and protection buffers. A service policy may be FIFO, epoch pro rata, or another deterministic rule. It must be stated before requests enter the queue.

\subsection{Conservative withdrawal NAV}
Book-value debt can overstate realizable pool value after a risk event has become publicly recognizable but before final write-down. Let
\begin{equation}
\widetilde{\mathcal A}_{p,t}
=
A_{p,t}+B_{p,t}
-H^{\mathrm{pend}}_{p,t}
-H^{\mathrm{liq}}_{p,t}
-H^{\mathrm{set}}_{p,t},
\label{eq:conservative-nav-assets}
\end{equation}
where the three non-negative holdbacks cover recognized pending-loss, liquidity, and settlement uncertainty under a versioned policy. Senior withdrawal pricing uses the senior claim on this conservative asset value, not face-value debt alone, once the corresponding risk state is active. A holdback is not a hidden reserve: it is a valuation deduction that remains auditable and is released or converted into final loss after reconciliation.

\begin{proposition}[No first-exit loss transfer after risk recognition]
\label{prop:no-first-exit}
Suppose (i) a publicly triggered holdback or finalized loss is applied before queue service, (ii) queued shares remain in total share supply until paid, and (iii) payment uses the current conservative share price. Then submitting a withdrawal request immediately before the trigger does not avoid the resulting per-share NAV reduction.
\end{proposition}
\begin{proof}
Escrowing changes ownership location but not total shares or claim rank. The holdback or loss reduces the NAV supporting all outstanding shares before the payment price is computed. Therefore queued and unqueued shares experience the same per-share reduction up to explicit integer rounding. \qedhere
\end{proof}

The proposition is conditional on timely public recognition. No accounting rule can charge an LP for information the protocol has not yet observed or validly encoded. Delayed loss recognition remains a governance, oracle, and adapter risk.

\subsection{Clearance bounds}
Suppose the applicable conservative share price at every service time lies in
\begin{equation}
0<\underline\pi\le\pi_t^S\le\overline\pi<\infty.
\label{eq:queue-price-bounds}
\end{equation}

\begin{theorem}[Share-denominated queue clearance]
\label{thm:queue-clearance}
Under \eqref{eq:queue-price-bounds}:
\begin{enumerate}[leftmargin=*]
\item if $C_T\ge Z_Q\overline\pi$, the entire queue can be paid by $T$ under any service order;
\item if $C_T<Z_Q\underline\pi$, full clearance by $T$ is impossible; and
\item if $\pi_t^S\equiv\pi$ is constant during service, full clearance by $T$ is possible if and only if $C_T\ge Z_Q\pi$.
\end{enumerate}
\end{theorem}
\begin{proof}
Every queued share costs at most $\overline\pi$, so $Z_Q\overline\pi$ is an upper bound on total cash required; this proves sufficiency. Every share costs at least $\underline\pi$, so $Z_Q\underline\pi$ is a lower bound; cash below it cannot clear the queue. Under constant price the exact requirement is $Z_Q\pi$, yielding necessity and sufficiency. \qedhere
\end{proof}

The theorem concerns settlement cash, not projected repayments. A dashboard may show expected service dates, but only settled free cash can move a request to \state{Paid}.

\subsection{Immediate liquidity and reserve buffers}
Let $A^W_{p,t}$ denote cash eligible for withdrawals after deducting open-position funding, protected capital, accrued protocol obligations, and minimum operating buffer. Immediate service satisfies
\begin{equation}
C_t^{\mathrm{now}}\le A^W_{p,t}.
\label{eq:withdrawal-cash-bound}
\end{equation}
Any excess claim remains queued. Borrower repayments, position closures, and senior deposits can add service cash; reserve or LBP balances do not become withdrawal cash unless their claim class is first released by the mandate.

\subsection{Run boundary}
A withdrawal queue makes delay explicit but does not eliminate run risk. If requested shares exceed expected liquid resources, share prices may fall as positions realize losses, which changes the cash required to redeem remaining shares. Emergency controls can pause new borrowing, shorten permitted position duration, raise rates, or allocate incoming cash to the queue. They cannot create settlement liquidity that the pool does not possess.

\begin{figure}[H]
\centering
\includegraphics[width=.80\textwidth]{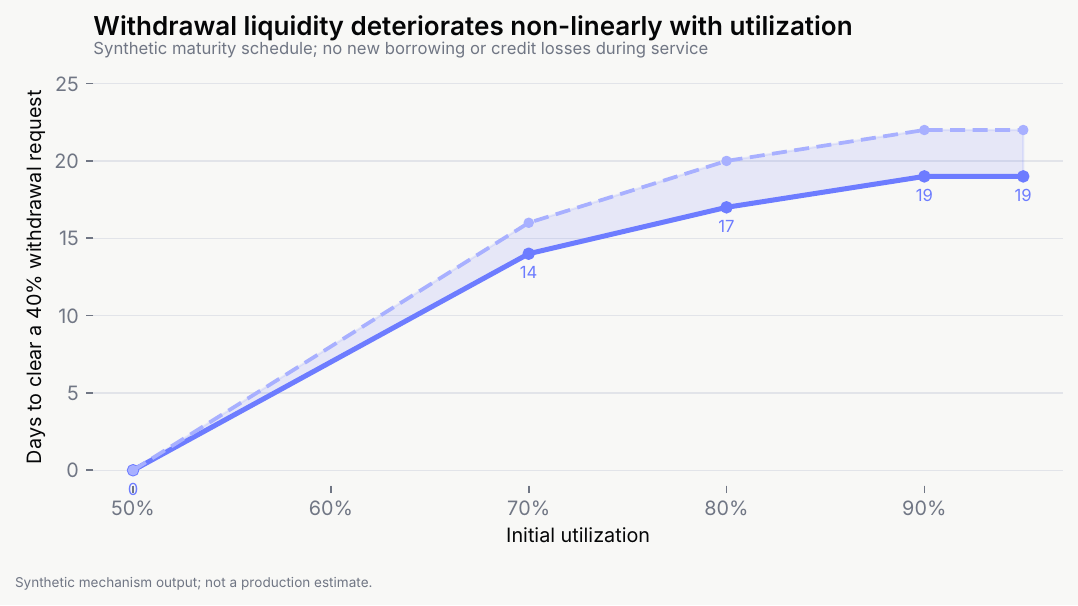}
\caption{Synthetic queue clearance after a withdrawal request equal to 40 percent of senior supply. The experiment assumes constant share price, no new borrowing, and no credit losses during service.}
\label{fig:withdrawal-queue}
\end{figure}

\subsection{Queue fairness}
FIFO minimizes rule ambiguity but can privilege early sophisticated withdrawers when loss recognition is delayed. Epoch pro-rata service distributes settlement liquidity across all requests in an epoch and reduces transaction-order races. The reference specification supports both but requires a pool to choose one at creation. Governance cannot retroactively change the priority of already queued claims.

\section{Leverage, Credit Capacity, and Admission}
\label{sec:leverage}

\subsection{Exact financing geometry}
For collateral $C$ and leverage $L>1$, gross notional and initial credit are
\begin{equation}
N=LC,
\qquad
D_0=(L-1)C.
\label{eq:leverage-credit}
\end{equation}
Hence the credit-funded fraction of gross exposure and gross exposure per unit of credit are
\begin{equation}
\delta(L)=\frac{L-1}{L}=1-\frac1L,
\qquad
\kappa(L)=\frac{L}{L-1}.
\label{eq:credit-geometry}
\end{equation}

\begin{proposition}[Credit geometry]
\label{prop:credit-geometry}
On $L>1$, $\delta(L)$ is strictly increasing and $\kappa(L)$ is strictly decreasing, with $\delta(L)\kappa(L)=1$. Thus higher leverage requires a larger financed share and supports less gross exposure per unit of scarce credit.
\end{proposition}

\begin{table}[H]
\centering
\caption{Exact credit requirement by leverage tier.}
\label{tab:credit-geometry}
\begin{tabular}{lccc}
\toprule
Leverage & Trader-funded share & Credit-funded share & Gross OI per \$1 credit\\
\midrule
$2\times$ & 50.0\% & 50.0\% & 2.00\\
$3\times$ & 33.3\% & 66.7\% & 1.50\\
$5\times$ & 20.0\% & 80.0\% & 1.25\\
\bottomrule
\end{tabular}
\end{table}

\begin{figure}[H]
\centering
\includegraphics[width=.74\textwidth]{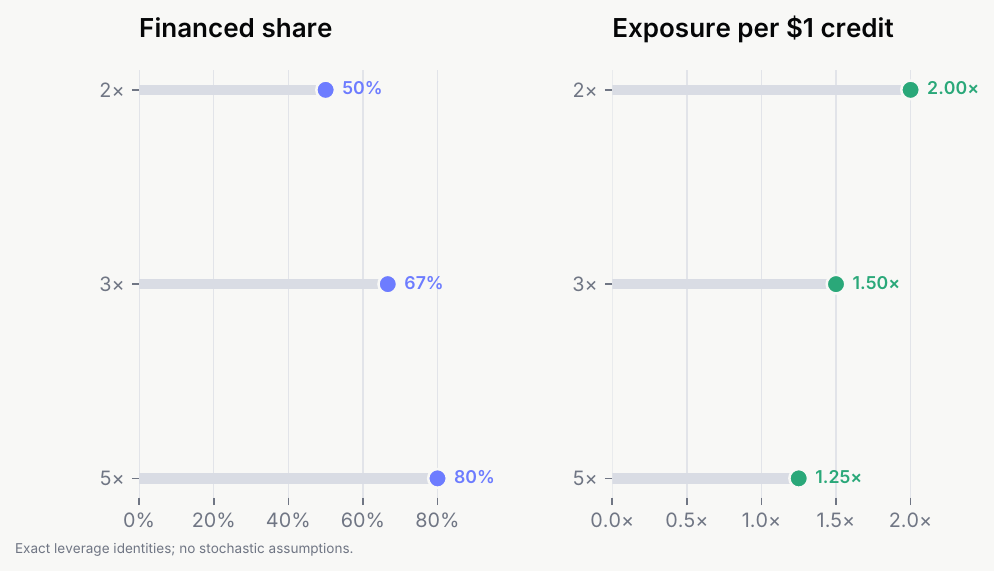}
\caption{Higher leverage increases the financed share of a position and decreases gross exposure supported by one unit of credit.}
\label{fig:capital}
\end{figure}

\subsection{Mixed-portfolio capacity}
For gross open interest $G_k$ admitted at leverage $L_k$, aggregate credit demand is
\begin{equation}
B=\sum_k G_k\frac{L_k-1}{L_k}.
\label{eq:mixed-credit-demand}
\end{equation}
Let $S^{\mathrm{lend}}$ be unencumbered senior resources after deducting withdrawal, operational, and mandate buffers.

\begin{theorem}[Credit-pool capacity]
\label{thm:credit-capacity}
A necessary accounting condition for fully funded admission of the mixed portfolio is
\begin{equation}
\sum_k G_k\frac{L_k-1}{L_k}\le S^{\mathrm{lend}}.
\label{eq:credit-capacity}
\end{equation}
For a single leverage tier $L$, the maximum gross open interest before other limits is
\begin{equation}
G_{\max}^{\mathrm{credit}}=S^{\mathrm{lend}}\frac{L}{L-1}.
\label{eq:gross-capacity-single-tier}
\end{equation}
\end{theorem}
\begin{proof}
Every admitted position requires real stablecoin credit equal to its financed fraction. Summing those amounts gives \eqref{eq:mixed-credit-demand}; the pool cannot transfer more than its unencumbered lendable resources. Solving the single-tier inequality for $G$ gives \eqref{eq:gross-capacity-single-tier}. \qedhere
\end{proof}

This condition is necessary but not sufficient. It ignores executable exit capacity, reserves, concentration, queue liquidity, and venue mode.

\subsection{The effective leverage gate}
The displayed leverage limit is the minimum of independent controls:
\begin{equation}
L_i^{\mathrm{eff}}=
\min\left\{
5,
L_i^{\mathrm{acct}},
L_e^{\mathrm{market}},
L_{i,e}^{\mathrm{exec}},
L_e^{\mathrm{time}},
L_{p,e}^{\mathrm{conc}},
L_p^{\mathrm{capital}},
L_i^{\mathrm{jur}}
\right\}.
\label{eq:effective-leverage}
\end{equation}
The absolute protocol ceiling is $5\times$. It is not an entitlement. An account may qualify for a higher tier while the selected market, book, time-to-close, or pool capital still forces a lower result.

The execution component $L_{i,e}^{\mathrm{exec}}$ is determined by the exact book-dependent certificate of \citet{nechepurenko2026axient}, including aggregate use of shared liquidity. The capital component applies both \eqref{eq:credit-capacity} and the protection-capacity condition \eqref{eq:reserve-supported-oi}. The time component compresses leverage as the hard-flat horizon approaches.

\subsection{Phased leverage policy}
A reference product rollout is:
\begin{table}[H]
\centering
\caption{Illustrative phased leverage policy. Tiers remain subject to every gate in \eqref{eq:effective-leverage}.}
\label{tab:leverage-rollout}
\begin{tabularx}{\textwidth}{lclX}
\toprule
Phase & Maximum & Account scope & Primary evidence gate\\
\midrule
Private beta & $2\times$ & Core accounts & settled execution, reconciliation, and queue performance\\
Whitelist alpha & up to $3\times$ & Advanced accounts & account history, deeper markets, measured aggregate exit capacity\\
Pro rollout & up to $5\times$ & Pro accounts & Tier A markets, stronger loss-bearing assessment, dedicated OI and reserve caps\\
\bottomrule
\end{tabularx}
\end{table}

The account labels are protocol risk tiers, not legal classifications. A jurisdiction adapter may later map them to regulated customer categories, but the protocol should not use the word ``qualified'' without an external classification process.

\subsection{Admission conjunction}
A position is admitted only if all of the following hold simultaneously:
\begin{enumerate}[leftmargin=*]
\item the trader supplies collateral and passes the account tier;
\item the venue and backing mode satisfy the pool mandate;
\item entry assets can be acquired within price and size limits;
\item the position-level robust debt-clearing certificate is feasible;
\item aggregate shared-book and bonded-MM capacities are not double counted;
\item pool cash, utilization, withdrawal buffer, and concentration limits remain valid;
\item market, LBP, reserve, and senior-capital encumbrance limits remain valid; and
\item the requested leverage does not exceed \eqref{eq:effective-leverage}.
\end{enumerate}
Failure of any conjunct produces a smaller quote or rejection. The protocol does not compensate for a failed execution certificate by silently drawing more reserve.

\section{Protocol Architecture and State Machines}
\label{sec:architecture}

\subsection{Contract modules}
The on-chain reference architecture is modular:
\begin{table}[H]
\centering
\small
\caption{Reference smart-contract modules and financial authority.}
\label{tab:contract-modules}
\begin{tabularx}{\textwidth}{lX}
\toprule
Module & Canonical responsibility\\
\midrule
\texttt{PoolFactory} & creates immutable or version-pinned pool mandates\\
\texttt{CreditPool} & senior deposits, virtual-offset shares, conservative NAV, lendable cash, debt receivables, loss-participating withdrawal queue\\
\texttt{LBPVault} & junior provider deposits, shares, encumbrance, income, and impairment ahead of senior principal\\
\texttt{ReserveVault} & segregated market, pool, and optional global reserve balances without an ordinary LP redemption path\\
\texttt{PositionManager} & position admission, collateral lock, debt-first settlement, residual release\\
\texttt{DebtManager} & borrow index, debt shares, accrual, partial settlement, idempotent repayment, and write-down\\
\texttt{RiskManager} & leverage, concentration, market, venue, time, and capital gates\\
\texttt{VenueRegistry} & capability evidence, adapter versions, backing modes, and exposure caps\\
\texttt{LiquidationManager} & admissibility, aggregate priority, bounties, and failure states\\
\texttt{MMCommitmentRegistry} & bonded capacity, allocation, performance, and slashing\\
\texttt{FeeRouter} & interest, fee, reserve, protocol, LBP, liquidator, and MM allocation\\
\texttt{GovernanceTimelock} & prospective parameter changes, emergency powers, and upgrade delay\\
\bottomrule
\end{tabularx}
\end{table}

Modules may be combined in an implementation, but their authorities remain logically separated. In particular, a router cannot burn debt shares, and a risk oracle cannot transfer LP assets.

\subsection{Position state machine}
The canonical position lifecycle is
\begin{equation}
\begin{aligned}
\state{Requested}
&\rightarrow\state{Admitted}
\rightarrow\state{FundingPending}
\rightarrow\state{Open}\\
&\rightarrow\state{ReduceOnly}
\rightarrow\state{Liquidating}
\rightarrow\state{SettlementPending}\\
&\rightarrow\state{PartiallySettled}
\rightarrow\{\state{Repaid},\state{Shortfall}\}.
\end{aligned}
\label{eq:position-state-machine}
\end{equation}
Exceptional transitions include
\begin{equation}
\begin{aligned}
\{\state{SettlementPending},\state{PartiallySettled}\}
\rightarrow{}&\{\state{SettlementFailed},\state{EvidenceConflict},\\
&\state{EventCloseException}\}.
\end{aligned}
\end{equation}
and
\begin{equation}
\state{Shortfall}\rightarrow\state{WaterfallFinalized}.
\end{equation}
No transition to \state{Repaid} occurs until debt shares are burned against recognized settled value. Duplicate settlement evidence produces no financial transition, and cumulative settlement cannot exceed matched capacity. If the position becomes fully funded spot after debt repayment, that claim can move to a separate finality/redeem state without remaining on the credit-pool balance sheet.

\subsection{Pool and queue states}
A pool operates in
\begin{equation}
\state{Active},\quad
\state{RateLimited},\quad
\state{ReduceOnly},\quad
\state{WithdrawalsOnly},\quad
\state{Paused},\quad
\state{Recovery}.
\end{equation}
Risk deterioration moves the pool monotonically toward fewer permitted actions. Emergency powers can prevent new exposure and release already-settled assets; they cannot rewrite the loss waterfall or seize residual trader equity outside the debt-first rule.

Withdrawal requests follow
\begin{equation}
\state{Requested}\rightarrow\state{Queued}\rightarrow
\state{PartiallyPaid}\rightarrow\state{Paid},
\end{equation}
with \state{Cancelled} permitted only before the request has consumed service cash and subject to the pool policy.

\subsection{Replaceable off-chain agents}
On-chain financial authority does not remove the need for computation outside the chain. Replaceable agents perform:
\begin{itemize}[leftmargin=*]
\item order-book reconstruction and robust quote generation;
\item venue and RFQ routing;
\item hard-flat and liquidation transaction submission;
\item settlement, redemption, and oracle indexing;
\item risk simulation, monitoring, and alerting.
\end{itemize}
An agent proposes a transition; contracts decide whether it is admissible. If the primary Axient backend disappears, another compatible keeper should be able to submit the same state transition from public data and receive the same result.

\subsection{Venue adapters}
Adapters translate protocol intents into venue-specific actions and evidence. Each adapter must define:
\begin{enumerate}[leftmargin=*]
\item custody location and beneficial claim;
\item order authorization and liquidation authority;
\item success evidence for match, settlement, and redemption;
\item cancellation and failure semantics;
\item maximum outstanding exposure and settlement horizon;
\item price, quantity, and asset-normalization rules;
\item whether the position is native, verifiable, attested, or synthetic backing mode.
\end{enumerate}
The core protocol does not assume support from a particular venue. An adapter that cannot provide enforceable lien or contract-controlled custody is assigned a weaker mode, lower caps, or a separate pool rather than treated as equivalent to native backing.

\section{Deterministic Verification and Synthetic Evaluation}
\label{sec:simulation-design}

\subsection{Evidence classes}
The release separates four evidence classes:
\begin{enumerate}[leftmargin=*]
\item exact-decimal fixtures for accounting identities;
\item deterministic implementation-level invariant checks;
\item a finite registered scenario grid; and
\item fixed-seed synthetic stochastic comparisons.
\end{enumerate}
Only the mathematical statements are proofs. Fixtures and code detect implementation inconsistencies; stochastic outputs compare mechanisms under author-specified distributions. None estimates production APY, loss frequency, venue liquidity, settlement reliability, or commercial viability.

\subsection{Exact fixtures and deterministic checks}
The r0.2.0 layer contains 16 exact-decimal fixtures and 3,641 checks covering proportional share minting, share minting after loss, debt-index accrual, waterfall conservation, fee routing, bonded-MM performance, queue clearance, and credit geometry at $2\times$, $3\times$, and $5\times$.

The r0.2.1 implementation-parity layer adds 12 integer-base-unit fixtures and 27,441 checks for:
\begin{itemize}[leftmargin=*]
\item partial settlement, cumulative matched-capacity bounds, and replay resistance;
\item debt-first payment of accrued interest and principal, followed by trader residual release;
\item separation of accrued receivables from settled cash;
\item reserve non-redemption and separation from LBP shares;
\item LBP withdrawal blocking while capital is encumbered;
\item loss participation by escrowed withdrawal shares until payment;
\item virtual-offset and minimum-deposit protection against zero-share inflation outcomes;
\item multiply-divide-up and multiply-divide-down bounds;
\item debt-share rounding that does not understate the obligation;
\item exact fee conservation with explicit dust; and
\item deterministic largest-remainder allocation of simultaneous global-reserve claims.
\end{itemize}
All 28 fixtures and all 31,082 deterministic checks pass. These are reference-implementation results, not evidence that untested or unaudited smart contracts are secure.

\begin{table}[H]
\centering
\small
\caption{Implementation-parity fixtures added in r0.2.1.}
\label{tab:implementation-parity}
\begin{tabularx}{\textwidth}{lX}
\toprule
Fixture class & Required property\\
\midrule
Partial settlement & unique evidence, cumulative cap, no duplicate repayment\\
Residual allocation & interest and principal first; excess to trader\\
Accrual/cash separation & unpaid interest cannot create withdrawal cash\\
Reserve/LBP separation & reserve has no LP redemption; LBP can be encumbered\\
Queued-share loss participation & payout uses current post-loss NAV\\
Initial-share guard & configured minimum deposit mints non-zero shares\\
Integer arithmetic & explicit up/down rounding and bounded dust\\
Shared reserve & exact deterministic allocation without overpayment\\
\bottomrule
\end{tabularx}
\end{table}

\subsection{Finite scenario grid}
The 84-row scenario grid crosses leverage $L\in\{2,3,5\}$, recovery ratios
\begin{equation}
\rho\in\{1.00,0.90,0.82,0.80,0.70,0.50,0.00\},
\end{equation}
and four protection configurations. Gross notional is fixed at one million settlement units. The grid makes two boundaries visible. First, $5\times$ begins with debt equal to 80 percent of gross exposure. Second, at zero recovery, no finite spread or ordinary execution fee protects senior principal unless reserves, LBP capital, or enforceable bonded capacity are present.

\subsection{Monte Carlo mechanism comparison}
The Monte Carlo experiment uses 200,000 paths and seed 20260714. Senior supply is ten million units. Utilization is drawn from a clipped $\mathrm{Beta}(5,2)$ distribution. The registered 30-day borrow curve has base rate 6 percent, kink 80 percent, first slope 12 percent, and second slope 65 percent. These are synthetic policy inputs.

For leverage tier $j$, recovery is $1-H_j$, where the synthetic haircut contains a common component, a heavy-tailed positive shock, tier-specific size effect, three registered incident channels, and an idiosyncratic beta shock. In compact form,
\begin{equation}
\begin{aligned}
H_j={}&0.012+0.055X_0+0.018[T_5]^+
+0.30I_1X_1+0.28I_2X_2\\
&+I_3(0.35+0.60X_3)+\varepsilon_j+s_j,
\end{aligned}
\label{eq:synthetic-haircut}
\end{equation}
where all distributions, incident probabilities, and tier shifts are listed in \Cref{app:synthetic-parameters}. The process is not calibrated to an observed venue.

Three leverage mixes are compared: $2\times$ only; phased $(55\%,30\%,15\%)$ across $2\times$, $3\times$, and $5\times$ gross OI; and a $5\times$-heavy mix $(25\%,25\%,50\%)$. Four capital configurations are:
\begin{enumerate}[leftmargin=*]
\item senior only;
\item senior plus a 1 percent reserve;
\item reserve plus LBP capital equal to 3 percent of senior supply; and
\item reserve plus LBP plus bonded-MM capacity credited on the higher leverage tiers.
\end{enumerate}
Interest allocations change by configuration; the protection layers are therefore not free.

\subsection{Withdrawal and contagion experiments}
The queue experiment uses 20,000 paths for each utilization--withdrawal cell, 100 equal-principal synthetic loans, an 8 percent cash buffer, and a 365-day horizon. It assumes constant share price, no new borrowing, and no credit losses during queue service. Its result is a liquidity-timing comparison, not a run model.

The contagion experiment uses 500,000 two-pool paths and seed 20260716. It compares strict pool isolation, where each pool receives its own 1 percent extra reserve, with a shared 2 percent global reserve consumed sequentially by pool one and then pool two. Common heavy-tailed shocks induce correlation. The experiment is designed to distinguish total protection from protection contagion.

\subsection{Statistical precision}
For each reported incidence $\widehat p$ based on $n$ paths, the release reports the binomial standard error
\begin{equation}
\widehat{\mathrm{se}}(\widehat p)=
\sqrt{\frac{\widehat p(1-\widehat p)}{n}}.
\end{equation}
In the Monte Carlo table below, the largest 95 percent normal-approximation half-width for senior-loss incidence is approximately 0.057 percentage points. This is simulation precision conditional on the specified generator, not model uncertainty.

\subsection{Results: capital protection stack}
\label{sec:simulation-results}
\Cref{tab:mc-phased} reports the phased leverage mix. The protection stack lowers senior-loss incidence from 1.7225 percent in the senior-only configuration to 0.3995 percent with reserve, LBP, and bonded-MM capacity. Mean senior 30-day return also falls because more interest is allocated to protection providers and MM commitments. The comparison is therefore not a free Pareto improvement.

\begin{table}[H]
\centering
\scriptsize
\caption{Fixed-seed synthetic results for the phased leverage mix; 200,000 paths. Returns and losses are percentages of the relevant initial provider capital over 30 days.}
\label{tab:mc-phased}
\resizebox{\textwidth}{!}{%
\begin{tabular}{lrrrrrr}
\toprule
Configuration & \shortstack{Senior loss\\incidence} & \shortstack{Mean senior\\loss} & \shortstack{Mean senior\\return} & \shortstack{5th pct. senior\\return} & \shortstack{Negative senior\\return} & \shortstack{Mean LBP\\return}\\
\midrule
Senior only & 1.7225\% & 0.1489\% & 1.3108\% & 0.3550\% & 0.7370\% & --\\
+ reserve & 0.7470\% & 0.1373\% & 1.1927\% & 0.3347\% & 0.4450\% & --\\
+ reserve + LBP & 0.4030\% & 0.1234\% & 1.1093\% & 0.3124\% & 0.4000\% & 2.7812\%\\
+ reserve + LBP + bonded MM & 0.3995\% & 0.1180\% & 1.0660\% & 0.3001\% & 0.3975\% & 2.8396\%\\
\bottomrule
\end{tabular}%
}
\end{table}

The LBP return is more dispersed than the senior return. Under the final configuration, 0.4155 percent of synthetic paths produce a negative 30-day LBP return, versus 0.3975 percent for senior LPs. The two incidences are close because the junior tranche is small and the registered incident process is sparse; conditional loss severity remains materially larger for the LBP.

\begin{figure}[H]
\centering
\includegraphics[width=.78\textwidth]{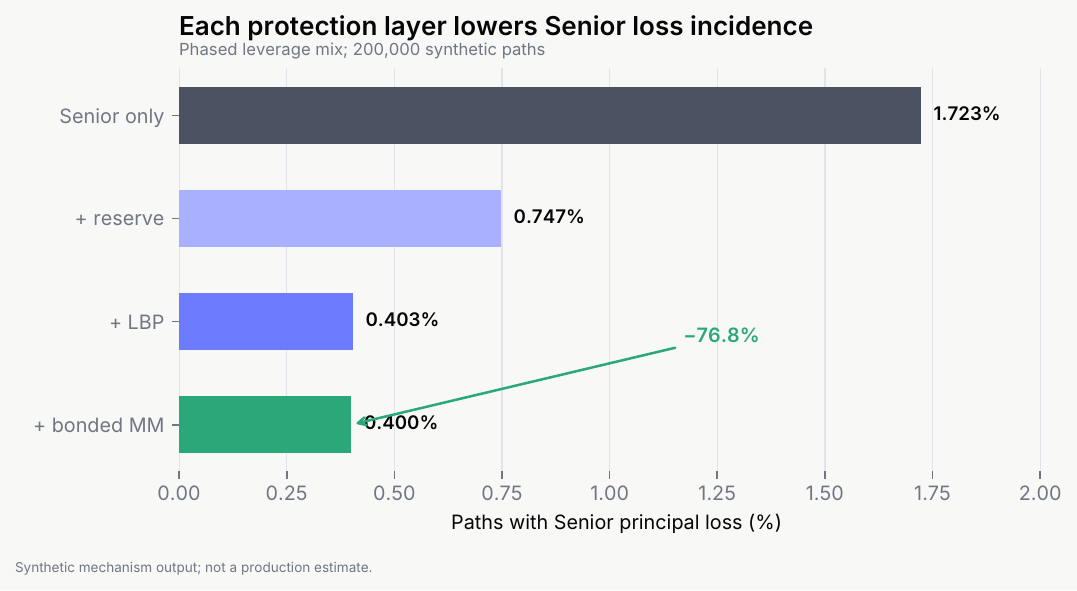}
\caption{Synthetic senior-principal loss incidence for the phased leverage mix. The values compare mechanisms under the registered generator and are not production-frequency estimates.}
\label{fig:senior-loss}
\end{figure}

\begin{figure}[H]
\centering
\includegraphics[width=.78\textwidth]{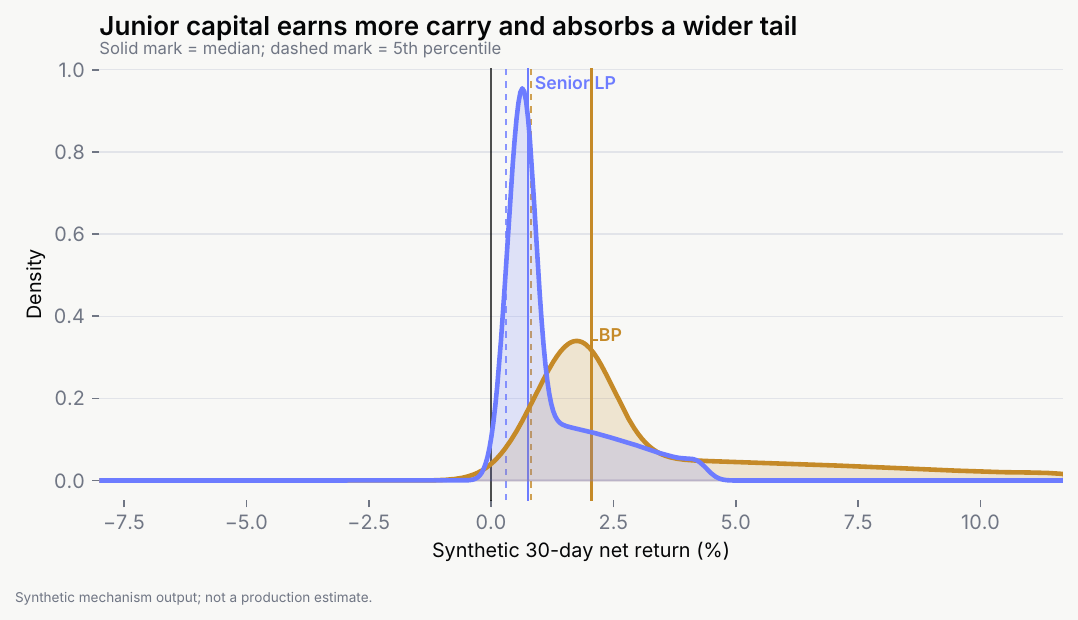}
\caption{Synthetic provider-return distributions under reserve, LBP, and bonded-MM protection. Junior capital receives more carry and absorbs a wider tail.}
\label{fig:provider-returns}
\end{figure}

\subsection{Results: leverage mix}
The protected configuration is applied to three leverage mixes in \Cref{tab:leverage-mix-results}. A $2\times$-only pool supports more gross OI per unit of senior supply and has the lowest mean senior loss. A $5\times$-heavy pool supports only 1.032 units of gross OI per unit of senior supply under the same utilization draw and produces the largest mean senior loss. The synthetic raw-shortfall incidence happens to coincide for the phased and $5\times$-heavy mixes under this generator; severity differs and should be the primary interpretation.

\begin{table}[H]
\centering
\small
\caption{Synthetic leverage-mix trade-off under reserve, LBP, and bonded-MM protection.}
\label{tab:leverage-mix-results}
\resizebox{\textwidth}{!}{%
\begin{tabular}{lrrrrr}
\toprule
Mix & \shortstack{Gross OI /\\senior supply} & \shortstack{Raw shortfall\\incidence} & \shortstack{Mean raw\\shortfall} & \shortstack{Senior loss\\incidence} & \shortstack{Mean senior\\loss}\\
\midrule
$2\times$ only & 1.427 & 0.3535\% & 0.1088\% & 0.3330\% & 0.0951\%\\
Phased & 1.199 & 1.7225\% & 0.1489\% & 0.3995\% & 0.1180\%\\
$5\times$ heavy & 1.032 & 1.7225\% & 0.1904\% & 0.4205\% & 0.1367\%\\
\bottomrule
\end{tabular}%
}
\end{table}

\begin{figure}[H]
\centering
\includegraphics[width=.80\textwidth]{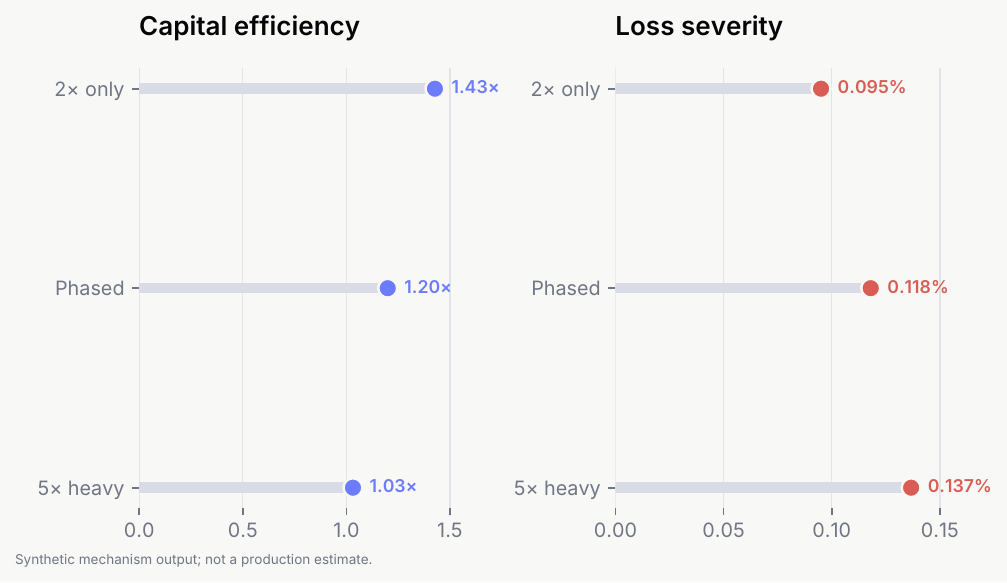}
\caption{Synthetic leverage-mix comparison. The bars use different units and should be read as a capital-efficiency and loss-severity trade-off, not a common scale.}
\label{fig:leverage-mix}
\end{figure}

\subsection{Results: withdrawal queue}
For a request equal to 40 percent of senior supply, the registered queue experiment gives:
\begin{table}[H]
\centering
\caption{Synthetic queue clearance for a 40 percent withdrawal request.}
\label{tab:queue-results}
\begin{tabular}{lrr}
\toprule
Initial utilization & Median days & 95th percentile days\\
\midrule
50\% & 0 & 0\\
70\% & 14 & 16\\
80\% & 17 & 20\\
90\% & 19 & 22\\
95\% & 19 & 22\\
\bottomrule
\end{tabular}
\end{table}
At 95 percent utilization and a 60 percent withdrawal request, the median and 95th percentile rise to 27 and 31 days. These numbers follow directly from the assumed loan-maturity generator and cannot be interpreted as expected Axient withdrawal times.

\subsection{Results: isolation versus shared reserve}
\Cref{tab:contagion-results} shows the synthetic two-pool comparison. A shared reserve lowers aggregate network senior-loss incidence because capital can move to the pool that needs it. It also makes pool two's protection depend on pool one's prior draw.

\begin{table}[H]
\centering
\small
\caption{Synthetic two-pool contagion experiment; 500,000 paths. Mean losses are percentages of one pool's senior supply.}
\label{tab:contagion-results}
\begin{tabular}{lrrr}
\toprule
Regime & \shortstack{Pool 2 senior-loss\\incidence} & \shortstack{Any network senior-loss\\incidence} & \shortstack{Mean network\\senior loss}\\
\midrule
Strict isolation & 0.9046\% & 1.8236\% & 0.1928\%\\
Shared global reserve & 0.5644\% & 1.1244\% & 0.1787\%\\
\bottomrule
\end{tabular}
\end{table}

In the shared regime, pool one draws more than its isolated fair slice on 0.93 percent of paths. Conditional on those paths, pool two's senior-loss incidence is 1.6774 percent, compared with 0.5540 percent when the first draw is at or below the fair slice. This is the protection-contagion channel of \Cref{prop:protection-contagion}.

\begin{figure}[H]
\centering
\includegraphics[width=.78\textwidth]{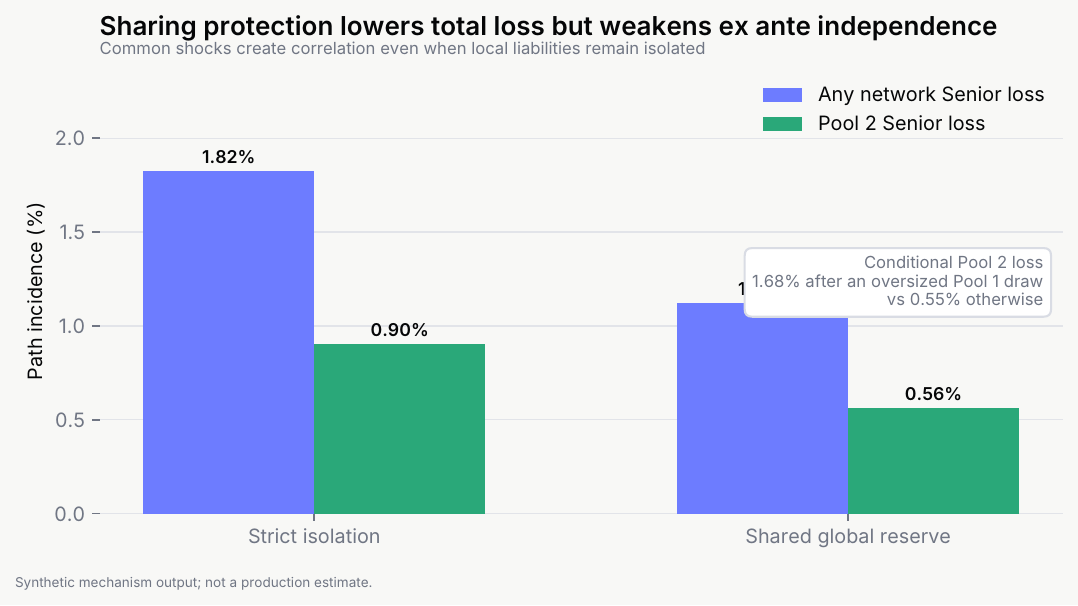}
\caption{Synthetic total-network senior-loss incidence under isolated and shared protection. Sharing lowers total loss in the experiment but weakens ex ante independence of protection.}
\label{fig:contagion}
\end{figure}

\subsection{Interpretation and negative controls}
The results support only directional mechanism statements:
\begin{enumerate}[leftmargin=*]
\item prior protection layers reduce senior impairment under the registered loss process;
\item junior protection has a cost in foregone senior or protocol carry;
\item bonded capacity matters only when collateralized and not double counted;
\item higher leverage consumes more credit and raises loss severity under the registered size effect;
\item withdrawal liquidity worsens with utilization; and
\item shared reserves exchange isolation for system-level capital efficiency.
\end{enumerate}

The senior-only configuration, zero-recovery scenario rows, uncollateralized MM promises, and strict isolation comparison are negative controls. Complete loss of executable liquidity, invalid settlement, stablecoin failure, oracle corruption, and smart-contract compromise are not assigned empirical probabilities. They remain external or failure-set assumptions.

\section{Endogenous Capital Supply and Rate Equilibrium}
\label{sec:endogenous-capital}

The preceding sections treat pool capital, LBP capacity, and market-making support as state variables. A live protocol cannot assume that those quantities remain fixed. Capital providers compare protocol compensation with outside opportunities and perceived loss; traders alter leverage demand in response to price, risk, and admission constraints; and the utilization rate is jointly determined by both sides. This section endogenizes those interactions.

\subsection{Senior provider participation}
Let $\Iset_S$ be a finite set of Senior Credit LPs. Provider $i$ has investable wealth $\bar s_i$, outside annualized return $o_i$, loss-aversion coefficient $\lambda_i\ge 0$, queue-delay aversion $\kappa_i\ge 0$, and a continuous non-decreasing participation function $G_i:\R\rightarrow[0,1]$. Let $y^S_t$ be expected Senior provider yield, $\mu^S_t$ expected impairment, $q_t$ expected queue delay, and $w_t$ the observed withdrawal rate.

\begin{definition}[Endogenous Senior supply]
\label{def:endogenous-senior-supply}
Senior supply is
\begin{equation}
S_t=\sum_{i\in\Iset_S}\bar s_i
G_i\!\left(y^S_t-o_i-\lambda_i\mu^S_t-\kappa_i q_t-\eta_i w_t\right),
\label{eq:senior-supply}
\end{equation}
where $\eta_i\ge 0$ captures withdrawal coordination or herding.
\end{definition}

Equation~\eqref{eq:senior-supply} is deliberately behavioral rather than representative-agent optimal. It allows heterogeneous reservation returns, risk tolerances, and queue sensitivity. A logistic $G_i$ is used in the synthetic agent model, but the formal results require only continuity and boundedness.

\subsection{LBP participation}
Let $\Iset_J$ denote Liquidation Backstop Providers. Provider $j$ has endowment $\bar j_j$, required return $o^J_j$, impairment aversion $\lambda^J_j$, and participation function $H_j$. If $y^J_t$ is expected LBP compensation and $\mu^J_t$ expected impairment, aggregate junior provider capital is
\begin{equation}
J_t=\sum_{j\in\Iset_J}\bar j_j
H_j\!\left(y^J_t-o^J_j-\lambda^J_j\mu^J_t-\kappa^J_j\ell_t\right),
\label{eq:lbp-supply}
\end{equation}
where $\ell_t$ is the expected lock or encumbrance cost. Unlike a reserve, $J_t$ is redeemable provider capital subject to the protocol's queue and encumbrance rules. Unlike Senior supply, it is impaired before Senior principal.

\subsection{Trader credit demand}
Let $\Iset_T$ be a finite set of traders. Trader $n$ has collateral budget $c_n$, perceived annualized edge $a_{n,t}$, risk-aversion coefficient $\gamma_n$, and rate sensitivity $\zeta_n$. For leverage tier $L\in\mathcal L$, with financed fraction $\delta(L)=(L-1)/L$, define utility
\begin{equation}
U_{n,t}(L)=L a_{n,t}
-\zeta_n r_t\delta(L)h_n
-\gamma_n \chi_t (L-1)^{\nu}
-f_t(L),
\label{eq:trader-tier-utility}
\end{equation}
where $h_n$ is expected holding time, $\chi_t$ is the trader's risk signal, $\nu>1$ gives convex leverage sensitivity, and $f_t(L)$ collects fees and impact. The selected tier is the highest-utility admissible tier with non-negative utility. Aggregate credit demand is
\begin{equation}
B_t(r)=\sum_{n\in\Iset_T}c_n\bigl(L_{n,t}(r)-1\bigr)
\one\{U_{n,t}(L_{n,t})\ge 0\}.
\label{eq:credit-demand}
\end{equation}

Demand is not automatically admitted. The admission controller constrains $B_t$ by pool cash, utilization, aggregate executable exit capacity, reserve policy, concentration limits, and venue capability.

\subsection{Rate and utilization equilibrium}
Let $r(U)$ be a continuous utilization-based rate curve and let $S(r,\mu,q,w)$ and $B(r,\chi)$ denote aggregate supply and admitted demand under fixed non-rate state variables. Define
\begin{equation}
\Phi(U)=\frac{B(r(U),\chi)}{S(r(U),\mu,q,w)}.
\label{eq:utilization-map}
\end{equation}

\begin{definition}[Capital-market fixed point]
\label{def:capital-fixed-point}
An interior capital-market equilibrium is $U^*\in[0,\bar U]$ satisfying
\begin{equation}
U^*=\Phi(U^*),
\label{eq:endogenous-utilization-fixed-point}
\end{equation}
where $\bar U<1$ is the protocol utilization ceiling.
\end{definition}

\begin{theorem}[Existence of a utilization fixed point]
\label{thm:utilization-existence}
Suppose $r$, $B$, and $S$ are continuous; $S$ is strictly positive on $[0,\bar U]$; and admission clips $\Phi$ to $[0,\bar U]$. Then at least one fixed point $U^*$ exists.
\end{theorem}

\begin{proposition}[Uniqueness and local stability]
\label{prop:utilization-contraction}
If
\begin{equation}
\sup_{U\in[0,\bar U]}|\Phi'(U)|<1,
\label{eq:utilization-contraction}
\end{equation}
then the fixed point is unique. The damped update
\begin{equation}
U_{t+1}=(1-\omega)U_t+\omega\Phi(U_t),\qquad \omega\in(0,1],
\label{eq:damped-utilization}
\end{equation}
converges locally whenever $|(1-\omega)+\omega\Phi'(U^*)|<1$.
\end{proposition}

The condition can fail when trader demand is highly rate-inelastic, provider supply responds discontinuously to loss signals, or queued withdrawals create a strong feedback from delay into supply. The protocol should therefore report convergence diagnostics rather than assume a unique equilibrium.

\subsection{Risk spread beyond utilization}
A utilization-only rate can move in the wrong direction under stress. Let $\mu$ denote expected credit loss and suppose stress reduces admitted demand and therefore utilization, $\partial U/\partial\mu<0$. If $r=r_U(U)$ with $r_U'>0$, then
\begin{equation}
\frac{\partial r}{\partial\mu}=r_U'(U)\frac{\partial U}{\partial\mu}<0.
\label{eq:wrong-way-rate}
\end{equation}
The borrow rate falls even though expected loss rises. A risk-aware rate
\begin{equation}
r(U,\mu)=r_U(U)+r_R(\mu)
\label{eq:risk-aware-rate}
\end{equation}
avoids this wrong-way response if
\begin{equation}
r_R'(\mu)>-r_U'(U)\frac{\partial U}{\partial\mu}.
\label{eq:risk-spread-condition}
\end{equation}
The agent-based results in \Cref{sec:agent-results} show this issue directly: the stress regime has lower mean utilization and a lower mean utilization-only borrow rate than the normal regime despite higher realized impairment.

\subsection{Endogenous credit capacity across leverage tiers}
Let $w_L$ be the share of gross admitted exposure at tier $L$ and let $S_t^{\mathrm{lend}}$ be lendable Senior capital. Under utilization cap $\bar U$, credit capacity is $\bar U S_t^{\mathrm{lend}}$. The corresponding maximum gross open interest is
\begin{equation}
G_t^{\max}=\frac{\bar U S_t^{\mathrm{lend}}}{\sum_{L\in\mathcal L}w_L\delta(L)}.
\label{eq:endogenous-gross-capacity}
\end{equation}

\begin{proposition}[Leverage-mix capital crowding]
\label{prop:leverage-mix-crowding}
Holding lendable Senior capital and the utilization cap fixed, shifting gross-exposure weight from a lower leverage tier to a higher tier weakly reduces $G_t^{\max}$ whenever $\delta(L)$ is increasing in $L$.
\end{proposition}

Thus a $5\times$ tier is not merely a higher-risk product feature. It uses more credit per unit of gross exposure and can crowd out lower-leverage demand even before risk limits bind.

\section{Strategic Market Making and Liquidation Competition}
\label{sec:strategic-execution}

Executable market liquidity and liquidation execution are supplied by self-interested agents. Treating displayed depth, bonded capacity, or liquidator participation as fixed can overstate the protection available to the credit pool. This section formalizes the relevant incentives.

\subsection{Revocable and bonded market-making capacity}
For market maker $k$, let $m^o_{k,t}$ be ordinary displayed capacity and $m^b_{k,t}$ be bonded commitment capacity. Let $\pi^{\mathrm{spr}}_{k,t}$ be expected spread and rebate income, $\pi^{\mathrm{com}}_{k,t}$ commitment premium, $c^{\mathrm{inv}}_{k,t}$ inventory cost, $c^{\mathrm{as}}_{k,t}$ adverse-selection cost, and $c^{\mathrm{lock}}_{k,t}$ capital lock cost. Open quote participation satisfies
\begin{equation}
V^o_{k,t}=\pi^{\mathrm{spr}}_{k,t}-c^{\mathrm{inv}}_{k,t}-c^{\mathrm{as}}_{k,t}.
\label{eq:mm-open-utility}
\end{equation}
A market maker can withdraw ordinary depth whenever its private continuation value becomes negative. This is why visible book depth is an observation, not a guarantee.

A bonded commitment adds a bond $b_{k,t}$ and a delivery obligation. Let $g_{k,t}$ be the private gain from refusing execution when called and let $s_{k,t}\le b_{k,t}$ be enforceable slashing.

\begin{proposition}[Commitment incentive compatibility]
\label{prop:mm-incentive-compatibility}
Delivery weakly dominates strategic non-delivery whenever
\begin{equation}
\pi^{\mathrm{com}}_{k,t}+\pi^{\mathrm{exec}}_{k,t}+s_{k,t}
\ge g_{k,t}+c^{\mathrm{exec}}_{k,t},
\label{eq:mm-incentive-condition}
\end{equation}
where $\pi^{\mathrm{exec}}$ is execution compensation and $c^{\mathrm{exec}}$ is delivery cost. If $g_{k,t}$ is unbounded or cannot be verified, no finite bond establishes universal delivery.
\end{proposition}

The result does not make bonded capacity equivalent to cash. Capacity substitutes for reserve only to the extent that delivery or slashing is enforceable, settlement is final, and the bond is denominated in an asset that retains value during the same shock.

\subsection{Liquidator entry and capacity}
Liquidator $\ell$ has executable capacity $q_{\ell,t}$, private execution cost $c_{\ell,t}(q)$, bounty rate $b_t$, and estimated success probability $p_{\ell,t}^{\mathrm{succ}}$. The liquidator enters when
\begin{equation}
b_t q_{\ell,t}p_{\ell,t}^{\mathrm{succ}}
\ge c_{\ell,t}(q_{\ell,t}).
\label{eq:liquidator-entry}
\end{equation}
Let $\mathcal A_t$ be the active liquidator set and $Q_t$ required liquidation notional.

\begin{theorem}[Sufficient liquidation coverage]
\label{thm:liquidator-coverage}
If the venue adapter admits the required risk-reducing execution and
\begin{equation}
\sum_{\ell\in\mathcal A_t}q_{\ell,t}\ge Q_t,
\label{eq:liquidator-capacity-condition}
\end{equation}
then aggregate submitted capacity is sufficient to cover the liquidation need. The theorem does not guarantee venue fill, settlement, or price: those remain adapter and market-state conditions.
\end{theorem}

Permissionless submission at the Axient contract boundary therefore means that any address can trigger a valid protocol transition. It does not imply that an external operator-gated venue must accept an order from any address.

\subsection{Capacity-induced demand and the backstop paradox}
Enforceable execution capacity can protect existing positions, but it can also raise the amount of credit admitted. Let $M$ denote backstop capacity, $D(M)$ admitted financed exposure, and $h(M)$ expected shortfall rate per admitted unit, with $D'(M)\ge 0$ and $h'(M)\le 0$. Aggregate expected shortfall is
\begin{equation}
\mathcal L(M)=D(M)h(M).
\label{eq:capacity-shortfall}
\end{equation}

\begin{proposition}[Capacity-induced shortfall]
\label{prop:capacity-induced-shortfall}
Even when additional bonded capacity reduces unit shortfall, aggregate shortfall increases whenever
\begin{equation}
\frac{D'(M)}{D(M)}> -\frac{h'(M)}{h(M)}.
\label{eq:capacity-paradox-condition}
\end{equation}
\end{proposition}

This condition is the execution-capacity analogue of risk compensation. A protocol that admits more debt because a market maker posts a modest bond can end with no improvement in realized shortfall. In the registered agent experiment, the bonded-MM hypothesis fails: normalized shortfall during strategic-withdrawal episodes is 4.2 percent worse in the adaptive configuration than in the segmented configuration. The failure is reported rather than tuned away.

\subsection{Admission implication}
Bonded capacity should be credited through a haircut and a concentration limit:
\begin{equation}
M_t^{\mathrm{credit}}=\sum_k \omega_{k,t}
\min\{m^b_{k,t},\,b_{k,t}/\sigma_{k,t}\},
\qquad 0\le\omega_{k,t}\le 1,
\label{eq:mm-admission-credit}
\end{equation}
where $\sigma_{k,t}$ converts bond value into settlement-value coverage and $\omega_{k,t}$ reflects delivery history, asset quality, and correlation with the protected risk. The admission controller must then limit incremental credit so that backstop-induced demand does not exceed the protected execution increment.

\section{Withdrawal Runs, Dynamic Reserves, and Common-Factor Contagion}
\label{sec:runs-reserves}

Credit pools transform liquid provider claims into capital deployed against positions that may not be immediately unwindable. Even when queued shares remain loss participating, expectations about delay can generate coordinated withdrawal demand. Reserves and pool isolation mitigate some channels but cannot eliminate common-factor risk.

\subsection{Withdrawal response map}
Let $x_{i,t}\in[0,1]$ be provider $i$'s requested-withdrawal fraction. A reduced-form response is
\begin{equation}
x_{i,t+1}=G_i\!\left(
 o_i-y^S_t+\lambda_i\mu^S_t+\kappa_i q_t+\eta_i\bar x_t
\right),
\label{eq:withdrawal-response}
\end{equation}
where $\bar x_t$ is the aggregate observed withdrawal rate. Aggregating gives a map
\begin{equation}
\bar x_{t+1}=F(\bar x_t).
\label{eq:aggregate-run-map}
\end{equation}

\begin{proposition}[Stable withdrawal state]
\label{prop:run-contraction}
If $F$ maps $[0,1]$ into itself and $\sup_x|F'(x)|<1$, the withdrawal state is unique and globally stable under iteration.
\end{proposition}

If peer sensitivity is sufficiently large that $F'$ exceeds one over part of the state space, multiple fixed points or run-like transitions can arise. The queue removes request-time priority over known losses, but it does not remove coordination around future delay or future impairment.

\subsection{Loss-participating queues}
The r0.2.1 mechanism keeps queued shares outstanding until payment. If two equal providers have identical shares immediately before a loss, one cannot escape that loss merely by requesting withdrawal one instant earlier.

\begin{proposition}[Request-time neutrality]
\label{prop:request-time-neutrality}
Under a loss-participating queue, a withdrawal request that does not finalize payment before a common NAV loss leaves the requester's post-loss claim equal to the claim of an otherwise identical non-requesting provider.
\end{proposition}

The deterministic agent fixture obtains an absolute request-time advantage of zero at numerical precision. The result does not promise immediate liquidity; it removes one specific first-exit transfer.

\subsection{Risk-sensitive reserve target}
Let $R_t^m$ and $R_t^p$ be market and pool reserves, $B_t$ outstanding credit, and $\widehat{\mathrm{ES}}_t$ a registered expected-shortfall proxy. Define
\begin{equation}
R_t^*=\alpha B_t+\beta\widehat{\mathrm{ES}}_t.
\label{eq:dynamic-reserve-target}
\end{equation}
Reserve evolution is
\begin{equation}
R_{t+1}=\pospart{R_t-d_t+\rho_t I_t+s_t},
\label{eq:reserve-evolution}
\end{equation}
where $d_t$ is reserve draw, $I_t$ realized interest, $\rho_t$ the reserve share, and $s_t$ slashing receipts. The reference controller increases $\rho_t$ by a bounded amount when $R_t<R_t^*$ and funds the increment from Senior and protocol revenue shares.

\begin{theorem}[Reserve replenishment without new loss]
\label{thm:reserve-replenishment}
Suppose $R_t<R^*$, $d_t=0$ after time $t_0$, and the controller guarantees $\rho_t I_t+s_t\ge\epsilon>0$ until the target is reached. Then $R_t$ reaches $R^*$ in at most
\begin{equation}
\left\lceil\frac{R^*-R_{t_0}}{\epsilon}\right\rceil
\label{eq:reserve-replenishment-bound}
\end{equation}
periods.
\end{theorem}

A reserve target is not a solvency theorem. Repeated correlated losses can outrun replenishment, and high reserve shares can make provider participation uneconomic.

\subsection{Direct and common-factor contagion}
Strict pool isolation prevents a local liability from being transferred to another pool. It does not remove common exposures. Let loss in pool $p$ be
\begin{equation}
L_{p,t}=\beta_p Z_t+\varepsilon_{p,t},
\label{eq:common-factor-loss}
\end{equation}
where $Z_t$ is a shared stablecoin, venue, oracle, or infrastructure shock and $\varepsilon_{p,t}$ is idiosyncratic.

\begin{theorem}[Common-factor covariance under isolation]
\label{thm:common-factor-contagion}
If $\varepsilon_p$ and $\varepsilon_q$ are mutually independent and independent of $Z$, then for $p\neq q$,
\begin{equation}
\operatorname{Cov}(L_p,L_q)=\beta_p\beta_q\operatorname{Var}(Z).
\label{eq:common-factor-covariance}
\end{equation}
Ledger isolation eliminates direct liability transfer, but loss correlation remains positive whenever both pools load on the same non-degenerate factor with the same sign.
\end{theorem}

A shared global reserve adds a second channel: one pool's draw reduces protection available to another. This is protection contagion, not liability contagion. The agent simulation uses stablecoin shocks as a negative control and finds cross-pool raw-shortfall correlation of 0.842 despite isolated local ledgers.

\section{Agent-Based Economic Validation Design}
\label{sec:agent-design}

The fixed-distribution experiments in \Cref{sec:simulation-design} hold capital and participation rules largely exogenous. This revision adds a fixed-seed agent-based model in which traders, Senior LPs, LBPs, market makers, and liquidators respond to protocol prices, losses, queues, and one another. Agent-based methods are appropriate where heterogeneous decision rules and feedback can generate aggregate states not represented by a single equilibrium agent \citep{tesfatsion2006,farmer2009,brock1998,lux1999}. The model remains synthetic and is not calibrated to Axient or venue observations.

\newcommand{\AgentRuns}{96}
\newcommand{\AgentWeeks}{104}
\newcommand{\AgentInvariantChecks}{70,207,488}
\newcommand{\AgentHypothesesPassed}{6}
\newcommand{\AgentHypothesesTotal}{7}

\subsection{Population and horizon}
Each path contains three isolated pools (macro, election, and sports). Each pool contains 64 Senior LPs, 24 LBPs, 12 market makers, 16 liquidators, and 96 traders: 636 heterogeneous agents per path. The release executes \AgentRuns paths over \AgentWeeks weekly periods for four protocol configurations and three leverage policies. Agent endowments, outside rates, risk aversion, quote reservation values, execution costs, and trader signals are fixed-seed heterogeneous draws.

\begin{table}[H]
\centering
\caption{Agent population per pool and decision variable.}
\label{tab:agent-population}
\begin{tabularx}{\textwidth}{l r X}
\toprule
Role & Count & Endogenous decision\\
\midrule
Senior Credit LP & 64 & deposit, remain, or queue withdrawal from expected net yield, loss, delay, and peer withdrawals\\
LBP & 24 & supply or release junior capital from expected premium, impairment, and encumbrance\\
Market maker & 12 & ordinary quote fraction, bonded commitment, delivery, or strategic withdrawal\\
Liquidator & 16 & enter and supply executable capacity when bounty exceeds private cost\\
Trader & 96 & participate and select $2\times$, $3\times$, or $5\times$ subject to policy and risk gates\\
\bottomrule
\end{tabularx}
\end{table}

\subsection{State cycle}
Each weekly transition applies the following order:
\begin{enumerate}[leftmargin=*]
\item update a common three-state regime (normal, stress, crisis) and draw common and idiosyncratic shocks;
\item accrue debt interest and identify maturing positions;
\item obtain ordinary and bonded market-making capacity and liquidator participation;
\item realize settlement recovery and apply the exact loss waterfall;
\item distribute settled interest and update risk-sensitive reserve targets;
\item apply Senior and LBP participation responses and service the loss-participating queue;
\item solve the damped utilization-rate fixed point, calculate tier demand, and admit new credit subject to cash, utilization, and execution capacity;
\item register state invariants and output metrics.
\end{enumerate}
The complete reference algorithm is in \Cref{app:agent-model}.

\subsection{Configurations}
\begin{table}[H]
\centering
\caption{Agent-based protocol configurations.}
\label{tab:agent-configurations}
\begin{tabularx}{\textwidth}{l X}
\toprule
Configuration & Included mechanisms\\
\midrule
Senior only & static rate, Senior capital only, no LBP, no reserve, no bonded MM, non-loss-participating queue\\
Segmented & local reserves, LBP capital, loss-participating queue, static rate\\
Adaptive & segmented layers plus utilization rate, dynamic reserve target, bonded MM capacity\\
Full Axient & adaptive design plus risk-sensitive leverage gates and class-specific leverage caps\\
\bottomrule
\end{tabularx}
\end{table}

The leverage policies are:
\begin{itemize}[leftmargin=*]
\item $2\times$ only;
\item phased weights $(0.55,0.30,0.15)$ over $(2\times,3\times,5\times)$; and
\item a $5\times$-heavy mix $(0.25,0.25,0.50)$.
\end{itemize}
Sports is capped at $3\times$ in every configuration. The full configuration additionally compresses admissible leverage in stress, crisis, and high-close-pressure states.

\subsection{Shock structure}
A Markov regime controls the intensity of common venue incidents, stablecoin depegs, oracle delays, strategic MM withdrawal, and idiosyncratic venue failures. The same fixed shock paths and agent traits are reused across protocol configurations, creating paired comparisons. Common stablecoin shocks load on every pool and provide a negative control for the proposition that isolation alone eliminates economic contagion.

\subsection{Release-registered hypotheses}
The seven hypotheses and thresholds were written to the versioned parameter file before the final execution. They are release-registered, not externally preregistered.

\begin{table}[H]
\centering
\caption{Agent-based hypotheses.}
\label{tab:agent-hypothesis-definitions}
\begin{tabularx}{\textwidth}{l X}
\toprule
ID & Hypothesis\\
\midrule
H1 & Full Axient reduces phased-mix Senior-loss incidence by at least 35 percent relative to Senior only.\\
H2 & Full Axient preserves at least 70 percent of Senior-only accepted trader credit demand in the phased mix.\\
H3 & Median normal-regime LBP participation is at least 60 percent of endowed LBP capital.\\
H4 & Bonded MM support reduces normalized shortfall during strategic-withdrawal episodes by at least 15 percent relative to segmented.\\
H5 & Loss-participating queued shares eliminate request-time first-exit advantage within $10^{-10}$.\\
H6 & Common stablecoin shocks preserve cross-pool raw-shortfall correlation above 0.40 despite ledger isolation.\\
H7 & Under full Axient, $5\times$-heavy Senior-loss incidence exceeds $2\times$-only incidence by at least 25 percent.\\
\bottomrule
\end{tabularx}
\end{table}

\subsection{Verification scope}
The execution records \AgentInvariantChecks scalar invariant evaluations with zero failures. The checks cover non-negative balances, admission capacity, waterfall conservation, queue bounds, debt and gross-position state, and reserve allocation. They establish internal state consistency of the specified simulator. They do not validate behavioral rules, distributions, or external-market realism.

\section{Agent-Based Results}
\label{sec:agent-results}

Six of seven release-registered hypotheses pass. The results support the value of layered capital and risk-sensitive leverage admission, but they do not validate a commercially safe or capital-light protocol. The strongest negative findings are economically important: LBP capital is impaired frequently, Senior loss remains material over two-year paths, modest bonded-MM capacity does not improve strategic-withdrawal shortfall once endogenous demand is admitted, and common stablecoin exposure creates high cross-pool correlation despite local isolation.

\subsection{Phased leverage comparison}
\begin{table}[H]
\centering
\scriptsize
\caption{Agent-based results under the phased leverage policy. ``Senior loss incidence'' is the fraction of 104-week paths with any Senior impairment; it is not a weekly or annual default probability. Returns and losses are synthetic.}
\label{tab:agent-phased-results}
\begin{tabular}{lrrrr}
\toprule
Configuration & Senior return & Path loss incidence & Accepted credit & LBP impairment\\
\midrule
Senior only & 4.3\% & 99.0\% & 23.3\% & 0.0\% \\
Segmented & 10.4\% & 64.6\% & 24.6\% & 96.9\% \\
Adaptive & 17.1\% & 69.8\% & 35.6\% & 97.9\% \\
Full Axient & 20.6\% & 57.3\% & 38.0\% & 94.8\% \\
\bottomrule

\end{tabular}
\end{table}

The full configuration reduces path-level Senior-loss incidence from 99.0 percent to 57.3 percent, a 42.1 percent relative reduction and a pass for H1. Mean cumulative Senior impairment falls from 40.8 percent to 17.4 percent of mean deployed Senior capital. At the same time, only 38.0 percent of requested credit is admitted: capital and execution constraints reject most synthetic demand. H2 passes because endogenous provider participation and dynamic pricing produce more accepted credit than the Senior-only baseline, but that comparison should not be read as unconstrained growth.

The mean synthetic annualized Senior return rises to 20.6 percent in the full configuration. This is not an APY forecast. It is generated by author-specified trader edges, rates, utilization, and losses. The same configuration still has Senior impairment in 57.3 percent of two-year paths. High modeled yield is compensation for high modeled risk, not evidence of safety.

LBP impairment occurs in 94.8 percent of full phased paths. The mean cumulative LBP draw equals 3.75 times mean deployed LBP capital because impaired capital can be replenished and impaired again over 104 weeks. This is a turnover measure, not a one-time loss fraction. It implies that the junior layer is expensive and cannot be treated as passive insurance capital.

\begin{figure}[H]
\centering
\includegraphics[width=0.78\textwidth]{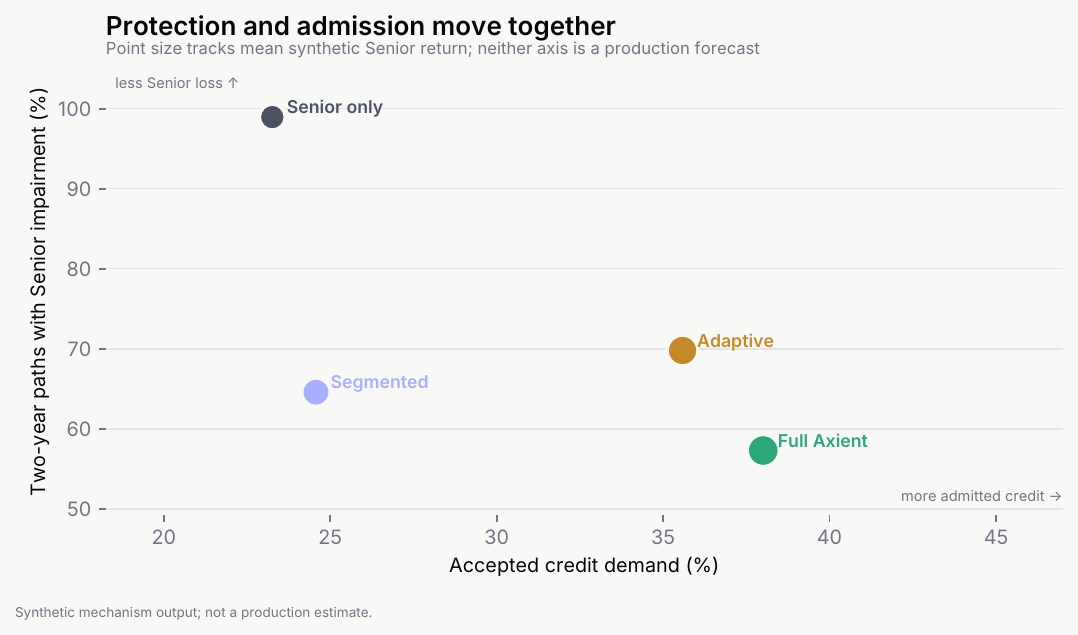}
\caption{Protection-demand frontier under the phased leverage policy. Both axes are synthetic outputs.}
\label{fig:agent-risk-demand}
\end{figure}

\subsection{Leverage sensitivity}
\begin{table}[H]
\centering
\scriptsize
\caption{Full-Axient results by leverage policy.}
\label{tab:agent-leverage-results}
\begin{tabular}{lrrrr}
\toprule
Policy & Senior return & Path loss incidence & Accepted credit & LBP impairment\\
\midrule
$2\times$ only & 17.5\% & 22.9\% & 59.3\% & 45.8\% \\
Phased & 20.6\% & 57.3\% & 38.0\% & 94.8\% \\
$5\times$ heavy & 20.9\% & 59.4\% & 34.5\% & 95.8\% \\
\bottomrule

\end{tabular}
\end{table}

Moving from $2\times$ only to the $5\times$-heavy mix raises Senior-loss incidence from 22.9 percent to 59.4 percent, a 159.1 percent relative increase and a pass for H7. Accepted-demand ratio falls from 59.3 percent to 34.5 percent. The result reflects both higher debt fraction and lower gross exposure per unit of credit. The simulation therefore supports a phased product policy: $2\times$ is not merely a conservative marketing launch; it is a materially different capital regime.

\begin{figure}[H]
\centering
\includegraphics[width=0.76\textwidth]{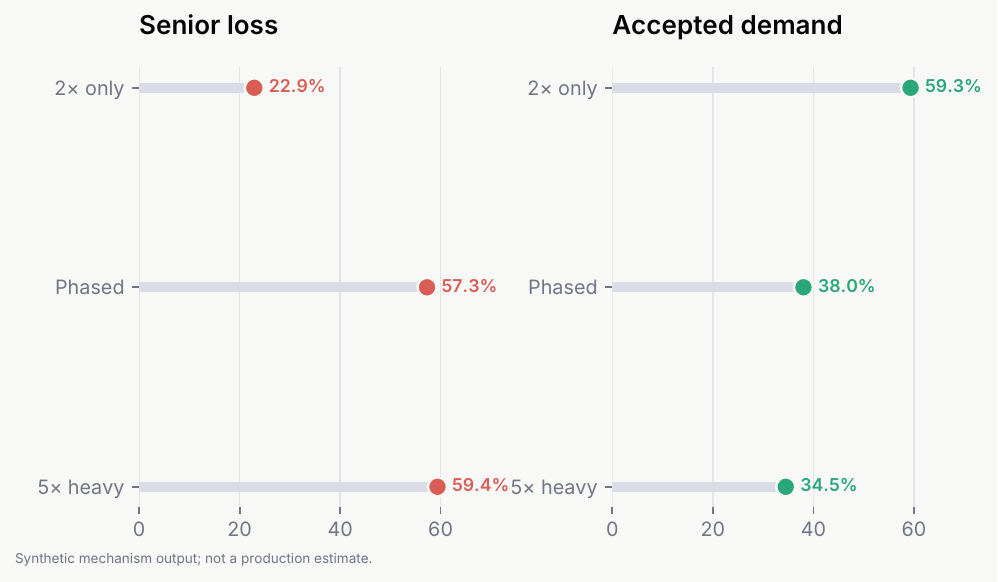}
\caption{Leverage-policy trade-off under the full configuration.}
\label{fig:agent-leverage}
\end{figure}

\subsection{Strategic market-maker withdrawal}
H4 fails. The adaptive configuration's normalized shortfall in strategic-withdrawal weeks is 4.3 percent worse than segmented, against a required 15 percent improvement. Bonded capacity improves delivery per unit, but it also expands the admission envelope. In the registered calibration, demand expansion dominates the direct protection effect, consistent with \Cref{prop:capacity-induced-shortfall}.

This negative result changes the design interpretation. A market-maker bond should not be credited one-for-one as reserve substitution or used automatically to expand open interest. Incremental admission must be limited by the stressed delivered capacity, concentration, bond asset quality, and the elasticity condition in \eqref{eq:capacity-paradox-condition}.

\begin{figure}[H]
\centering
\includegraphics[width=0.74\textwidth]{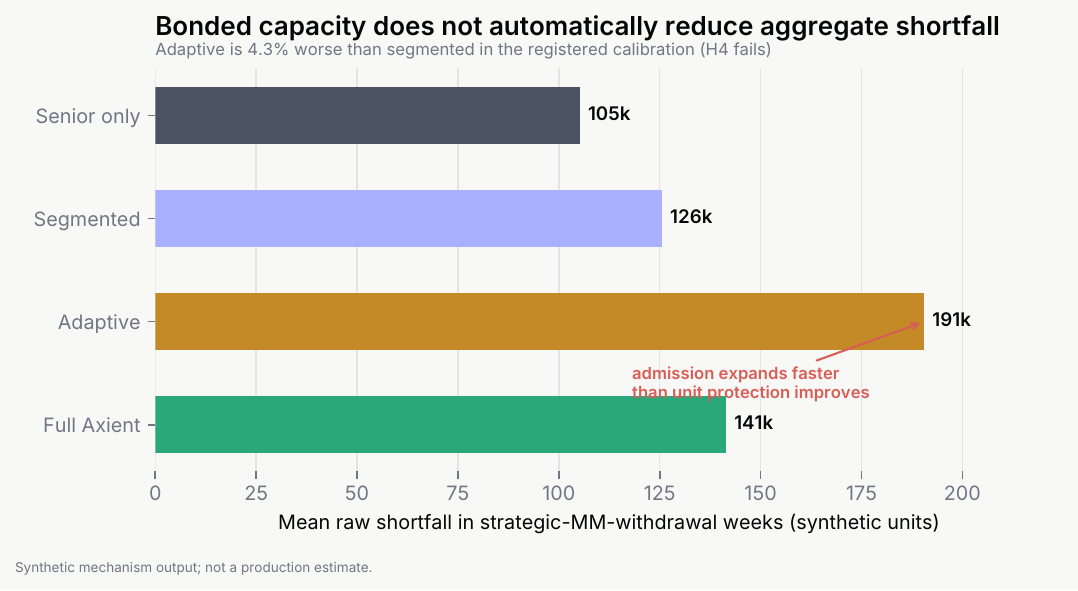}
\caption{Mean raw shortfall during registered strategic-MM-withdrawal weeks.}
\label{fig:agent-mm-withdrawal}
\end{figure}

\subsection{Runs and LBP participation}
The loss-participating queue deterministic fixture yields zero request-time advantage and H5 passes. Aggregate run-week incidence falls from 4.1 percent in Senior only to 2.6 percent in full phased. This does not show that runs are eliminated: queue delay and perceived loss remain coordination inputs, and crisis-state queue delay is much higher than normal-state delay.

Normal-regime median LBP participation is 68.7 percent of endowed LBP capital, passing H3. Participation falls materially in crisis states. The LBP result is therefore two-sided: enhanced compensation attracts capital in ordinary states, but repeated impairment and encumbrance reduce capital precisely when it is most valuable.

\begin{figure}[H]
\centering
\includegraphics[width=0.76\textwidth]{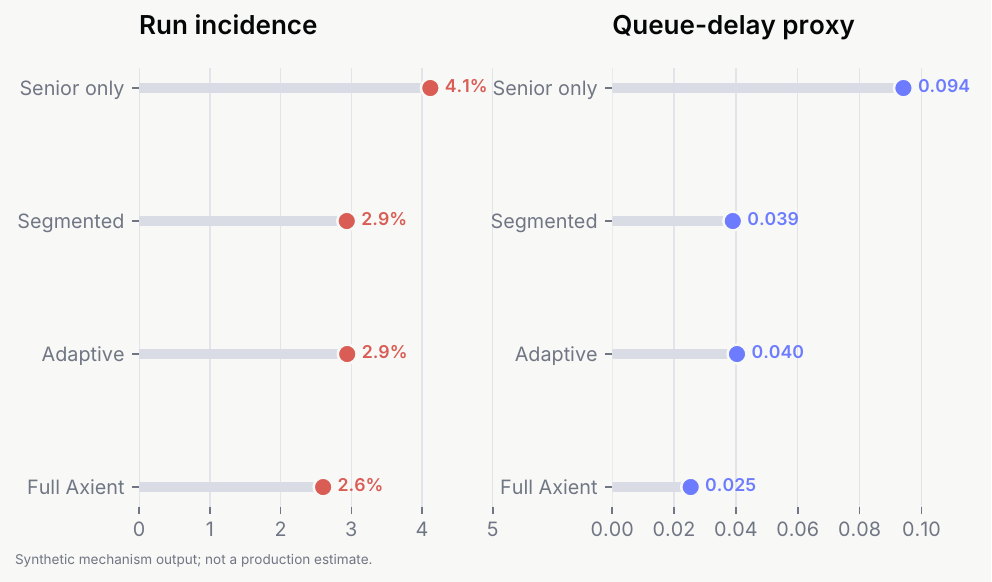}
\caption{Run-week incidence and mean queue-delay proxy across configurations under phased leverage. The panels use different units and should not be read on a common scale.}
\label{fig:agent-runs}
\end{figure}

\subsection{Regime response and wrong-way pricing}
\begin{table}[H]
\centering
\scriptsize
\caption{Full-Axient phased results by synthetic regime. Senior impairment is mean settlement units per pool-week.}
\label{tab:agent-regime-results}
\begin{tabular}{lrrrrr}
\toprule
Regime & Utilization & Borrow rate & Accepted credit & Senior impairment & Liquidation coverage\\
\midrule
Normal & 75.1\% & 20.5\% & 36.6\% & 1.4k & 99.2\% \\
Stress & 66.5\% & 19.1\% & 46.2\% & 6.4k & 92.1\% \\
Crisis & 61.3\% & 24.3\% & 41.5\% & 67.2k & 74.7\% \\
\bottomrule

\end{tabular}
\end{table}

The stress regime has lower utilization and a lower mean borrow rate than the normal regime while impairment rises. The utilization-only rate responds to reduced admitted demand rather than directly to higher expected loss. Crisis rates rise because capital supply and utilization react differently in the most severe state, but the non-monotone pattern is enough to reject utilization as a complete risk price. A separate risk spread satisfying \eqref{eq:risk-spread-condition} is required.

Liquidator participation is nearly saturated in the registered model: mean active liquidators are 15.9 of 16 and mean coverage is 97.8 percent in the full phased configuration. This should not be interpreted as empirical abundance. It indicates that the registered bounty is generous relative to synthetic costs and may overpay execution agents.

\subsection{Pool differences and leverage caps}
\begin{table}[H]
\centering
\scriptsize
\caption{Full-Axient phased results by isolated pool. Senior-loss incidence is per pool-week.}
\label{tab:agent-pool-results}
\begin{tabular}{lrrrrr}
\toprule
Pool & Utilization & Accepted demand & Senior loss incidence & Raw shortfall rate & LBP participation\\
\midrule
Macro & 75.7\% & 32.2\% & 0.76\% & 0.29\% & 55.5\% \\
Election & 75.3\% & 32.9\% & 0.93\% & 0.33\% & 51.8\% \\
Sports & 70.0\% & 48.9\% & 0.36\% & 0.18\% & 55.0\% \\
\bottomrule

\end{tabular}
\end{table}

Sports has the highest registered base event-risk multiplier but the lowest realized shortfall and Senior-loss incidence because it is capped at $3\times$. Macro and election pools admit $5\times$ in eligible states and experience higher loss. This is not empirical evidence that sports is safer; it shows that class-specific leverage policy can dominate a class's raw shock parameter in the simulated loss distribution.

\subsection{Common-factor contagion}
Conditional on stablecoin shocks, mean cross-pool raw-shortfall correlation is 0.842, passing H6. Strict local ledgers prevent one pool's debt from becoming another pool's liability, yet all pools lose together when their settlement asset is impaired. A global reserve can redistribute protection across the same event, but cannot remove the common factor.

\subsection{Hypothesis outcomes}
\begin{table}[H]
\centering
\caption{Release-registered hypothesis outcomes.}
\label{tab:agent-hypothesis-results}
\begin{tabular}{lrrl}
\toprule
ID & Realized & Threshold & Verdict\\
\midrule
H1 & 42.1\% & 35.0\% & Pass \\
H2 & 155.6\% & 70.0\% & Pass \\
H3 & 68.7\% & 60.0\% & Pass \\
H4 & -4.3\% & 15.0\% & Fail \\
H5 & 0.00e+00 & $\le 1e-10$ & Pass \\
H6 & 84.2\% & 40.0\% & Pass \\
H7 & 159.1\% & 25.0\% & Pass \\
\bottomrule

\end{tabular}
\end{table}

\begin{figure}[H]
\centering
\includegraphics[width=0.76\textwidth]{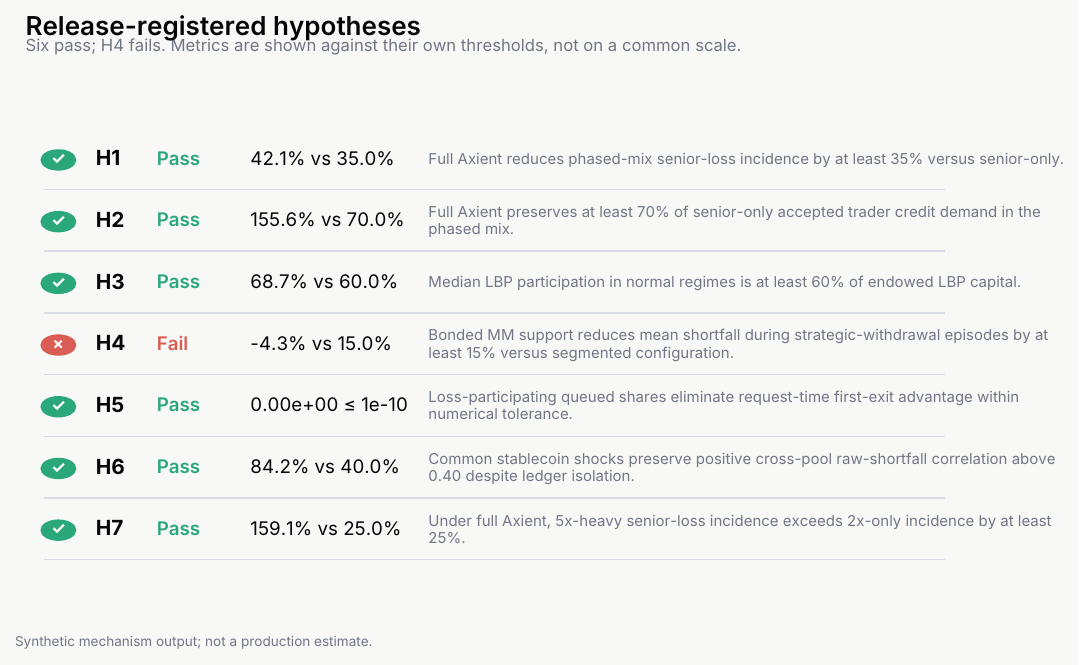}
\caption{Six hypotheses pass and H4 fails. Values use different metric scales; the figure is a compact threshold panel, not a cross-metric magnitude comparison.}
\label{fig:agent-hypotheses}
\end{figure}

\subsection{Independent daily robustness implementation}
The independently coded daily model in \Cref{app:agent-diagnostics} does not contribute to the H1--H7 verdicts. It provides a second mechanism implementation with different horizon, pool structure, and scenario rules. Three qualitative results survive that change: layered protection reduces material Senior credit write-downs inside the registered path set; a combined venue-and-withdrawal shock can produce negative conservative provider-wealth returns and non-zero run incidence even when finalized Senior principal is protected; and increasing the $5\times$ demand weight lowers hard-flat success and raises rejection. The daily queue sweep also displays a sharp run phase transition as peer sensitivity rises and immediately serviceable cash falls. Absolute values are not pooled across models.

\subsection{Independent daily robustness model}
The independently coded daily model in \Cref{app:agent-diagnostics} is not pooled with the weekly evidence. Its ordinary panel finds that layered policies eliminate the registered material-Senior-loss threshold within that particular grid while admitting more demand than the untranched benchmark. Under the combined venue-and-withdrawal scenario, however, conservative provider wealth returns are negative and run incidence is non-zero even where finalized Senior principal is not impaired. Its withdrawal sensitivity also exhibits a sharp buffer-dependent phase transition, and increasing the $5\times$ share lowers hard-flat success. The qualitative overlap supports the mechanism directions; the materially different levels demonstrate that neither implementation is a calibrated forecast.

\subsection{Interpretive boundary}
The agent model does not establish production probabilities. It establishes that the specified behavior rules produce several mechanisms worth carrying into design review: endogenous capital reduces but does not eliminate Senior loss; LBP capital is repeatedly impaired; $5\times$ materially raises capital risk; strategic MM support can induce offsetting demand; utilization-only pricing can be wrong-way; and pool isolation does not remove shared-asset contagion. Those are mechanism-design findings conditional on the registered synthetic environment.

\section{Smart-Contract Guarantees and External Assumptions}
\label{sec:contract-guarantees}

The protocol's DeFi claim should be evaluated by which financial facts are enforced by contracts, not by whether every computation occurs on-chain.

\begin{table}[H]
\centering
\small
\caption{Guarantee boundary of the reference architecture.}
\label{tab:guarantee-boundary}
\begin{tabularx}{\textwidth}{>{\raggedright\arraybackslash}X>{\raggedright\arraybackslash}X}
\toprule
Can be enforced or made publicly auditable on-chain & Remains an external assumption or failure channel\\
\midrule
LP deposits, virtual-offset shares, loss-participating queue claims, and withdrawal payments & future order-book and RFQ liquidity\\
Debt shares, interest index, partial-settlement idempotency, and repayment priority & correctness of a real-world event outcome\\
Collateral lock, debt-first allocation, and trader residual release & liveness and solvency of an external venue\\
Pool, LBP, non-redeemable reserve, and treasury segregation & bridge, stablecoin, or chain impairment\\
Exposure, leverage, and concentration caps & off-chain data availability and latency\\
Bonded-MM collateral and deterministic slashing & non-collateralized MM willingness to quote\\
Loss waterfall and share-price write-down & absence of smart-contract, governance, or key compromise\\
Versioned adapter and mandate registry & legal enforceability outside the contract boundary\\
\bottomrule
\end{tabularx}
\end{table}

\subsection{Trust ladder}
The reference architecture distinguishes four deployment levels:
\begin{enumerate}[leftmargin=*]
\item \textbf{T0: research simulator.} No external capital and no claim of financial enforcement.
\item \textbf{T1: on-chain pool with controlled venue account.} LP accounting, debt, reserves, and waterfall are on-chain, but position custody or signing depends on an operator-controlled account.
\item \textbf{T2: verifiable adapter.} Trade settlement and position ownership are publicly verifiable; a controlled signer or operator may still be required for execution.
\item \textbf{T3: contract-controlled margin account.} The protocol has enforceable lien, withdrawal lock, and delegated liquidation over recognized position assets, and off-chain agents are replaceable.
\end{enumerate}
T1 and T2 can be physically backed and settlement-verifiable without being fully non-custodial. Only T3 supports the strongest trust-minimization claim. A venue-independent architecture allows pools in different modes, but it may not describe them as economically equivalent.

\subsection{On-chain source of truth}
For each pool, the authoritative values of cash, debt shares, LP shares, LBP shares, reserve balances, withdrawal claims, bonds, and finalized losses are contract state. Off-chain databases are indexes and caches. A disagreement is resolved by the chain and the versioned adapter evidence, not by a privileged database write.

\subsection{Replaceability criterion}
An off-chain function is replaceable if an independent agent can reconstruct the required inputs from public or authenticated data, submit the same admissible transaction, and obtain the same contract result. Quote generation need not be deterministic across agents; accounting and validity of the executed transition must be.

\subsection{Protocol-native synthetic mode}
Venue independence may ultimately include protocol-native synthetic exposure hedged across several venues. That mode changes the risk boundary: the user claim is no longer fully represented by one set of outcome tokens, and Axient becomes the direct counterparty to gross or net exposure. It must therefore use a separate pool, reserve, valuation rule, and disclosure. This paper permits the mode in the taxonomy but does not treat it as equivalent to native physical backing.

\subsection{Protocol-boundary permissionlessness}
Permissionless liquidation is a statement about who may submit a transition to Axient contracts, not a claim that every external venue accepts arbitrary callers. The adapter may require a contract account, delegated signer, controlled account, or venue operator. A pool can describe liquidation as permissionless only at the boundary where the condition is actually enforced, and must disclose any external authorization dependency.

\subsection{Tagged reference implementation versus production guarantee}
The repository-linked r0.2.2 release identifies a private tagged implementation that imports the decimal and integer fixture suites, exposes the normative claim and invariant registries, and exercises reference Solidity modules through unit, bounded-fuzz, stateful-invariant, and Python-to-Solidity differential tests. This satisfies the paper's \emph{reference-code verified} boundary for the identified snapshot. It is not a formal proof of bytecode, an independent audit, evidence that any external adapter reports truthful settlement, or a production deployment. Research status, reference-code status, external-audit status, and production status are therefore recorded as separate dimensions.

\subsection{No regulatory inference}
On-chain custody, transparent priority, or permissionless liquidation can reduce operator discretion and improve auditability. They do not by themselves determine whether a product is regulated as lending, derivatives, collective investment, gaming, or another activity. Legal characterization is outside the formal guarantee set.

\section{Security Invariants and Threat Model}
\label{sec:security}

A reference implementation should expose each invariant as contract assertions, stateful fuzz properties, or formal-verification targets where feasible.

\subsection{Accounting invariants}
\begin{enumerate}[leftmargin=*]
\item \textbf{Book closure.} Recognized local assets equal recognized claims after every finalized transition, subject to explicit suspense accounts.
\item \textbf{Settled-only repayment.} Pending or reverted execution cannot burn debt shares.
\item \textbf{Share conservation.} Deposits and withdrawals mint or burn shares at the applicable class price; no privileged path mints unbacked claims.
\item \textbf{Fee conservation.} Allocated interest and fees equal collected amounts.
\item \textbf{Loss visibility.} Every debt write-down produces an equal waterfall allocation and a public share-price effect.
\end{enumerate}

\subsection{Control invariants}
\begin{enumerate}[leftmargin=*,resume]
\item \textbf{No collateral escape.} While debt is positive, recognized position assets cannot leave except through an admissible transition that preserves debt priority.
\item \textbf{Prospective governance.} Existing position priority and queue order cannot be changed retroactively.
\item \textbf{Pause monotonicity.} Emergency transitions can reduce permissions but cannot create new borrowing or release encumbered collateral.
\item \textbf{Role separation.} A liquidator does not acquire LBP ownership merely by executing; an LBP does not receive liquidation bounty without a valid action.
\end{enumerate}

\subsection{Capacity invariants}
\begin{enumerate}[leftmargin=*,resume]
\item \textbf{No liquidity double counting.} Shared venue depth, RFQ capacity, or MM bond cannot certify more simultaneous exits than its aggregate allocation.
\item \textbf{Bond cap.} Credited MM protection cannot exceed enforceable collateral or atomic delivery value.
\item \textbf{Reserve and LBP separation.} Market, pool, and global reserves are non-redeemable protocol equity accounts in the reference mechanism; LBP capital is a distinct share class. All are excluded from lendable senior cash while counted as protection.
\item \textbf{Admission conjunction.} No position opens unless every account, venue, execution, time, concentration, capital, and jurisdiction gate passes.
\item \textbf{Isolation.} A local pool adapter cannot debit another pool or use another pool's collateral unless an explicit shared layer is invoked.
\end{enumerate}

\subsection{Liquidity and finality invariants}
\begin{enumerate}[leftmargin=*,resume]
\item \textbf{Withdrawal honesty.} Immediately claimable cash never exceeds recognized unencumbered withdrawal cash.
\item \textbf{Queue share lock and loss participation.} Queued shares cannot be transferred or simultaneously redeemed elsewhere and continue to participate in NAV changes until payment.
\item \textbf{Finality separation.} Match, settlement, oracle finality, and redemption are different states; only the state relevant to a claim can finalize that claim.
\item \textbf{Explicit failure.} Premature venue close, persistent settlement failure, signer loss, or invalid adapter evidence creates an exception state rather than a fabricated successful close.
\end{enumerate}

\subsection{Implementation-parity invariant identifiers}
The companion protocol registry extends the original bootstrap identifiers with the following normative mappings:
\begin{table}[H]
\centering
\scriptsize
\caption{r0.2.1 implementation-parity invariant IDs.}
\label{tab:r021-invariants}
\begin{tabularx}{\textwidth}{lX}
\toprule
ID & Invariant\\
\midrule
AX-INV-011 & trader residual after debt and approved fees cannot be redirected\\
AX-INV-012 & settlement evidence is unique, cumulative, and idempotent\\
AX-INV-013 & accrued receivable and settled cash are distinct\\
AX-INV-014 & protocol reserve is finite and non-redeemable by LPs\\
AX-INV-015 & queued shares participate in NAV changes until paid\\
AX-INV-016 & integer rounding residuals are explicit and bounded\\
AX-INV-017 & virtual offsets and minimum deposit prevent zero-share admission\\
AX-INV-018 & simultaneous shared-reserve allocation is exact and deterministic\\
\bottomrule
\end{tabularx}
\end{table}

\subsection{Four-dimensional evidence status}
The repository-linked release uses separate status dimensions so that ``code verified'' cannot be mistaken for audited or deployed software.
\begin{table}[H]
\centering
\small
\caption{Evidence status of the tagged r0.2.2 research release.}
\label{tab:evidence-status}
\begin{tabularx}{\textwidth}{l l X}
\toprule
Dimension & Status & Meaning\\
\midrule
Research & verified & formal statements, proofs, registered parameters, and synthetic outputs are versioned in the manuscript release\\
Reference code & verified & the identified private snapshot imports the 28 fixtures, executes 31,082 deterministic checks, and includes unit, fuzz, invariant, and differential test boundaries\\
Independent audit & not performed & no third-party security opinion or formal-verification certificate is claimed\\
Production & not deployed & no live pool, external capital, venue adapter, or production bytecode is evaluated\\
\bottomrule
\end{tabularx}
\end{table}

\subsection{Threat classes}
The principal threats are:
\begin{itemize}[leftmargin=*]
\item contract reentrancy, arithmetic, authorization, and upgrade bugs;
\item oracle or adapter data forgery;
\item keeper censorship or sequencing manipulation;
\item MM bond evasion or correlated provider failure;
\item stablecoin depeg, bridge compromise, or chain reorganization;
\item governance capture and malicious parameter changes;
\item venue insolvency, premature closure, or settlement failure;
\item event manipulation and informed trading amplified by leverage.
\end{itemize}
The last class motivates market- and event-specific leverage caps developed in the broader research programme \citep{nechepurenko2026manipulation}. Smart contracts can enforce caps and transparent losses but cannot make the underlying event incorruptible.

\subsection{Verification programme}
Before real external LP capital, the minimum programme is:
\begin{enumerate}[leftmargin=*]
\item unit and exact-decimal parity tests against the paper fixtures;
\item stateful fuzzing of shares, debt, queues, and waterfall;
\item invariant testing across all cross-module transitions;
\item differential testing against an independent accounting implementation;
\item adapter-specific adversarial settlement tests;
\item third-party audit and public remediation record;
\item timelocked deployment with caps and emergency withdrawal procedures.
\end{enumerate}
The tagged reference harness completes items 1--4 at the research/reference boundary. Items 5--7 remain future deployment prerequisites. A passed audit would be evidence about a particular code version, not a guarantee of economic safety.

\section{Discussion and Design Implications}
\label{sec:discussion}

\subsection{The capital protocol is a market, not a static balance sheet}
The r0.2.2 accounting layer establishes who owns what and how realized losses are allocated. The endogenous extension shows that these identities do not determine how much capital remains available. Senior LPs respond to expected yield, impairment, queue delay, and peers; LBPs respond to a more severe impairment distribution and lock cost; traders respond to rates and leverage gates; market makers and liquidators respond to compensation and adverse conditions. The protocol therefore needs two kinds of correctness: transition correctness on-chain and participation robustness at the economic layer.

\subsection{High synthetic yield is not evidence of safety}
The full phased configuration produces a mean synthetic annualized Senior return of 20.6 percent while Senior principal is impaired in 57.3 percent of 104-week paths. Both values arise from the same author-specified rates, utilization, trader edge, and loss distributions. Yield and risk move together. A deck or interface that presents the return without the path-level impairment result would invert the meaning of the experiment. No public APY should be inferred before empirical calibration and legal, accounting, and security review.

\subsection{LBP capital is active risk capital}
LBP impairment occurs in 94.8 percent of full phased paths, and cumulative LBP draw exceeds initial average capital because providers can replenish and be impaired repeatedly. The junior layer is therefore not a passive insurance wrapper. It is scarce, expensive, encumbered provider capital whose compensation must be sufficient to sustain participation across states. If the protocol underprices that role, the layer disappears precisely when it is most valuable. If it overprices the role, Senior economics or protocol contribution margin may become unattractive.

\subsection{Bonded execution can expand risk}
H4 fails because modest bonded capacity expands the admission envelope. The result does not show that bonds are useless. It shows that execution capacity should be credited on a stressed delivered-value basis and that incremental debt admission must be smaller than the protected capacity increment after concentration and correlation haircuts. A market-maker partnership is therefore both an execution relationship and a credit-policy input; it is not automatically a reserve substitute.

\subsection{Utilization pricing can become wrong-way}
Stress can reduce admitted demand and available supply simultaneously. In the registered experiment, mean utilization and the utilization-only borrow rate are lower in the stress regime than in the normal regime despite higher impairment. This is the wrong-way mechanism in \eqref{eq:wrong-way-rate}. A production rate policy needs an explicit risk spread, reserve deficit term, or exposure-price component in addition to utilization. Otherwise the protocol can make credit cheaper when expected loss rises.

\subsection{Leverage policy is capital policy}
The $2\times$ and $5\times$ states are not adjacent settings of the same product. At $2\times$, half of gross exposure is trader equity; at $5\times$, only one fifth is. In the full synthetic configuration, shifting from $2\times$ only to a $5\times$-heavy mix raises Senior-loss incidence by 159.1 percent and lowers accepted demand. This supports a phased policy: $2\times$ is a distinct capital regime for early validation, $3\times$ is a gated extension, and up to $5\times$ should require both account and market eligibility plus demonstrated exit capacity.

\subsection{Class-specific caps can dominate raw event risk}
The sports pool has a higher raw event-risk calibration but a $3\times$ leverage cap. Its realized pool-week Senior-loss incidence is lower than macro and election in the full phased configuration. This is not evidence that sports is intrinsically safer. It shows that a binding class-specific leverage cap can dominate a raw shock parameter. Event integrity, insider exposure, and manipulation concerns remain separate eligibility inputs \citep{nechepurenko2026manipulation}.

\subsection{Isolation has three meanings}
Strict pool accounting prevents one pool's local liability from becoming another pool's debt. It does not eliminate common stablecoin, venue, oracle, or infrastructure shocks, and it does not prevent one pool from consuming a shared reserve before another. The 0.842 common-shock correlation is a deliberately designed negative control. Public claims should therefore distinguish ledger isolation, economic correlation, and protection contagion.

\subsection{Venue independence remains capability-based}
The capital core should not depend on one venue granting a bespoke margin vault. But a protocol-native account, a verifiable on-chain venue, an operator-gated venue, and a custodial claim cannot share one undifferentiated mandate. Each pool must disclose which adapter capabilities it accepts, which execution agents it depends on, and what evidence is sufficient for settlement-confirmed repayment. Venue independence is achieved by explicit capability classes and isolation, not by pretending every adapter provides the same guarantees.

\subsection{Research and product boundary}
Paper I establishes the position-level Debt-Free Finality mechanism. Paper II r0.2.2 establishes protocol accounting and implementation parity. This revision adds an endogenous synthetic capital market and exposes design tensions that static fixtures cannot reveal. None establishes live commercial viability. The next empirical programme should calibrate provider supply, LBP participation, MM delivery, execution deterioration, queue behavior, and common-factor losses on actual adapters before public deposits or APY claims are considered.

\section{Limitations and Empirical Agenda}
\label{sec:limitations}

\subsection{Synthetic rather than empirical behavior}
All new agent outputs come from author-specified decision rules and shock processes. Fixed seeds and paired paths improve reproducibility and comparative precision, but do not address model misspecification. The model is not fitted to Axient users, DeFi LPs, prediction-market market makers, or any particular venue. Its incidence, return, and participation values are not forecasts.

\subsection{Release registration is not external preregistration}
The seven hypotheses were written into a versioned parameter file before the final execution, and the failed H4 threshold was not modified after observation. That discipline reduces some researcher degrees of freedom but does not provide the institutional separation of a public preregistration. Future empirical work should preregister hypotheses and analysis plans externally before data collection or holdout evaluation.

\subsection{Behavioral-rule and learning limitations}
The agents use reduced-form response rules, fixed heterogeneous traits, and limited state memory. They do not learn strategic models of one another, form Bayesian beliefs about hidden venue state, optimize over long horizons, collude, fork governance, or migrate across protocols. Logistic participation and threshold entry are transparent, but other plausible rules can produce different equilibria. Global sensitivity and model-comparison exercises remain necessary.

\subsection{No full strategic equilibrium}
The paper provides sufficient conditions and an agent simulation, not a solved equilibrium among traders, Senior LPs, LBPs, market makers, liquidators, and governance. Bond size, commitment premium, ordinary depth, liquidator bounty, rates, reserve shares, and admission can all adjust jointly. The capacity-induced-shortfall finding is a comparative mechanism result, not an equilibrium characterization.

\subsection{External execution and settlement}
The protocol can reserve capital and enforce priorities but cannot force an external venue to quote, match, or settle. The agent model represents delivery and settlement through synthetic capacities and shocks. Real adapter semantics require venue-specific tests, chain evidence, operational service levels, and failure-recovery procedures.

\subsection{Withdrawal model}
The run map includes heterogeneous outside options, delay aversion, and herding, but omits secondary trading of pool shares, off-chain side agreements, strategic deposit timing, governance intervention, fee changes, and capital migration to competing protocols. Request-time neutrality holds only when valid loss recognition occurs before payment. Delayed or manipulated recognition can still transfer latent loss.

\subsection{Reserve and LBP replenishment}
The reserve theorem assumes a sustained positive inflow and no new draw until the target is restored. The agent model allows repeated losses but does not solve for an optimal reserve target, commitment fee, or LBP lock term. Repeated junior impairment may make the layer unavailable at any price or may require long-duration capital with materially different legal and economic treatment.

\subsection{Common-factor model}
Stablecoin, venue, oracle, and infrastructure shocks are represented by stylized common factors. The model does not include network topology among multiple stablecoins, rehypothecation, bridge exposure, cross-chain finality, correlated smart-contract bugs, or endogenous asset-price fire sales. The measured 0.842 correlation is a property of the registered generator, not an estimate for production pools.

\subsection{Market and manipulation feedback}
Trader demand and execution capacity respond to a risk signal, but the model does not generate a full endogenous event-price path or allow traders to manipulate the real-world outcome. Leverage can amplify informed-flow and outcome-manipulation incentives \citep{nechepurenko2026manipulation}; those channels require separate surveillance, eligibility, and empirical evaluation.

\subsection{Contract implementation}
The tagged private reference implementation covers registered accounting semantics, typed settlement evidence, admitted Senior fees, integrated loss allocation, fuzz tests, stateful invariants, and deterministic differential vectors. No production smart contracts, external venue adapter, independent security audit, or deployed bytecode are evaluated. Gas bounds, upgrade architecture, cross-chain design, governance hardening, formal verification coverage, and operational key risk remain engineering work.

\subsection{Legal and governance scope}
The paper is not legal advice. Permissionless submission does not eliminate regulation, contractual obligations, sanctions exposure, market-conduct rules, consumer-protection requirements, fiduciary questions, or licensing of lending and derivatives activity. Public, permissioned, and bilateral pools may have different requirements.

\subsection{Empirical calibration programme}
A future study should externally preregister and estimate:
\begin{itemize}[leftmargin=*]
\item Senior and LBP deposit, withdrawal, lock, and replenishment responses to realized yield and loss;
\item quote-to-match and match-to-settlement deterioration by venue, market class, and time to close;
\item MM ordinary-depth withdrawal, bonded commitment delivery, bond recovery, and concentration;
\item liquidator cost, entry, competition, latency, and realized execution capacity;
\item actual hard-flat completion, partial-fill, settlement-failure, and redemption-delay distributions;
\item run behavior, queue delay, and secondary liquidity under public stress disclosures;
\item common stablecoin, venue, oracle, bridge, and infrastructure shocks;
\item protocol contribution margin and provider returns under observed cash flows;
\item results separately for $2\times$, $3\times$, and $5\times$ tiers and by event class.
\end{itemize}
Calibration and out-of-sample validation must use versioned holdouts. Parameters should not be enlarged retrospectively merely to classify observed failures as inside the model.

\section{Conclusion}
\label{sec:conclusion}

A physically backed leveraged event position requires real stablecoin credit. Once external participants supply that credit, Axient becomes an on-chain capital protocol rather than only a trading interface. Its core problem is to connect event traders, Senior Credit LPs, market makers, liquidators, and Liquidation Backstop Providers without allowing an operator to rewrite debt, ownership, priority, or realized loss allocation in a private database.

The earlier releases formalize that protocol layer. They define Senior and junior shares, debt shares, utilization pricing, idempotent partial settlement, debt-first trader residual release, capability-based adapters, collateralized MM commitments, loss-participating queues, isolated pools, non-redeemable reserves, integer rounding, and a deterministic waterfall, and tie those semantics to a parity-tested private reference implementation.

This revision asks what happens when capital and execution participants react. The formal extension establishes existence and stability conditions for utilization equilibria, identifies wrong-way utilization pricing, specifies market-maker and liquidator participation constraints, proves a capacity-induced-shortfall condition, gives withdrawal-run and reserve-replenishment results, and separates direct liability contagion from common-factor and protection contagion. The fixed-seed agent experiment then compares four protocol configurations across three pools and three leverage policies.

The results are deliberately mixed. Six of seven registered hypotheses pass: layered protection reduces Senior-loss incidence, endogenous supply supports more accepted credit than the Senior-only benchmark, LBP participation remains material in ordinary states, queued shares do not gain request-time priority over a common loss, common factors preserve cross-pool correlation, and a $5\times$-heavy mix is materially more damaging than $2\times$ only. H4 fails: modest bonded market-making capacity expands admitted debt faster than it reduces unit shortfall. The full phased configuration still impairs Senior principal in 57.3 percent of two-year synthetic paths and LBP capital in 94.8 percent. Its high modeled return compensates for high modeled risk; it is not an APY forecast.

Three design implications follow. First, admission control must credit execution capacity conservatively and separately from cash reserves. Second, interest pricing requires an explicit risk component because utilization can fall under stress. Third, leverage tiers are capital regimes: $2\times$, $3\times$, and $5\times$ differ not only in trader experience but in credit intensity, exit fragility, junior-capital demand, and system-wide capacity.

The architectural conclusion remains: the financial source of truth should be on-chain, while data, quoting, routing, and keeper execution can remain off-chain if agents are replaceable and contracts enforce admissibility. Venue independence is achieved through explicit capability classes and isolated mandates, not through uniform trust assumptions.

The next research step is empirical rather than another layer of synthetic confidence. It should calibrate provider participation, market-maker delivery, liquidation capacity, queue behavior, and common-factor losses on actual venue adapters with externally registered holdouts. Axient's present contribution is a falsifiable protocol design, a parity-tested reference implementation, and a transparent agent-based stress laboratory. It is not a promise of yield, insurance, production safety, or universal liquidation success.

\appendix
\section{Notation and Glossary}
\label{app:glossary}

\subsection{Sets, indices, and clocks}
\begin{longtable}{>{\raggedright\arraybackslash}p{.20\textwidth}>{\raggedright\arraybackslash}p{.71\textwidth}}
\toprule
Symbol & Definition\\
\midrule
\endhead
$p\in\Pset$ & Credit pool.\\
$v\in\Vset$ & External venue or protocol-native exposure adapter.\\
$e\in\Eset$ & Event or event-linked market.\\
$i\in\Iset$ & Financed trader position.\\
$m\in\Mset$ & Market maker or bonded execution provider.\\
$t$ & Contract checkpoint, block, or accounting time.\\
$u$ & Position-level hard-flat decision time inherited from the Debt-Free Finality mechanism.\\
$\nu$ & Match-completion time for a position-level execution policy.\\
$\sigma$ & Settlement-confirmed debt-extinction time.\\
$\tau_f$ & Final event payout-vector time.\\
$\tau_r$ & Confirmed redemption time.\\
$T_c^{\mathrm{sched}}$, $T_c^{\mathrm{act}}$ & Scheduled and actual cessation of risk-reducing venue execution.\\
\bottomrule
\end{longtable}

\subsection{Position quantities}
\begin{longtable}{>{\raggedright\arraybackslash}p{.20\textwidth}>{\raggedright\arraybackslash}p{.71\textwidth}}
\toprule
Symbol & Definition\\
\midrule
\endhead
$C_i$ & Trader collateral assigned to position $i$.\\
$L_i$ & Admitted leverage, bounded by $1\le L_i\le5$.\\
$N_i=L_iC_i$ & Gross acquisition notional.\\
$D_{i,t}$ & Settlement-recognized debt, including accrued interest.\\
$d_i$ & Position debt shares.\\
$I^B_{p,t}$ & Pool borrow index, so $D_{i,t}=d_iI^B_{p,t}$.\\
$K_{i,t}$ & Settled position cash.\\
$Q_{i,t}$ & Recognized position asset quantity or venue claim.\\
$P^{\mathrm{pend}}_{i,t}$ & Matched or submitted but unsettled proceeds.\\
$P^{\mathrm{set}}_{i,t}$ & Settlement-confirmed proceeds available for debt allocation.\\
$X_{i,t}$ & Total settlement-recognized position value controlled by the position account.\\
$E_{i,t}=[X_{i,t}-D_{i,t}]^+$ & Trader residual equity.\\
$\ell_{i,t}=[D_{i,t}-X_{i,t}]^+$ & Latent or recognized credit shortfall, depending on state.\\
$\underline B_{i,u,\Delta}(x)$ & Conservative lower settled-proceeds envelope for a sale of $x$, inherited from the position-level mechanism.\\
$\overline H_{i,u,\Delta}$ & Upper bound on debt and execution costs over the settlement horizon.\\
$m_i$ & Explicit position-level operational buffer.\\
$\Gamma_i(u)$ & Robust coverage ratio used only to rank emergency execution; not an additive solvency measure.\\
\bottomrule
\end{longtable}

\subsection{Pool quantities and claim classes}
\begin{longtable}{>{\raggedright\arraybackslash}p{.20\textwidth}>{\raggedright\arraybackslash}p{.71\textwidth}}
\toprule
Symbol & Definition\\
\midrule
\endhead
$A_{p,t}$ & Aggregate settled cash controlled by the local pool, including tagged sub-accounts.\\
$A^L_{p,t}$ & Unencumbered senior cash eligible for new lending.\\
$A^W_{p,t}$ & Unencumbered cash eligible for immediate withdrawal service.\\
$B_{p,t}$ & Recognized debt receivable.\\
$\mathcal A_{p,t}=A_{p,t}+B_{p,t}$ & Recognized local pool assets.\\
$V^S_{p,t}$ & Senior LP net asset value.\\
$V^J_{p,t}$ & LBP net asset value.\\
$R^M_{p,e,v,t}$ & Market, event, venue, or risk-bucket reserve.\\
$R^P_{p,t}$ & Pool reserve.\\
$G_t$ & Optional shared protocol-global reserve.\\
$T_{p,t}$ & Protocol treasury or accrued protocol-fee claim inside the local accounting boundary.\\
$\mathcal C_{p,t}$ & Aggregate recognized local claims and equity layers.\\
$Z^S_{p,t}$, $Z^J_{p,t}$ & Senior and LBP share supplies.\\
$\pi^S_{p,t}=V^S_{p,t}/Z^S_{p,t}$ & Senior share price.\\
$\pi^J_{p,t}=V^J_{p,t}/Z^J_{p,t}$ & Junior share price.\\
$W_{p,t}$ & Outstanding withdrawal claims.\\
$U_{p,t}=B_{p,t}/(A^L_{p,t}+B_{p,t})$ & Senior credit utilization.\\
$r^B(U)$ & Borrow-rate policy curve.\\
$H_{p,t}$ & Gross accrued borrower interest over an accounting interval.\\
$\alpha_S,\alpha_J,\alpha_R,\alpha_P,\alpha_M$ & Non-negative interest-allocation weights for senior, LBP, reserve, protocol, and MM claims.\\
$H^{\mathrm{acc}}_{p,t}$ & Interest accrued into debt receivables.\\
$H^{\mathrm{cash}}_{p,t}$ & Interest received as settled cash.\\
$H^{\mathrm{pend}}_{p,t},H^{\mathrm{liq}}_{p,t},H^{\mathrm{set}}_{p,t}$ & Conservative NAV holdbacks for pending-loss, liquidity, and settlement uncertainty.\\
$\widetilde{\mathcal A}_{p,t}$ & Conservative recognized assets used for withdrawal pricing after a risk trigger.\\
$\delta_t$ & Explicit rounding-dust balance.\\
$u_A,u_Z$ & One asset and one share base unit.\\
$v_A,v_Z$ & Virtual asset/share offsets used in initial-share protection.\\
\bottomrule
\end{longtable}

\subsection{Protection and execution quantities}
\begin{longtable}{>{\raggedright\arraybackslash}p{.20\textwidth}>{\raggedright\arraybackslash}p{.71\textwidth}}
\toprule
Symbol & Definition\\
\midrule
\endhead
$H_k$ & Capacity of waterfall layer $k$.\\
$a_k$ & Loss allocated to waterfall layer $k$.\\
$\rho_k$ & Residual shortfall after waterfall layer $k$.\\
$J_{p,t}$ & LBP junior capital balance; economically represented by $V^J_{p,t}$.\\
$F^{\mathrm{liq}}(a)$ & Liquidator compensation for valid settled action $a$.\\
$\kappa_m$ & Collateralized MM commitment tuple.\\
$Q_m$ & Maximum quantity under MM commitment $m$.\\
$\underline P_m$ & Minimum price or net-proceeds condition of a commitment.\\
$\Gamma_m$ & Admissible execution conditions of a commitment; unrelated to position ratio $\Gamma_i$.\\
$b_m$ & Stablecoin bond or atomic-delivery value supporting commitment $m$.\\
$c_m(x)$ & Incremental robust recovery credited to commitment $m$ for quantity $x$.\\
$y_m(x)$ & Realized incremental settled delivery.\\
$s_m(x,y)$ & Enforceable bond transfer after underperformance.\\
$K_p^\star$ & Maximum registered residual shortfall requiring protected capital.\\
$\Omega_p^\star$ & Finite registered scenario family for pool $p$.\\
\bottomrule
\end{longtable}

\subsection{Roles}
\begin{longtable}{>{\raggedright\arraybackslash}p{.30\textwidth}>{\raggedright\arraybackslash}p{.61\textwidth}}
\toprule
Term & Definition\\
\midrule
\endhead
Trader & Supplies first-loss collateral and receives residual event exposure after debt priority.\\
Senior Credit LP & Supplies revolving stablecoin principal, receives priority repayment and senior pool returns, and may suffer loss after earlier layers are exhausted.\\
Liquidator & Executes an admissible risk-reducing transition and receives a bounty; need not supply principal-risk capital.\\
Liquidation Backstop Provider (LBP) & Supplies junior loss-absorbing provider capital ahead of senior principal and receives designated income; may, but need not, operate a liquidator.\\
Market maker & Supplies ordinary quotes, hedging, RFQ, or collateralized executable commitments.\\
Keeper & Replaceable off-chain agent that observes state and submits admissible transactions.\\
Governance & Timelocked authority over prospective mandates, caps, adapter versions, and emergency controls.\\
\bottomrule
\end{longtable}

\subsection{Venue capability and backing terms}
\begin{longtable}{>{\raggedright\arraybackslash}p{.30\textwidth}>{\raggedright\arraybackslash}p{.61\textwidth}}
\toprule
Term & Definition\\
\midrule
\endhead
Venue capability vector $\chi_v$ & Binary or graded evidence for contract custody, lender lien, liquidation authority, verifiable settlement, verifiable redemption, and withdrawal lock.\\
Mode A: native physical backing & Contract-controlled recognized outcome assets with enforceable lien and liquidation.\\
Mode B: verifiable venue settlement & Publicly verifiable execution and settlement with some operator or signer dependence.\\
Mode C: attested or custodial backing & Assets or claims depend on a custodian, controlled account, or attestation.\\
Mode D: protocol-native synthetic & Axient issues the user claim and hedges gross or net exposure externally; not equivalent to physical backing.\\
Strict isolation & No shared cash, reserves, junior capital, bonds, collateral, guarantees, or write access across pools.\\
Protection contagion & Reduction in one pool's future shared protection caused by another pool's draw, without direct transfer of the first pool's liability.\\
\bottomrule
\end{longtable}

\subsection{Selected state names}
\begin{longtable}{>{\raggedright\arraybackslash}p{.30\textwidth}>{\raggedright\arraybackslash}p{.61\textwidth}}
\toprule
State & Meaning\\
\midrule
\endhead
\state{Requested} & User or LP intent received but not yet admitted.\\
\state{Admitted} & All current admission gates passed; funding may proceed.\\
\state{FundingPending} & Credit and collateral transfers or venue acquisition are pending.\\
\state{Open} & Position assets and debt are recognized.\\
\state{ReduceOnly} & Risk may be reduced but not increased.\\
\state{Liquidating} & Risk-reducing execution is in progress.\\
\state{SettlementPending} & Match exists but repayment cash is not yet recognized.\\
\state{PartiallySettled} & Some unique settlement evidence has been applied but debt or matched capacity remains.\\
\state{EvidenceConflict} & Settlement or loss evidence is contradictory, duplicated, or invalid pending reconciliation.\\
\state{Repaid} & Debt shares have been fully burned against settled value.\\
\state{Shortfall} & Recognized position recovery is below debt.\\
\state{WaterfallFinalized} & Loss allocations and share-price effects are final.\\
\state{EventCloseException} & Venue became non-executable before ordinary debt-clearing completion.\\
\state{Recovery} & Pool is processing loss, reconciliation, or constrained withdrawals.\\
\bottomrule
\end{longtable}

\subsection{Endogenous-capital and agent-model quantities}
\begin{longtable}{>{\raggedright\arraybackslash}p{.22\textwidth}>{\raggedright\arraybackslash}p{.69\textwidth}}
\toprule
Symbol & Definition\\
\midrule
\endhead
$\Iset_S,\Iset_J,\Iset_T$ & Finite sets of Senior LPs, LBPs, and traders in the endogenous model.\\
$\bar s_i,\bar j_j,c_n$ & Senior endowment, LBP endowment, and trader collateral budget.\\
$o_i,o_j^J$ & Outside opportunity or required return of a capital provider.\\
$G_i,H_j$ & Bounded participation functions for Senior and LBP capital.\\
$y_t^S,y_t^J$ & Expected Senior and LBP provider compensation.\\
$\mu_t^S,\mu_t^J$ & Expected Senior and LBP impairment.\\
$q_t,w_t,\ell_t$ & Expected queue delay, observed withdrawal rate, and LBP encumbrance cost.\\
$S_t,J_t$ & Endogenous aggregate Senior and LBP capital.\\
$a_{n,t},\gamma_n,\zeta_n$ & Trader perceived edge, risk aversion, and rate sensitivity.\\
$U_{n,t}(L)$ & Trader utility of leverage tier $L$.\\
$B_t(r)$ & Aggregate requested or admitted credit demand at rate $r$, as stated.\\
$\Phi(U)$ & Utilization response map induced by supply, demand, and the rate curve.\\
$U^*$ & Fixed-point utilization satisfying $U^*=\Phi(U^*)$.\\
$\omega$ & Damping parameter in the utilization update.\\
$\mu$ & Expected-loss or risk-state input used in the explicit risk spread.\\
$w_L$ & Gross-exposure weight assigned to leverage tier $L$.\\
$G_t^{\max}$ & Maximum gross open interest supported by endogenous lendable capital and a leverage mix.\\
$m^o_{k,t},m^b_{k,t}$ & Ordinary revocable and bonded commitment capacity of market maker $k$.\\
$b_{k,t},s_{k,t}$ & MM bond and enforceable slashing amount.\\
$g_{k,t}$ & Private gain from strategic non-delivery.\\
$q_{\ell,t},c_{\ell,t},p^{\mathrm{succ}}_{\ell,t}$ & Liquidator capacity, cost, and estimated success probability.\\
$\mathcal A_t$ & Active liquidator set; not to be confused with pool assets $\mathcal A_{p,t}$.\\
$Q_t$ & Aggregate liquidation notional required at time $t$.\\
$M,D(M),h(M)$ & Credited backstop capacity, admitted debt, and unit shortfall.\\
$\mathcal L(M)$ & Aggregate expected shortfall as a function of credited capacity.\\
$x_{i,t},\bar x_t,F$ & Individual and aggregate withdrawal fractions and the aggregate run response map.\\
$R_t^*,\widehat{\mathrm{ES}}_t$ & Dynamic reserve target and registered expected-shortfall proxy.\\
$Z_t,\beta_p,\varepsilon_{p,t}$ & Common loss factor, pool loading, and idiosyncratic loss.\\
\bottomrule
\end{longtable}

\subsection{Agent evidence terms}
\begin{longtable}{>{\raggedright\arraybackslash}p{.30\textwidth}>{\raggedright\arraybackslash}p{.61\textwidth}}
\toprule
Term & Definition\\
\midrule
\endhead
Path-level Senior-loss incidence & Fraction of 104-week synthetic paths in which any Senior impairment occurs; not a weekly default probability.\\
Pool-week loss incidence & Fraction of pool-week observations with positive Senior impairment.\\
Accepted-demand ratio & Admitted financed credit divided by synthetic requested credit.\\
Run-week incidence & Fraction of pool-week observations crossing the registered withdrawal-run threshold.\\
LBP impairment incidence & Fraction of paths with at least one LBP impairment.\\
Cumulative LBP draw / capital & Sum of LBP loss allocations over time divided by mean deployed LBP capital; can exceed one when capital is replenished and impaired repeatedly.\\
Release-registered hypothesis & Threshold stored in the versioned parameter registry before the final execution; not externally preregistered.\\
Scalar invariant evaluation & One evaluated state assertion for one simulated state element; counts are computational evidence, not independent theorem proofs.\\
\bottomrule
\end{longtable}

\section{Supplementary Proofs and Identities}
\label{app:proofs}

\subsection{Fee-routing conservation}
\begin{proof}[Proof of \Cref{prop:fee-conservation}]
Using \eqref{eq:interest-weights},
\[
\Delta V^S+\Delta V^J+\Delta R+\Delta T+\Delta M
=H(\alpha_S+\alpha_J+\alpha_R+\alpha_P+\alpha_M)=H.
\]
Non-negativity of the weights gives non-negative allocations. \qedhere
\end{proof}

\subsection{Share-price monotonicity}
\begin{proof}[Proof of \Cref{cor:share-monotonicity}]
By \Cref{prop:no-dilution}, proportional deposits do not change $\pi^S$. A withdrawal that burns shares at the prevailing price removes equal proportions of NAV and share supply and also leaves the price unchanged. Between such transitions, retained non-negative income weakly increases $V^S$ while $Z^S$ is fixed, and the absence of senior loss prevents a negative NAV jump. Therefore $\pi^S$ is weakly increasing. \qedhere
\end{proof}

\subsection{Credit geometry}
\begin{proof}[Proof of \Cref{prop:credit-geometry}]
For $L>1$,
\[
\delta'(L)=\frac{1}{L^2}>0,
\qquad
\kappa'(L)=-\frac{1}{(L-1)^2}<0.
\]
Moreover,
\[
\delta(L)\kappa(L)=
\frac{L-1}{L}\frac{L}{L-1}=1.
\]
Thus the financed fraction rises and gross exposure per unit of credit falls with leverage. \qedhere
\end{proof}

\subsection{Proportional reserve allocation}
\begin{proof}[Proof of \Cref{prop:pro-rata-reserve}]
If $\sum_j\ell_j\le R$, then $g_i=\ell_i$ and all claims are immediate. Otherwise, $0<R/\sum_j\ell_j<1$, so $0\le g_i<\ell_i$ for every positive shortfall. Summing gives
\[
\sum_i g_i=R\frac{\sum_i\ell_i}{\sum_j\ell_j}=R.
\]
The coverage ratio $g_i/\ell_i$ equals the common factor $R/\sum_j\ell_j$. \qedhere
\end{proof}

\subsection{Cumulative-capacity characterization of the waterfall}
Let $C_k=\sum_{j=1}^{k}H_j$.
\begin{lemma}[Closed-form residual]
\label{lem:waterfall-closed-form}
For the waterfall in \Cref{def:waterfall},
\begin{equation}
\rho_k=\pospart{\ell-C_k},
\qquad
\sum_{j=1}^{k}a_j=\min\{\ell,C_k\}.
\label{eq:waterfall-closed-form}
\end{equation}
\end{lemma}
\begin{proof}
For $k=1$, $\rho_1=\ell-\min\{H_1,\ell\}=[\ell-H_1]^+$. Suppose $\rho_{k-1}=[\ell-C_{k-1}]^+$. Then
\[
\rho_k=\rho_{k-1}-\min\{H_k,\rho_{k-1}\}
=[\rho_{k-1}-H_k]^+
=[\ell-C_k]^+.
\]
The allocation identity follows from conservation. \qedhere
\end{proof}

\subsection{Necessity of collateral for credited MM substitution}
\begin{proposition}[Uncollateralized promise cannot improve a robust lower bound]
\label{prop:uncollateralized-mm}
If a market-maker commitment admits a failure path with $y_m=0$ and no enforceable transfer $s_m=0$, then no positive amount can be credited to a worst-case recovery lower bound over a set containing that path.
\end{proposition}
\begin{proof}
On the failure path, realized incremental recovery is zero. Any credited amount $c>0$ would exceed realized recovery, contradicting lower-bound validity. Therefore the largest valid credited amount is zero. \qedhere
\end{proof}

\subsection{Global-reserve allocation and network conservation}
Suppose two pools draw $g_1$ and $g_2$ sequentially from a global reserve $G$, with $g_1\le G$ and $g_2\le G-g_1$. The reserve decline is exactly $g_1+g_2$. If the draws reduce recognized pool shortfalls by the same amounts, network book closure is preserved when the global reserve is included in the consolidated accounting boundary. Local pool closure is preserved by recording a contribution from the global layer and an equal shortfall allocation.

\subsection{Queue theorem at variable price}
The bounds in \Cref{thm:queue-clearance} do not require knowledge of service order. If each paid share costs a contemporaneous price inside $[\underline\pi,\overline\pi]$, any complete service path costs between $Z_Q\underline\pi$ and $Z_Q\overline\pi$. The sufficient and impossible regions therefore remain valid under FIFO, epoch pro rata, or deterministic priority. The interval between them is deliberately indeterminate because the actual price path matters.

\subsection{Interest allocation and senior return}
A positive mean senior return in the synthetic experiment is not guaranteed by \eqref{eq:senior-return}. For a given path,
\[
R_T^S<0
\quad\Longleftrightarrow\quad
L_T^S>I_T^S+F_T^S.
\]
The LBP inequality is analogous. This identity is why the simulation reports negative-return incidence in addition to principal-loss incidence.

\subsection{Implementation-parity results}
\paragraph{Trader residual.}
For settled amount $y$, allocations $y^D$ and $y^F$ are capped by debt and approved fees. The residual $y^T=y-y^D-y^F$ is non-negative. The position identity assigns all recognized value above debt to trader equity; any other allocation requires a separate authorized transfer. This proves \Cref{prop:trader-residual}.

\paragraph{Settlement idempotency.}
The evidence-consumption flag and financial journal are updated atomically. An already consumed identifier cannot satisfy the transition precondition. A distinct identifier remains capped by unmatched capacity. Hence no settled unit can be credited twice, proving \Cref{prop:settlement-idempotency}.

\paragraph{Queued-share loss participation.}
Escrow changes neither total share supply nor the queued claim's rank. Applying a conservative holdback or loss before service reduces the common NAV numerator. Payment at current price therefore applies the same reduction to queued and unqueued shares, proving \Cref{prop:no-first-exit} under its timing assumptions.

\paragraph{Integer pro rata.}
The floor stage allocates no more than the target and leaves fewer than one base unit per positive claimant. The largest-remainder stage distributes exactly the residual units without exceeding claimant caps. Deterministic tie-breaks establish uniqueness, proving \Cref{prop:integer-pro-rata}.

\paragraph{Dust closure.}
Every rounded journal equals its ideal journal minus explicit residual. Posting the residual to dust restores equality. Each downward recipient allocation contributes less than one asset base unit, and summation yields \eqref{eq:dust-bound}. This proves \Cref{thm:integer-closure}.

\subsection{Endogenous-capital results}
\begin{proof}[Proof of \Cref{thm:utilization-existence}]
Continuity of $r$, $B$, and $S$, together with strict positivity of $S$, implies continuity of $\Phi$ on $[0,\bar U]$. Admission clipping makes $\Phi$ a continuous self-map of the non-empty compact convex interval $[0,\bar U]$. Brouwer's fixed-point theorem in one dimension therefore gives at least one $U^*$ with $U^*=\Phi(U^*)$. Equivalently, $f(U)=\Phi(U)-U$ is continuous and cannot map both interval endpoints strictly outward. \qedhere
\end{proof}

\begin{proof}[Proof of \Cref{prop:utilization-contraction}]
Condition~\eqref{eq:utilization-contraction} makes $\Phi$ a contraction on the complete metric space $[0,\bar U]$. Banach's fixed-point theorem gives uniqueness and global convergence of direct iteration. The damped map is $\Psi(U)=(1-\omega)U+\omega\Phi(U)$, with derivative $\Psi'(U^*)=(1-\omega)+\omega\Phi'(U^*)$. The stated absolute-value condition gives local contraction and hence local convergence. \qedhere
\end{proof}

\begin{proof}[Proof of \Cref{prop:leverage-mix-crowding}]
In \eqref{eq:endogenous-gross-capacity}, the numerator is fixed. Shifting mass $\epsilon>0$ from $L_a$ to $L_b>L_a$ changes the denominator by $\epsilon[\delta(L_b)-\delta(L_a)]\ge0$ because $\delta$ is increasing. The reciprocal therefore weakly falls. It falls strictly whenever the shift is positive and the two financed fractions differ. \qedhere
\end{proof}

\subsection{Strategic-execution results}
\begin{proof}[Proof of \Cref{prop:mm-incentive-compatibility}]
Relative to strategic non-delivery, delivery yields commitment and execution compensation and avoids enforceable slashing, while incurring execution cost and foregoing the private gain from refusal. Delivery weakly dominates exactly when the left side of \eqref{eq:mm-incentive-condition} is at least the right side. If the private gain is unbounded, every finite bond and premium is dominated on some state. \qedhere
\end{proof}

\begin{proof}[Proof of \Cref{thm:liquidator-coverage}]
Every active liquidator satisfies its private entry condition and submits at most $q_{\ell,t}$. If the adapter admits the intended risk-reducing transition and the sum of submitted capacities is at least $Q_t$, the requested notional can be assigned across active liquidators without exceeding individual capacities. This proves capacity sufficiency only; fill, price, and settlement require the separate market and adapter conditions stated in the theorem. \qedhere
\end{proof}

\begin{proof}[Proof of \Cref{prop:capacity-induced-shortfall}]
Differentiating \eqref{eq:capacity-shortfall} gives
\[
\mathcal L'(M)=D'(M)h(M)+D(M)h'(M).
\]
For positive $D$ and $h$, $\mathcal L'(M)>0$ if and only if $D'/D>-h'/h$, which is \eqref{eq:capacity-paradox-condition}. \qedhere
\end{proof}

\subsection{Withdrawal, reserve, and common-factor results}
\begin{proof}[Proof of \Cref{prop:run-contraction}]
The assumptions make $F$ a contraction from the complete metric space $[0,1]$ to itself. Banach's theorem gives one fixed point and convergence to it from every initial withdrawal rate. \qedhere
\end{proof}

\begin{proof}[Proof of \Cref{prop:request-time-neutrality}]
Let the two providers hold equal shares immediately before the request. Escrowing one provider's shares changes neither total share supply nor participation in the common NAV loss. Because payment has not finalized, both share lots are multiplied by the same post-loss share price. Their claims therefore remain equal. \qedhere
\end{proof}

\begin{proof}[Proof of \Cref{thm:reserve-replenishment}]
While $R_t<R^*$, no draw occurs and the net inflow is at least $\epsilon$, so $R_{t+1}\ge R_t+\epsilon$. By induction, after $n$ periods $R_{t_0+n}\ge R_{t_0}+n\epsilon$. Choosing $n=\lceil(R^*-R_{t_0})/\epsilon\rceil$ reaches or exceeds the target. The positive-part operator cannot reduce the value because the recursion is non-negative. \qedhere
\end{proof}

\begin{proof}[Proof of \Cref{thm:common-factor-contagion}]
Using \eqref{eq:common-factor-loss}, bilinearity gives
\[
\operatorname{Cov}(L_p,L_q)=\beta_p\beta_q\operatorname{Var}(Z)
+\beta_p\operatorname{Cov}(Z,\varepsilon_q)
+\beta_q\operatorname{Cov}(\varepsilon_p,Z)
+\operatorname{Cov}(\varepsilon_p,\varepsilon_q).
\]
Independence makes the last three terms zero, yielding \eqref{eq:common-factor-covariance}. \qedhere
\end{proof}

\subsection{Endogenous-capital results}
\begin{proof}[Proof of \Cref{thm:utilization-existence}]
By assumption, $r$, $B$, and $S$ are continuous and $S$ is strictly positive. Therefore $\Phi(U)=B(r(U),\chi)/S(r(U),\mu,q,w)$ is continuous wherever evaluated. Admission clips the image to the compact convex interval $[0,\bar U]$. Every continuous self-map of a compact interval has a fixed point. Equivalently, $g(U)=\Phi(U)-U$ is continuous, $g(0)\ge 0$, and $g(\bar U)\le 0$, so the intermediate-value theorem gives $U^*$ with $g(U^*)=0$. \qedhere
\end{proof}

\begin{proof}[Proof of \Cref{prop:utilization-contraction}]
If $\sup_U|\Phi'(U)|<1$, the mean-value theorem implies that $\Phi$ is a contraction on $[0,\bar U]$. Banach's fixed-point theorem gives uniqueness. The derivative of the damped update map $T(U)=(1-\omega)U+\omega\Phi(U)$ at the fixed point is
\[
T'(U^*)=(1-\omega)+\omega\Phi'(U^*).
\]
The standard one-dimensional local-stability condition is $|T'(U^*)|<1$, which is the stated inequality. \qedhere
\end{proof}

\begin{proof}[Proof of \Cref{prop:leverage-mix-crowding}]
The denominator in \eqref{eq:endogenous-gross-capacity} is the weighted average $\bar\delta=\sum_Lw_L\delta(L)$. Since $\delta(L)=(L-1)/L$ is strictly increasing, a mean-preserving transfer of weight from a lower to a higher leverage tier weakly increases $\bar\delta$. The numerator is fixed, so $G_t^{\max}=\bar U S_t^{\mathrm{lend}}/\bar\delta$ weakly decreases. \qedhere
\end{proof}

\subsection{Strategic-execution results}
\begin{proof}[Proof of \Cref{prop:mm-incentive-compatibility}]
Under delivery, the market maker receives $\pi^{\mathrm{com}}+\pi^{\mathrm{exec}}$ and incurs $c^{\mathrm{exec}}$. Under non-delivery it receives private gain $g$ and loses the enforceable slash $s$. Delivery weakly dominates precisely when
\[
\pi^{\mathrm{com}}+\pi^{\mathrm{exec}}-c^{\mathrm{exec}}\ge g-s,
\]
which rearranges to \eqref{eq:mm-incentive-condition}. If $g$ has no finite upper bound, then for every finite $s\le b$ there exists a state with $g$ large enough to violate the inequality. \qedhere
\end{proof}

\begin{proof}[Proof of \Cref{thm:liquidator-coverage}]
By definition, every member of $\mathcal A_t$ has entered and can submit at most $q_{\ell,t}$. If the adapter accepts the required transition and aggregate active capacity is at least $Q_t$, the active set can partition the required notional into non-negative pieces whose total is $Q_t$ and no piece exceeds the corresponding capacity. This proves capacity sufficiency. It does not imply execution price, fill, or settlement, which are not premises of the capacity statement. \qedhere
\end{proof}

\begin{proof}[Proof of \Cref{prop:capacity-induced-shortfall}]
Differentiate $\mathcal L(M)=D(M)h(M)$:
\[
\mathcal L'(M)=D'(M)h(M)+D(M)h'(M).
\]
For positive $D$ and $h$, $\mathcal L'(M)>0$ exactly when
\[
\frac{D'(M)}{D(M)}> -\frac{h'(M)}{h(M)}.
\]
Thus aggregate shortfall can rise even while unit shortfall falls. \qedhere
\end{proof}

\subsection{Run, reserve, and contagion results}
\begin{proof}[Proof of \Cref{prop:run-contraction}]
The assumption $\sup_x|F'(x)|<1$ makes $F$ a contraction on the complete metric space $[0,1]$. Banach's fixed-point theorem gives a unique fixed point and convergence of every iterated withdrawal state to it. \qedhere
\end{proof}

\begin{proof}[Proof of \Cref{prop:request-time-neutrality}]
Let two providers hold equal shares immediately before a common NAV loss. A request moves one provider's shares to escrow but does not burn them or remove them from the NAV denominator. The common loss therefore reduces the share price applied to both holdings by the same factor. If the request is unpaid before the loss, its post-loss economic claim equals the otherwise identical unqueued claim. \qedhere
\end{proof}

\begin{proof}[Proof of \Cref{thm:reserve-replenishment}]
While $R_t<R^*$, the assumptions imply $R_{t+1}\ge R_t+\epsilon$. By induction, after $n$ periods,
\[
R_{t_0+n}\ge R_{t_0}+n\epsilon.
\]
The smallest integer $n$ for which the right-hand side reaches $R^*$ is the ceiling in \eqref{eq:reserve-replenishment-bound}. \qedhere
\end{proof}

\begin{proof}[Proof of \Cref{thm:common-factor-contagion}]
For distinct pools,
\begin{align*}
\operatorname{Cov}(L_p,L_q)
&=\operatorname{Cov}(\beta_pZ+\varepsilon_p,\beta_qZ+\varepsilon_q)\\
&=\beta_p\beta_q\operatorname{Var}(Z)
+\beta_p\operatorname{Cov}(Z,\varepsilon_q)
+\beta_q\operatorname{Cov}(\varepsilon_p,Z)
+\operatorname{Cov}(\varepsilon_p,\varepsilon_q).
\end{align*}
Independence makes the last three terms zero, yielding \eqref{eq:common-factor-covariance}. \qedhere
\end{proof}

\section{Reference Algorithms}
\label{app:algorithms}

\begin{algorithm}[H]
\caption{Position admission and funding}
\label{alg:admission}
\begin{algorithmic}[1]
\Require trader collateral $C$, requested leverage $L$, market $e$, venue adapter $v$, pool $p$
\State verify account tier and jurisdiction adapter
\State verify $v$ and backing mode satisfy pool mandate
\State compute exact entry cost and robust aggregate exit certificate
\State compute $L^{\mathrm{eff}}$ from \eqref{eq:effective-leverage}
\If{$L>L^{\mathrm{eff}}$} \State reduce quote or reject \EndIf
\State compute credit $D_0=(L-1)C$ and all fee and reserve encumbrances
\State verify lendable cash, utilization, queue buffer, OI, and concentration limits
\State reserve shared-book and bonded-MM capacity atomically or by versioned allocation
\State lock trader collateral and transfer $D_0$ from \texttt{CreditPool}
\State acquire recognized position assets through $v$
\If{asset receipt and settlement evidence are invalid}
  \State enter \state{FundingFailed}; unwind or reconcile without minting live debt
\Else
  \State mint debt shares; enter \state{Open}; emit mandate and certificate hashes
\EndIf
\end{algorithmic}
\end{algorithm}

\begin{algorithm}[H]
\caption{Settlement-confirmed repayment}
\label{alg:repayment}
\begin{algorithmic}[1]
\Require position $i$, settled evidence $E$, amount $y$
\State verify $E$ under the position's versioned adapter and prevent replay
\State reject consumed evidence; cap cumulative settlement by matched capacity
\State credit the unique settled chunk to the position account
\State pay accrued interest and principal in admitted priority
\State burn corresponding debt shares with explicit rounding
\State route only the remaining residual to the trader claim
\If{$D_i=0$}
  \State release residual position assets to trader or debt-free finality account
  \State enter \state{Repaid}
\Else
  \State remain \state{Open}, \state{ReduceOnly}, or \state{Liquidating}
\EndIf
\end{algorithmic}
\end{algorithm}

\begin{algorithm}[H]
\caption{Aggregate emergency execution order}
\label{alg:aggregate-priority}
\begin{algorithmic}[1]
\Require positions sharing one executable book or commitment set
\For{each position $i$}
  \State compute $\Gamma_i$ from \eqref{eq:robust-coverage-ratio}
\EndFor
\State sort by ascending $\Gamma_i$, then earlier hard-flat deadline, then position ID
\For{each position in sorted order}
  \State allocate remaining shared-book and bonded-MM capacity
  \State submit admissible risk-reducing execution
  \State debit capacity only after deterministic allocation; never reuse consumed levels
\EndFor
\State reconcile matched and settled amounts; compute recognized shortfalls
\end{algorithmic}
\end{algorithm}

\begin{algorithm}[H]
\caption{Loss waterfall and shared-layer allocation}
\label{alg:waterfall}
\begin{algorithmic}[1]
\Require finalized position shortfalls $\ell_i$, ordered local capacities, shared reserve $R$
\State apply position-specific and market-specific layers to each $\ell_i$
\State aggregate residual shortfalls entering the shared reserve
\State allocate $R$ by \eqref{eq:pro-rata-reserve}
\For{each position $i$}
  \State apply remaining pool and senior layers using \eqref{eq:waterfall-allocation}
  \State emit every allocation and resulting share-price effect
\EndFor
\State assert conservation \eqref{eq:waterfall-conservation}
\end{algorithmic}
\end{algorithm}

\begin{algorithm}[H]
\caption{Bonded market-maker commitment settlement}
\label{alg:mm-settlement}
\begin{algorithmic}[1]
\Require commitment $m$, credited capacity $c_m$, delivered settled value $y_m$
\State verify quantity, time, price, venue, and evidence conditions
\State compute deficit $d_m=[c_m-y_m]^+$
\State transfer $\min\{d_m,b_m^{\mathrm{available}}\}$ from bond to protected account
\State reduce bond and commitment capacity by the consumed amounts
\If{delivered value plus transfer is below credited capacity}
  \State recognize uncovered commitment failure and route to waterfall
  \State suspend further credit to commitment $m$
\EndIf
\end{algorithmic}
\end{algorithm}

\begin{algorithm}[H]
\caption{Share-denominated withdrawal queue service}
\label{alg:queue}
\begin{algorithmic}[1]
\Require queued share claims, deterministic service rule, withdrawal-eligible cash $A^W$
\State escrow queued shares without removing them from loss participation
\While{$A^W>0$ and queue non-empty}
  \State select next request or epoch according to pool rule
  \State apply active conservative NAV holdbacks and compute payment at current class share price
  \State pay no more than $A^W$; update claim and NAV consistently
  \State emit partially-paid or paid state
  \State refresh $A^W$ after mandatory buffers and finalized settlements
\EndWhile
\end{algorithmic}
\end{algorithm}

\begin{algorithm}[H]
\caption{Pool risk-mode update}
\label{alg:risk-mode}
\begin{algorithmic}[1]
\Require pool state, reserve coverage, queue coverage, venue status, settlement status
\If{invalid adapter evidence or unresolved accounting difference}
  \State enter \state{Paused} or \state{Recovery}
\ElsIf{capital or protection coverage below mandate floor}
  \State enter \state{ReduceOnly}; reject new debt
\ElsIf{utilization or queue pressure above rate-limit threshold}
  \State enter \state{RateLimited}; lower caps and raise borrow rate per policy
\Else
  \State remain or return to \state{Active} after timelocked or objective recovery conditions
\EndIf
\end{algorithmic}
\end{algorithm}

\begin{algorithm}[H]
\caption{Exact integer shared-reserve allocation}
\label{alg:integer-reserve}
\begin{algorithmic}[1]
\Require integer shortfalls $\ell_i$, integer reserve capacity $R$, deterministic IDs
\State set target $T\leftarrow\min\{R,\sum_i\ell_i\}$
\For{each positive claimant $i$}
  \State $g_i\leftarrow\left\lfloor T\ell_i/\sum_j\ell_j\right\rfloor$
  \State record fractional remainder and claimant ID
\EndFor
\State allocate remaining base units by descending remainder, then ID
\State assert $0\le g_i\le\ell_i$ and $\sum_i g_i=T$
\end{algorithmic}
\end{algorithm}

\begin{algorithm}[H]
\caption{Loss-participating withdrawal request}
\label{alg:loss-participating-queue}
\begin{algorithmic}[1]
\Require owner shares $z$, active pool mandate
\State move $z$ to non-transferable escrow; do not reduce total share supply
\State retain queued shares in every income, holdback, and loss calculation
\State before service, apply valid conservative NAV holdbacks and finalized losses
\State pay $\left\lfloor z\widetilde\pi^S\right\rfloor$ only from eligible settled cash
\State burn escrowed shares only when payment is finalized
\end{algorithmic}
\end{algorithm}

\begin{algorithm}[H]
\caption{Damped endogenous rate and credit admission}
\label{alg:endogenous-rate}
\begin{algorithmic}[1]
\Require provider state, trader traits, risk state, leverage policy, utilization ceiling
\State initialize $U^{(0)}$ from outstanding debt and lendable capital
\For{$k=0,\ldots,K-1$}
  \State compute utilization component $r_U(U^{(k)})$ and explicit risk spread $r_R(\mu)$
  \State evaluate Senior and LBP participation and trader tier demand
  \State apply pool-cash, utilization, leverage, concentration, and executable-capacity gates
  \State compute clipped response $\Phi(U^{(k)})$
  \State $U^{(k+1)}\leftarrow(1-\omega)U^{(k)}+\omega\Phi(U^{(k)})$
  \If{$|U^{(k+1)}-U^{(k)}|<\varepsilon$} \State \textbf{break} \EndIf
\EndFor
\State admit debt only up to the final gated capacity and emit convergence diagnostics
\end{algorithmic}
\end{algorithm}

\begin{algorithm}[H]
\caption{Weekly agent-market transition}
\label{alg:agent-week}
\begin{algorithmic}[1]
\Require prior pool, agent, loan, queue, reserve, and regime states; fixed shock path
\State transition the common regime and reveal common and idiosyncratic shocks
\State accrue interest; move maturing debt into settlement processing
\State compute ordinary MM capacity, bonded delivery, slashing, and active liquidator capacity
\State realize settlement recovery; apply position, reserve, LBP, global, and Senior waterfall layers
\State distribute settled income; update expected loss, queue delay, and dynamic reserve targets
\State update Senior and LBP participation; service queued shares at current loss-participating NAV
\State execute \Cref{alg:endogenous-rate}; select trader tiers and admit new credit
\State assert cash, debt, share, reserve, non-negativity, cap, and waterfall invariants
\State emit weekly and event-level evidence
\end{algorithmic}
\end{algorithm}

\begin{algorithm}[H]
\caption{Risk-sensitive reserve adjustment}
\label{alg:dynamic-reserve}
\begin{algorithmic}[1]
\Require outstanding credit $B_t$, perceived loss $\widehat{\mathrm{ES}}_t$, current reserve $R_t$, settled interest $I_t$
\State compute target $R_t^*=\alpha B_t+\beta\widehat{\mathrm{ES}}_t$
\If{$R_t<R_t^*$}
  \State increase reserve interest share by at most the registered adjustment cap
  \State fund the increment from the configured Senior and protocol shares
\Else
  \State apply the base fee split
\EndIf
\State add reserve interest and valid slashing receipts; subtract finalized draws
\State prevent reserve redemption and record any target deficit separately
\end{algorithmic}
\end{algorithm}

\begin{algorithm}[H]
\caption{Endogenous weekly capital-market update}
\label{alg:endogenous-week}
\begin{algorithmic}[1]
\Require prior pool state, registered agent traits, protocol configuration, paired shock path
\State update normal/stress/crisis regime and external incident channels
\State accrue debt; mature eligible cohorts; separate matched, settled, and delayed value
\State compute revocable MM depth, bonded commitment delivery, and active liquidator capacity
\State apply settlement-confirmed repayment and the exact loss waterfall
\State route interest and slashing; update dynamic reserve target
\State update Senior and LBP participation from expected return, impairment, queue, and peer signals
\State enqueue withdrawals without burning shares; service from eligible settled cash
\State solve damped utilization response and compute leverage-tier demand
\State admit credit subject to cash, queue buffer, execution capacity, pool limits, and risk state
\State evaluate state invariants and emit evidence rows
\end{algorithmic}
\end{algorithm}

\begin{algorithm}[H]
\caption{Independent daily robustness cycle}
\label{alg:daily-robustness}
\begin{algorithmic}[1]
\Require two-pool state, venue states, provider agents, market makers, registered policy
\State update global stress and venue impairment Markov states
\State accrue position interest and shift maturity buckets
\State draw ordinary quote withdrawal and bonded-MM response
\State realize settlement recovery, principal shortfall, and protection-layer allocation
\State update conservative provider wealth, loss signals, and reserve signals
\State submit and service Senior and LBP withdrawal queues from eligible cash
\State update provider participation and admission cash buffer
\State originate new $2\times$, $3\times$, and $5\times$ cohorts within policy gates
\State record material loss, run, rejection, hard-flat, and provider-wealth metrics
\end{algorithmic}
\end{algorithm}

\section{Registered Synthetic Parameters}
\label{app:synthetic-parameters}

All values in this appendix are author-specified mechanism-comparison inputs. They are not fitted to live Axient or venue observations.

\subsection{Global settings}
\begin{table}[H]
\centering
\caption{Synthetic experiment registry.}
\label{tab:synthetic-registry}
\begin{tabularx}{\textwidth}{lX}
\toprule
Item & Registered value\\
\midrule
Paper version & r0.2.2 (stochastic generator unchanged from r0.2.0; integer fixtures unchanged from r0.2.1)\\
Evidence class & author-specified synthetic mechanism comparison\\
Senior supply & 10,000,000 settlement units\\
Monte Carlo seed / paths & 20260714 / 200,000\\
Queue seed / paths per cell & 20260715 / 20,000\\
Contagion seed / paths & 20260716 / 500,000\\
Borrow cycle & 30 days\\
Absolute leverage set & $\{2\times,3\times,5\times\}$\\
\bottomrule
\end{tabularx}
\end{table}

\subsection{Borrow-rate policy}
The rate curve in \eqref{eq:rate-curve} uses
\begin{equation}
r_0=0.06,
\quad U^\star=0.80,
\quad s_1=0.12,
\quad s_2=0.65.
\end{equation}
Utilization is sampled as $U\sim\mathrm{Beta}(5,2)$ and clipped to $[0.15,0.97]$.

\subsection{Leverage mixes}
\begin{table}[H]
\centering
\caption{Gross-OI weights by leverage tier.}
\label{tab:leverage-mix-parameters}
\begin{tabular}{lccc}
\toprule
Mix & $2\times$ & $3\times$ & $5\times$\\
\midrule
$2\times$ only & 1.00 & 0.00 & 0.00\\
Phased & 0.55 & 0.30 & 0.15\\
$5\times$ heavy & 0.25 & 0.25 & 0.50\\
\bottomrule
\end{tabular}
\end{table}

\subsection{Synthetic haircut generator}
In \eqref{eq:synthetic-haircut}:
\begin{align*}
X_0&\sim\mathrm{Beta}(2,10),\\
T_5&\sim\text{Student-}t(5),\\
I_1&\sim\mathrm{Bernoulli}(0.035), & X_1&\sim\mathrm{Beta}(2,4),\\
I_2&\sim\mathrm{Bernoulli}(0.012), & X_2&\sim\mathrm{Beta}(2,3),\\
I_3&\sim\mathrm{Bernoulli}(0.004), & X_3&\sim\mathrm{Beta}(2,2),\\
\varepsilon_j&=0.008\,\mathrm{Beta}(2,6), &
(s_{2x},s_{3x},s_{5x})&=(0,0.010,0.030).
\end{align*}
The common terms create within-path correlation across leverage tiers. The tier shift is a synthetic size-effect proxy, not an empirical estimate.

\subsection{Protection configurations and interest splits}
\begin{table}[H]
\centering
\scriptsize
\caption{Capital layers and gross-interest allocation. Values are fractions of senior supply or gross interest.}
\label{tab:protection-parameters}
\begin{tabularx}{\textwidth}{Xrrrrrrr}
\toprule
Configuration & Reserve & LBP & Senior & LBP income & Reserve & Protocol & MM\\
\midrule
Senior only & 0.00 & 0.00 & 0.90 & 0.00 & 0.00 & 0.10 & 0.00\\
Reserve & 0.01 & 0.00 & 0.82 & 0.00 & 0.13 & 0.05 & 0.00\\
Reserve + LBP & 0.01 & 0.03 & 0.76 & 0.06 & 0.13 & 0.05 & 0.00\\
Reserve + LBP + bonded MM & 0.01 & 0.03 & 0.73 & 0.06 & 0.13 & 0.05 & 0.03\\
\bottomrule
\end{tabularx}
\end{table}
Bonded capacity equals 1.5 percent of $3\times$ gross notional plus 4.5 percent of $5\times$ gross notional in the final configuration. The values are deliberately transparent policy fixtures.

\subsection{Finite scenario grid}
The grid fixes gross notional at 1,000,000 units, reserve capacity at 1 percent of senior supply, LBP capacity at 3 percent, and bonded capacity at 4 percent of gross notional for $3\times$ and $5\times$ rows in the bonded-MM configuration. Recovery ratios are $1.00$, $0.90$, $0.82$, $0.80$, $0.70$, $0.50$, and $0.00$.

\subsection{Withdrawal queue}
The queue experiment uses 100 equal-principal loans, an 8 percent cash buffer, utilization levels $0.50$, $0.70$, $0.80$, $0.90$, and $0.95$, and withdrawal fractions $0.10$, $0.20$, $0.40$, and $0.60$. Synthetic loan maturities follow a ceiling-rounded lognormal distribution with log-mean $\log(21)$ and log-standard deviation $0.65$. One percent of loans receive an additional uniformly sampled delay from 15 to 60 days. The experiment assumes constant share price, no new borrowing, and no credit loss during queue service.

\subsection{Two-pool contagion}
Each pool has senior supply 5,000,000 units, local reserve 0.5 percent, and LBP capital 2 percent. Strict isolation adds a separate 1 percent reserve to each pool. The shared regime replaces those amounts with one 2 percent global reserve, consumed first by pool one and then by pool two. Synthetic shortfalls contain common positive Student-$t(4)$ exposure, a 3 percent ordinary incident channel, and a 0.4 percent severe incident channel.

\subsection{Implementation-parity registry}
The r0.2.1 integer verifier, retained unchanged in r0.2.2, uses one settlement-asset base unit, a $10^{27}$ debt-index scale, virtual offsets $v_A=1$ and $v_Z=10^6$, and a configured public minimum-deposit fixture of $10^6$ asset base units. These are reference arithmetic parameters, not recommended production values. Twelve named fixtures and 27,441 deterministic checks are generated independently from the stochastic experiments.

\subsection{Agent-based registry r0.3.1}
The agent experiment is a separate evidence layer. Its parameter registry is \path{agent_config/agent_parameters_r0.3.1.json}, with explicit PCG64 seed 20260720. It executes 96 paths of 104 weekly periods. Every path contains macro, election, and sports pools. Each pool contains 64 Senior LPs, 24 LBPs, 12 market makers, 16 liquidators, and 96 traders.

\begin{table}[H]
\centering
\caption{Agent pool registry. Monetary quantities are synthetic settlement units.}
\label{tab:agent-pool-registry}
\small
\begin{tabular}{lrrrrr}
\toprule
Pool & Senior target & Trader collateral & Haircut & Risk & Max $L$\\
\midrule
Macro & 10,000,000 & 4,600,000 & 1.0\% & 0.65 & $5\times$\\
Election & 8,000,000 & 3,600,000 & 1.4\% & 0.90 & $5\times$\\
Sports & 7,000,000 & 3,200,000 & 2.0\% & 1.20 & $3\times$\\
\bottomrule
\end{tabular}
\end{table}

The rate curve retains $(r_0,U^*,s_1,s_2)=(0.06,0.80,0.12,0.65)$; the static-rate configurations use 0.15. Loan-duration probabilities for one, two, and four weeks are $(0.35,0.40,0.25)$. The leverage-policy weights are those in \Cref{tab:leverage-mix-parameters}.

\begin{table}[H]
\centering
\scriptsize
\caption{Agent configuration mechanisms and interest shares.}
\label{tab:agent-configuration-registry}
\begin{tabular}{lcccccc}
\toprule
Configuration & Dynamic rate & Dynamic reserve & LBP & Bonded MM & Loss queue & Risk leverage\\
\midrule
Senior only & No & No & No & No & No & No\\
Segmented & No & No & Yes & No & Yes & No\\
Adaptive & Yes & Yes & Yes & Yes & Yes & No\\
Full Axient & Yes & Yes & Yes & Yes & Yes & Yes\\
\bottomrule
\end{tabular}
\end{table}

The Markov regime transition matrix is
\[
\begin{pmatrix}
0.955&0.042&0.003\\
0.220&0.720&0.060\\
0.120&0.380&0.500
\end{pmatrix},
\]
for normal, stress, and crisis. Regime multipliers are $(1.0,2.5,6.0)$. Per-week probabilities by regime are:
\begin{itemize}[leftmargin=*]
\item common venue: $(0.001,0.005,0.030)$;
\item stablecoin depeg: $(0.0005,0.0025,0.018)$;
\item oracle delay: $(0.002,0.010,0.050)$; and
\item strategic MM withdrawal: $(0.010,0.050,0.180)$.
\end{itemize}

The dynamic reserve target used in the reference model is the sum of configured market and pool target fractions applied to outstanding credit plus $2.5$ times perceived loss applied to Senior NAV. When below target, the reserve share can increase by at most eight percentage points, funded 60 percent from the Senior income share and 40 percent from the protocol income share. This rule is a transparent fixture, not an optimal policy claim.

\subsection{Agent hypothesis registry}
The seven thresholds H1--H7 are reproduced in \Cref{tab:agent-hypothesis-definitions}. Their final realized values are in \Cref{tab:agent-hypothesis-results}. The parameter-file hash and output hashes are recorded in the agent manifest. Registration before the final execution is verifiable from the release files but was not performed by an independent registry.

\subsection{Release-registered weekly agent model}
The r0.3.1 agent experiment uses explicit PCG64 seed 20260720, 96 paths, 104 weekly periods, three isolated pools, four protocol configurations, and three leverage policies. Each pool contains 64 Senior LPs, 24 LBPs, 12 market makers, 16 liquidators, and 96 traders. The regime transition matrix is given in \eqref{eq:agent-regime-matrix}. The pool parameters, interest shares, shock channels, and seven hypothesis thresholds are stored in \texttt{agent\_config/agent\_parameters\_r0.3.1.json}. The registration is release-local and was completed before final execution; it is not an external preregistration.

The seven thresholds are: 35 percent relative Senior-loss-incidence reduction; 70 percent preservation of benchmark accepted credit; 60 percent normal-regime LBP participation; 15 percent strategic-withdrawal shortfall reduction from bonded MM support; request-time advantage at most $10^{-10}$; common-shock correlation above 0.40; and at least 25 percent relative Senior-loss-incidence increase under the $5\times$-heavy mix versus $2\times$ only.

\subsection{Independent daily robustness model}
The secondary implementation uses seeds 20260720, 20260721, and 20260722. The ordinary panel contains 600 paths per policy over 365 days. The forced-scenario panel contains 100 paths per policy-scenario over 180 days. Sensitivity cells contain 120 paths over 150 days. It uses two pools and two venue states, with 40 Senior agents, 20 LBP agents, and six market makers per pool. Its parameter registry is \texttt{config/agent\_based\_parameters.json}. The daily model is a robustness implementation, not a second preregistered test family; it does not alter the weekly model's hypothesis verdicts.

\section{Reproducibility and Release Structure}
\label{app:reproducibility}

\subsection{Software environment}
The full research package contains:
\begin{itemize}[leftmargin=*]
\item the complete LaTeX manuscript and BibTeX database;
\item the r0.2.2 exact-decimal and integer implementation-parity suites;
\item the release-registered weekly agent model and parameter registry;
\item the independently coded daily robustness model and parameter registry;
\item CSV, JSON, compressed path outputs, generated TeX tables, and figure PDFs;
\item invariant, hypothesis, release-parity, and agent-release verification reports; and
\item SHA-256 manifests at the experiment and release levels.
\end{itemize}
The requirements file pins NumPy, pandas, and Matplotlib. Figures are generated as vector PDFs with embedded TrueType fonts and fixed metadata.

\subsection{Reproduction commands}
From the release root:
\begin{lstlisting}[language=bash]
python -m venv .venv
. .venv/bin/activate
pip install -r requirements.txt
python src/simulate.py
python src/implementation_parity.py
python src/verify_release.py
# Weekly r0.3.1 evidence is reproduced in the pinned container.
./scripts/reproduce_r031_container.sh

# Independent daily robustness model and publication layer.
python src/agent_based.py
python src/build_daily_robustness_tables.py
python visualization/render_axient_figures.py
python tools/verify_r032_release.py
pdflatex main.tex
bibtex main
pdflatex main.tex
pdflatex main.tex
\end{lstlisting}
The convenience script \texttt{reproduce.sh} executes the local exact-accounting, independent-daily, visualization, and manuscript checks, subject to local LaTeX availability. The registered weekly model is reproduced separately through \texttt{scripts/reproduce\_r031\_container.sh}, because its scientific identity includes the pinned Linux image, dependency lock, locale, timezone, single-thread numerical-library settings, and explicit PCG64 generator. Runtime is not a protocol performance claim.

\subsection{Original exact and stochastic outputs}
The r0.2.2 layer writes:
\begin{itemize}[leftmargin=*]
\item \texttt{deterministic\_fixtures.csv} and \texttt{implementation\_parity\_fixtures.csv};
\item integer rounding and implementation-parity summaries;
\item the finite scenario grid, Monte Carlo summary, leverage-mix table, queue stress, and two-pool contagion outputs;
\item seven original figure PDFs; and
\item \texttt{simulation\_manifest.json} and \texttt{release\_parity.txt}.
\end{itemize}
The 28 fixtures and 31,082 deterministic checks are retained unchanged from r0.2.2.

\subsection{Weekly agent outputs}
The canonical weekly source is \path{src/agent_market.py}. It reads the r0.3.1 registry at \path{agent_config/agent_parameters_r0.3.1.json}, uses an explicit PCG64 generator in the fixed Linux runtime, and writes:
\begin{itemize}[leftmargin=*]
\item run, pool, regime, event, weekly-quantile, and hypothesis tables under \texttt{agent\_outputs/};
\item the request-time queue fixture and invariant reports;
\item six vector figures under \texttt{agent\_figures/};
\item exact LaTeX table rows generated by \path{src/build_agent_tables.py}; and
\item \texttt{agent\_simulation\_manifest.json}, including the parameter hash, source hashes, path count, horizon, hypothesis count, and 70,207,488 scalar invariant evaluations.
\end{itemize}
The weekly execution is the controlling source for H1--H7.

\subsection{Independent daily robustness outputs}
The secondary source \texttt{src/agent\_based.py} reads \texttt{config/agent\_based\_parameters.json} and writes:
\begin{itemize}[leftmargin=*]
\item ordinary policy paths and summaries;
\item forced venue, withdrawal, and combined scenario panels;
\item cash-buffer/social-run, LBP-capital, and leverage-mix sensitivity tables;
\item six vector figures under \texttt{figures/abm\_*};
\item a 27-check structural verification report; and
\item a separate manifest with parameter, source, figure, and output hashes.
\end{itemize}
The daily implementation is a robustness model. Its numerical levels are not pooled with the weekly hypothesis experiment.

\subsection{Release verifiers}
The r0.3.1 weekly evidence is checked against its 26-file registered bundle, manifest, exact H1--H7 values, zero-failure invariant report, and the common hash from three clean container replays. The inherited r0.2.2 verifier independently confirms 28 fixtures, 31,082 deterministic accounting checks, 84 scenario rows, 200,000 Monte Carlo paths, 20,000 withdrawal-queue paths per cell, and 500,000 two-pool contagion paths. The daily robustness manifest and 27-check structural report remain separate. The top-level r0.3.3 publication verifier composes these evidence classes, verifies that manuscript tables are byte-identical to the registered weekly tables, checks the publication figure and PDF manifests, and scans the revised prose for stale release values. A passing verifier establishes artifact and code-output consistency, not behavioral or empirical validity.

\subsection{Determinism and evidence labels}
The decimal and integer fixtures are deterministic under the same Python arithmetic. The registered weekly model uses an explicit PCG64 generator under its pinned Linux and dependency contract; the independent daily model uses its own fixed-seed implementation and is retained as a separate robustness layer. Every agent output is labeled synthetic. Reproduction confirms that the code implements the registered generators; it does not establish that those generators match real providers, market makers, venues, or traders.

Release registration records H1--H7 before final execution but is not external preregistration. Any empirical calibration must use a new version, explicit holdout data, and distinct artifact names rather than silently replacing the r0.3.0 or r0.3.1 synthetic outputs.

\subsection{Reference-code boundary}
The tagged private r0.2.2 snapshot imports the same named accounting fixtures and invariant IDs and exercises reference Solidity modules through unit, bounded-fuzz, stateful-invariant, and Python-to-Solidity differential tests. The r0.3.1 deterministic agent layer is research code, not a smart-contract implementation target in that tag. A later engineering release should import agent parameters and output contracts only where they correspond to enforceable protocol policies; behavioral response rules remain off-chain research assumptions.

\subsection{Release classes}
The minimal arXiv package contains only source files required to compile the manuscript. The public research package adds code, parameters, outputs, figures, and public QA. A confidential package may additionally contain private repository snapshots or internal audits. The cumulative Axient archive nests prior immutable releases rather than overwriting them. Readers should verify SHA-256 manifests before relying on any table or figure.

\subsection{Publication visualization layer}
The r0.3.3 reproducible publication identity does not alter the r0.3.1 registered numerical CSV, JSON, or LaTeX-table artifacts. It re-renders manuscript figures from those canonical artifacts and from the unchanged r0.2.2 and independent-daily robustness outputs using the versioned renderer at \path{visualization/render_axient_figures.py}. The renderer applies the Axient academic palette, explicit Inter typography, direct labeling, and vector PDF/SVG output inside a pinned Linux runtime. Figure-level presentation changes are therefore separated from the deterministic numerical identity. Any future change to a plotted datum requires a new numerical release identity rather than a visualization-only revision.

\section{Repository-Linked Reference Snapshot}
\label{app:repository-reference}

\subsection{Immutable identity}
The r0.2.2 research release links the manuscript to the following access-controlled engineering snapshot. Repository privacy is an access-control property, not part of the mathematical claim; the complete source export is included in the full research package.

\begin{table}[H]
\centering
\small
\caption{Canonical reference snapshot metadata.}
\label{tab:repository-snapshot}
\begin{tabularx}{\textwidth}{lX}
\toprule
Field & Value\\
\midrule
Repository & \texttt{AxientLab/axient-workspace} (private)\\
Branch & \texttt{research/paper2-r0.2.1-implementation-parity}\\
Tagged commit & \texttt{1cbbd25d3fd8e0dd7569111c0f4c824f564b857f}\\
Annotated tag & \texttt{paper2-r0.2.1-reference.1}\\
Tag-object SHA & \texttt{e9087c1b4f9932b412655d30592aa009da802f71}\\
Source-export SHA-256 & \texttt{052594276545005882824b61513af2a7}\newline\texttt{fe592e2edc61852228b5cd6c23fd89a3}\\
Public-claim snapshot SHA-256 & \texttt{eb71866675f5f7a58e1a9da891771aefe}\newline\texttt{987834bcf9de8dc0a1cac4575bf031f}\\
\bottomrule
\end{tabularx}
\end{table}

The GitHub-generated source export carries the tagged commit as its archive comment. The annotated-tag object identifier is release metadata recorded by the project owner; the full release preserves both values rather than inferring one from the other.

\subsection{Canonical commands and evidence}
The repository workflow uses Python 3.11, Foundry, and a clean LaTeX runner. The canonical research/reference gates are:
\begin{lstlisting}[language=bash]
python3.11 tools/build.py
python3.11 tools/scan.py --json
make verify-paper2-r021
python3.11 contracts/scripts/generate_paper2_differential_vectors.py
forge test --root contracts

env TEXINPUTS=. pdflatex -interaction=nonstopmode -halt-on-error main.tex
\end{lstlisting}
The tagged reports record 28 fixtures, 31,082 deterministic checks, 45 passing Foundry test functions, 10 fuzz functions with 2,560 bounded cases, six stateful invariants configured for 96 runs by 48 calls, and three deterministic Python-to-Solidity differential tests. These counts are release evidence, not economic or security guarantees.

\subsection{Traceability chain}
The release links each normative implementation-parity invariant to a claim identifier, a named reference fixture, a protocol-specification section, a Solidity module, and an executable test boundary. The four-dimensional claim model records research status, reference-code status, independent-audit status, and production status separately. Public documentation consumes an immutable claim snapshot rather than reinterpreting manuscript prose.

\subsection{Snapshot verification note}
The supplied project-level checksum list contains 288 entries. Against the GitHub-generated source archive, 286 entries verify directly. One generated documentation-build manifest is not present in the source export, and one nested checksum file changed after the project-level list was generated. The r0.2.2 package therefore does not treat that list as the canonical archive checksum. It records the source-export SHA-256 above and supplies a new independent file manifest over the archived snapshot. This qualification concerns release packaging, not the 28-fixture and 31,082-check execution result, which was reproduced from the imported source suites.

\subsection{Deterministic agent-evidence snapshot}
The agent-based numerical evidence used by this publication has a separate immutable identity from the r0.2.2 contract-reference snapshot. The project owner supplied the following private-repository metadata, while the uploaded source archive independently carries the tagged commit in its ZIP comment and the public research package preserves the complete deterministic evidence bundle.

\begin{table}[H]
\centering
\small
\caption{Deterministic agent-evidence metadata.}
\label{tab:agent-repository-snapshot}
\begin{tabularx}{\textwidth}{lX}
\toprule
Field & Value\\
\midrule
Repository & \texttt{AxientLab/axient-workspace} (private)\\
Branch & \texttt{research/paper2-r0.3.0-agent-economics}\\
Tagged commit & \nolinkurl{38a0360feb66953e78a4113358b37258da3e113f}\\
Annotated tag & \texttt{paper2-r0.3.1-deterministic}\\
Tag-object SHA & \nolinkurl{7463b0b4235861e0a0d1af63db8dc186e3b87a52}\\
Canonical numerical-bundle SHA-256 & \texttt{d240f0a805fe7996e5f17e87284c1b01}\newline\texttt{597fb8b885e5b989f92ff6a4c1df8429}\\
Base image digest & \texttt{python@sha256:}\newline\texttt{b18992999dbe963a45a8a4da40ac2b1}\newline\texttt{975be1a776d939d098c647482bcad5cba}\\
Dependency-lock SHA-256 & \texttt{5c65e9817793c5f7a93c3cf46dceba3}\newline\texttt{f8d3034d7f447532de84fe07ae91f104a}\\
Historical r0.3.1 manuscript PDF SHA-256 & \texttt{486fb9982d17227e0eae45ff92631ecc}\newline\texttt{09f8fd012a2b83b7ee84024247d8c06f}\\
\bottomrule
\end{tabularx}
\end{table}

The r0.3.0 numerical artifacts are preserved as historical outputs but are not the controlling evidence for this manuscript. They were generated under an incomplete historical runtime/RNG contract and differed from three mutually identical clean runs in the pinned environment. The r0.3.1 identity fixes the generator to PCG64, pins Python and dependencies, and obtains the same canonical bundle hash across all clean replays. The present r0.3.3 reproducible publication identity changes manuscript serialization and figure rendering only; it does not modify the r0.3.1 CSV, JSON, or generated LaTeX-table values.

The r0.3.1 release-decision note contains a stale manuscript-PDF checksum. The authoritative r0.3.1 value is recorded in \Cref{tab:agent-repository-snapshot}. The r0.3.2 packages are retained as untagged historical publication candidates. The r0.3.3 package uses newly generated independent manifests and does not rely on either historical publication PDF for its identity.

\subsection{Access and archival boundary}
The minimal arXiv source package contains only manuscript compilation assets. The full research release includes the private-repository source export, its independent manifest, the claim snapshot, test reports, and repository metadata. A reader can therefore inspect and execute the tagged reference source without GitHub access. No private credentials, Git history, environment files, or deployment secrets are included.

\section{Agent-Based Capital-Market Model}
\label{app:agent-model}

This appendix specifies the release-registered weekly agent model used in \Cref{sec:agent-design,sec:agent-results}. The model is a synthetic mechanism laboratory. It is not fitted to Axient, a prediction venue, or a deployed lending protocol. Its purpose is to expose feedback among capital supply, credit demand, execution capacity, withdrawal coordination, and loss allocation under transparent rules.

\subsection{Simulation identity and scope}
The canonical source is \path{src/agent_market.py}; the parameter registry is \path{agent_config/agent_parameters_r0.3.1.json}. The final execution uses explicit PCG64 seed 20260720, 96 paired paths, 104 weekly periods, three isolated pools, four protocol configurations, and three leverage policies. The parameter-file SHA-256 and every generated artifact hash are recorded in \path{agent_outputs/agent_simulation_manifest.json}.

The three pools are:
\begin{table}[H]
\centering
\caption{Registered pool parameters in the weekly agent model.}
\label{tab:agent-pool-parameters}
\small
\begin{tabular}{lrrrrr}
\toprule
Pool & Senior target & Trader collateral & Haircut & Risk & Max $L$\\
\midrule
Macro & 10.0M & 4.6M & 1.0\% & 0.65 & $5\times$\\
Election & 8.0M & 3.6M & 1.4\% & 0.90 & $5\times$\\
Sports & 7.0M & 3.2M & 2.0\% & 1.20 & $3\times$\\
\bottomrule
\end{tabular}
\end{table}

Each pool contains 64 Senior LP agents, 24 LBP agents, 12 market makers, 16 liquidators, and 96 trader agents. Endowments and behavioral coefficients are heterogeneous fixed-seed draws. Common random numbers are reused across protocol configurations so that comparisons are paired at the path level.

\subsection{Regime process}
The common regime $z_t\in\{0,1,2\}$ represents normal, stress, and crisis states and follows the transition matrix
\begin{equation}
P=
\begin{pmatrix}
0.955 & 0.042 & 0.003\\
0.220 & 0.720 & 0.060\\
0.120 & 0.380 & 0.500
\end{pmatrix}.
\label{eq:agent-regime-matrix}
\end{equation}
The regime multiplies base loss and execution intensities by $(1,2.5,6.0)$. Conditional incident channels include common venue failure, stablecoin impairment, oracle delay, strategic market-maker withdrawal, and idiosyncratic pool shocks. The probabilities are author-specified and are not empirical frequency estimates.

\subsection{Senior participation rule}
Senior agent $i$ has active capital $s_{i,t}$, maximum capital $\bar s_i$, reservation return $o_i$, impairment aversion $\lambda_i$, queue aversion $\kappa_i$, and peer-withdrawal sensitivity $\eta_i$. Its desired participation fraction is a bounded logistic response to
\begin{equation}
\Delta^S_{i,t}=\widehat y^S_t-o_i-
\lambda_i\widehat\mu^S_t-
\kappa_i\widehat q_t-
\eta_i\widehat w_t.
\label{eq:agent-senior-score}
\end{equation}
Positive changes create deposits subject to pool capacity. Negative changes create share-denominated queued withdrawal requests. In configurations with loss-participating queues, requested shares remain in NAV and loss allocation until payment.

\subsection{LBP participation rule}
LBP agent $j$ compares expected enhanced compensation with outside return, impairment, and encumbrance:
\begin{equation}
\Delta^J_{j,t}=\widehat y^J_t-o^J_j-
\lambda^J_j\widehat\mu^J_t-
\kappa^J_j\widehat\ell_t.
\label{eq:agent-lbp-score}
\end{equation}
A logistic response maps this score to desired junior capital. LBP capital is not lendable Senior cash. It is a separate provider claim, can be encumbered, and is impaired before Senior principal.

\subsection{Trader demand and leverage selection}
Trader $n$ has collateral endowment, perceived edge, risk aversion, and leverage preference. It evaluates the admissible set $\mathcal L_t\subseteq\{2,3,5\}$ with the utility in \eqref{eq:trader-tier-utility}. The full configuration compresses the set in stress and crisis states and applies the pool-specific cap. Requested debt equals collateral times $L-1$. Admission is the minimum of:
\begin{enumerate}[leftmargin=*]
\item available Senior cash after the withdrawal buffer;
\item utilization headroom;
\item ordinary and bonded execution capacity after haircuts;
\item pool, market, event-class, and leverage limits; and
\item the risk-sensitive state cap.
\end{enumerate}
Rejected demand is recorded rather than assumed to disappear without measurement.

\subsection{Market-maker and liquidator rules}
Ordinary MM quote participation is a logistic function of expected spread income minus inventory and adverse-selection costs. It is revocable. Bonded capacity is admitted only when the commitment premium and expected execution compensation exceed lock cost and expected slashing. A strategic withdrawal shock removes ordinary capacity and can trigger bonded delivery or slashing.

Liquidator $\ell$ enters when the bounty condition in \eqref{eq:liquidator-entry} holds. Aggregate submitted capacity is capped by active liquidator resources. The model distinguishes submitted capacity, venue recovery, and settlement recovery; it does not equate entry with successful fill.

\subsection{Position cohorts and settlement}
New positions are aggregated into maturity cohorts of one, two, or four weeks with probabilities $(0.35,0.40,0.25)$. Each cohort stores principal, gross exposure, accrued interest, MM-supported quantity, and bond allocation by pool and leverage tier. When a cohort matures, the model:
\begin{enumerate}[leftmargin=*]
\item realizes regime-, pool-, leverage-, venue-, and close-pressure-dependent recovery;
\item applies ordinary and bonded execution capacity;
\item recognizes delivered or slashed MM value;
\item applies settlement and oracle-delay effects;
\item repays recognized debt and routes interest shares; and
\item applies any remaining principal shortfall through the exact waterfall.
\end{enumerate}
No matched or provisional value reduces debt before the synthetic settlement step.

\subsection{Dynamic rate and reserve controller}
The adaptive rate follows the two-slope curve from \Cref{sec:interest}, with registered parameters
\begin{equation}
r_0=0.06,
\qquad U^\star=0.80,
\qquad s_1=0.12,
\qquad s_2=0.65.
\label{eq:agent-rate-parameters}
\end{equation}
The rate is capped at 150 percent annualized in the simulator to prevent numerical runaway; the cap is a synthetic policy parameter. The dynamic reserve target equals the local market and pool target applied to outstanding credit plus $2.5$ times the perceived-loss signal applied to Senior NAV. When below target, up to eight percentage points of gross interest allocation can be redirected from Senior and protocol revenue to reserve replenishment under the registered rule.

\subsection{Weekly transition algorithm}
\begin{algorithm}[H]
\caption{Release-registered weekly agent transition}
\label{alg:agent-weekly}
\begin{algorithmic}[1]
\Require prior pool states, common regime, fixed agent traits, protocol configuration
\State transition the common regime and draw common and idiosyncratic shocks
\State accrue interest and identify maturity cohorts
\State obtain ordinary MM depth, bonded capacity, and active liquidator capacity
\State realize execution, settlement, MM delivery or slashing, and principal shortfall
\State apply the deterministic waterfall and fee routing
\State update loss, queue, MM-default, and reserve signals
\State update Senior and LBP desired capital; add deposits or queue withdrawals
\State service the share-denominated queue from eligible settled cash
\State solve the damped utilization response and admit new tiered credit
\State check non-negativity, capacity, queue, debt, reserve, and waterfall invariants
\State emit pool-week, run-level, event, and hypothesis evidence
\end{algorithmic}
\end{algorithm}

\subsection{Configurations and interest allocation}
The Senior-only configuration uses a static 15 percent borrow rate, no LBP or reserve, revocable ordinary MM capacity, and a non-loss-participating queue. Segmented adds local reserve layers, LBP capital, and a loss-participating queue. Adaptive adds the utilization rate, dynamic reserve target, and bonded MM commitments. Full Axient adds risk-sensitive leverage gates and class-specific leverage caps.

The phased gross-exposure mix is $(0.55,0.30,0.15)$ across $(2\times,3\times,5\times)$. The $5\times$-heavy mix is $(0.25,0.25,0.50)$. Interest shares by configuration are stored in the parameter registry and sum exactly to one.

\subsection{Evidence and verification}
The canonical execution writes run-, pool-, regime-, event-, and weekly-quantile tables. It evaluates 5,850,624 scalar state checks per configuration-policy cell, or 70,207,488 in aggregate, with zero failures in the release run. These checks cover arithmetic and state constraints only. They do not establish behavioral validity, parameter realism, or production safety.

\section{Independent Daily Robustness Model}
\label{app:agent-diagnostics}

The weekly release-registered agent model is the controlling experiment for the hypothesis verdicts in \Cref{sec:agent-results}. This appendix reports an independently coded daily robustness model. Its purpose is not to average two calibrations or to select preferred numbers. It asks whether selected qualitative mechanisms survive a materially different time scale, pool aggregation, path count, and scenario design.

\subsection{Design differences}
The daily model in \texttt{src/agent\_based.py} uses two pools, two venue states, 40 heterogeneous Senior agents and 20 LBP agents per pool, six market makers per pool, and trader demand aggregated over $2\times$, $3\times$, and $5\times$ cohorts. The main panel uses 600 paths per policy over 365 days. The registered scenario panel uses 100 paths per policy-scenario over 180 days, and each sensitivity cell uses 120 paths over 150 days. The seed registry is 20260720--20260722. All paths and parameters are synthetic.

The four policies are:
\begin{enumerate}[leftmargin=*]
\item static untranched;
\item adaptive isolated without LBP capital;
\item full isolated with reserve and LBP capital; and
\item full with a shared global reserve.
\end{enumerate}
The daily model marks provider wealth as outside cash plus the value of active and queued shares under conservative holdbacks. A negative provider return can therefore arise from liquidity or uncertainty marking even when no finalized Senior principal write-down occurs.

\subsection{Ordinary synthetic panel}
\begin{table}[H]
\centering
\scriptsize
\caption{Independent daily robustness panel. Each policy contributes 1,200 pool-path observations. ``Material Senior loss'' exceeds 0.1 percent of initial Senior supply. Results are synthetic and are not directly comparable in level with the weekly agent model.}
\label{tab:daily-robustness-pooled}
\begin{tabular}{lrrrrrr}
\toprule
Policy & Senior ret. & Material loss & Run & Rejected & Hard-flat & LBP ret.\\
\midrule
Static untranched & 0.44\% & 3.08\% & 51.2\% & 94.4\% & 99.83\% & -- \\
Adaptive isolated & 2.13\% & 0.67\% & 0.0\% & 35.1\% & 99.90\% & -- \\
Full isolated & 2.71\% & 0.00\% & 0.0\% & 8.2\% & 99.84\% & 1.34\% \\
Full + shared reserve & 2.73\% & 0.00\% & 0.0\% & 6.8\% & 99.86\% & 1.36\% \\
\bottomrule

\end{tabular}
\end{table}

The untranched policy exhibits both material credit loss and frequent queue stress. In the ordinary synthetic panel, the full policies eliminate material Senior write-downs within the registered path set while admitting materially more demand. This is a scenario-conditional comparison, not a universal guarantee.

\subsection{Adversarial scenario panel}
\begin{table}[H]
\centering
\scriptsize
\caption{Daily robustness scenario panel. Returns are conservative provider-wealth changes over the scenario horizon, not annualized APY.}
\label{tab:daily-robustness-scenarios}
\begin{tabular}{llrrrrr}
\toprule
Scenario & Policy & Senior ret. & Material loss & Run & Hard-flat & Rejected\\
\midrule
Stochastic & Adaptive & 1.1\% & 1.5\% & 0.0\% & 99.9\% & 39.0\% \\
Stochastic & Full isolated & 1.4\% & 0.0\% & 0.0\% & 99.8\% & 7.6\% \\
Stochastic & Full shared & 1.4\% & 0.0\% & 0.0\% & 99.9\% & 7.5\% \\
Venue stress & Adaptive & 1.1\% & 1.5\% & 0.0\% & 98.6\% & 43.7\% \\
Venue stress & Full isolated & 0.6\% & 0.0\% & 0.0\% & 98.6\% & 8.0\% \\
Venue stress & Full shared & 0.9\% & 0.0\% & 0.0\% & 98.5\% & 11.8\% \\
Withdrawal contagion & Adaptive & 1.3\% & 1.0\% & 0.0\% & 99.8\% & 35.9\% \\
Withdrawal contagion & Full isolated & 1.6\% & 0.0\% & 0.0\% & 99.8\% & 12.7\% \\
Withdrawal contagion & Full shared & 1.6\% & 0.0\% & 0.0\% & 99.8\% & 11.4\% \\
Combined & Adaptive & -11.6\% & 3.5\% & 13.0\% & 97.9\% & 48.3\% \\
Combined & Full isolated & -25.4\% & 0.0\% & 11.0\% & 97.6\% & 25.5\% \\
Combined & Full shared & -24.4\% & 0.0\% & 14.0\% & 97.5\% & 25.9\% \\
\bottomrule

\end{tabular}
\end{table}

Under the combined venue-and-withdrawal shock, the full policies preserve Senior principal against the registered credit-loss threshold but generate negative conservative provider-wealth returns and non-zero run incidence. This distinction is economically important: a credit waterfall can protect principal against finalized position shortfall without guaranteeing withdrawal liquidity, mark stability, or immediate redeemability.

\begin{figure}[H]
\centering
\includegraphics[width=0.80\textwidth]{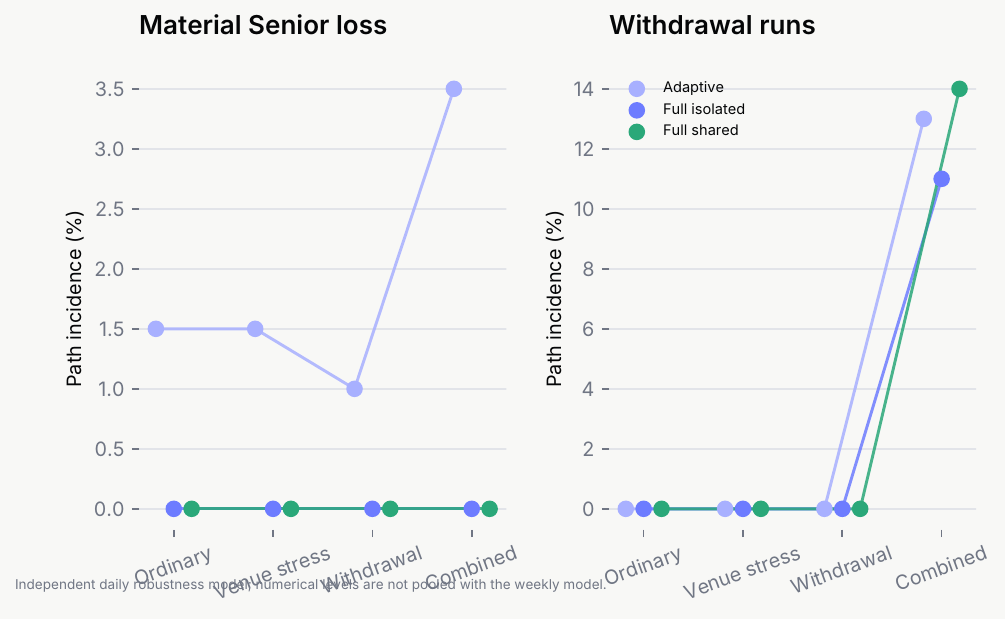}
\caption{Material Senior-loss incidence and withdrawal-run incidence in the independent daily scenario panel. Absolute values are specific to its author-specified generator.}
\label{fig:daily-scenario-panel}
\end{figure}

\subsection{Withdrawal phase transition}
The sensitivity sweep varies a social-withdrawal multiplier and the cash buffer held back from origination. The resulting phase map shows a sharp transition: with a two percent cash buffer, the run incidence rises from below one percent at unit social sensitivity to 35 percent at $2\times$ sensitivity and above 97 percent at $4\times$ or $6\times$. A ten percent buffer prevents registered runs in every cell of this grid, although it also reduces lendable capital.

\begin{figure}[H]
\centering
\includegraphics[width=0.78\textwidth]{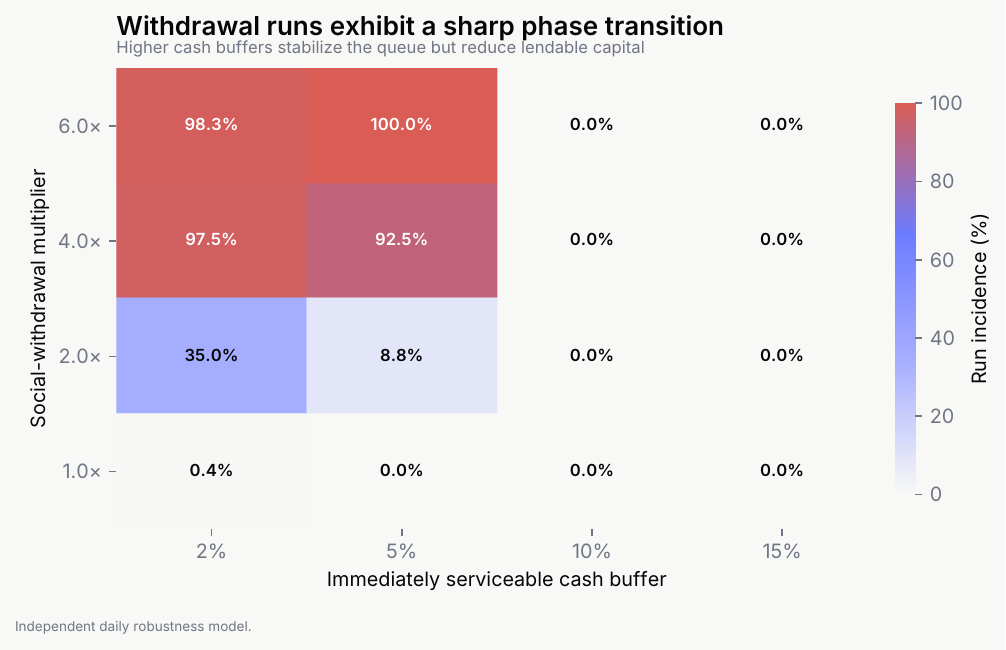}
\caption{Synthetic withdrawal-run phase map. A cash buffer is a liquidity policy, not a solvency proof.}
\label{fig:daily-run-map}
\end{figure}

The finding is consistent with \Cref{prop:run-contraction}: once peer-response amplification exceeds service capacity, the queue can become self-reinforcing. The buffer changes the service coefficient but carries an opportunity cost through lower utilization.

\subsection{Leverage-mix sensitivity}
\begin{table}[H]
\centering
\scriptsize
\caption{Daily full-shared policy under the combined scenario as the $5\times$ gross-demand share increases.}
\label{tab:daily-leverage-sensitivity}
\begin{tabular}{rrrrrr}
\toprule
$5\times$ share & Tier weights & Material loss & Senior ret. & Rejected & Hard-flat\\
\midrule
0\% & 0.647/0.353/0.000 & 0.0\% & -24.9\% & 20.0\% & 98.5\% \\
15\% & 0.550/0.300/0.150 & 0.0\% & -25.1\% & 22.0\% & 97.5\% \\
30\% & 0.453/0.247/0.300 & 0.0\% & -25.6\% & 23.8\% & 96.3\% \\
50\% & 0.324/0.176/0.500 & 0.0\% & -27.2\% & 23.9\% & 95.1\% \\
\bottomrule

\end{tabular}
\end{table}

The registered waterfall prevents material Senior principal loss across this particular sensitivity grid, but higher $5\times$ weight worsens provider wealth, increases rejection, and reduces hard-flat success from 98.5 to 95.1 percent. The result reinforces the capital-crowding proposition without claiming that this grid locates a production-safe leverage mix.

\begin{figure}[H]
\centering
\includegraphics[width=0.76\textwidth]{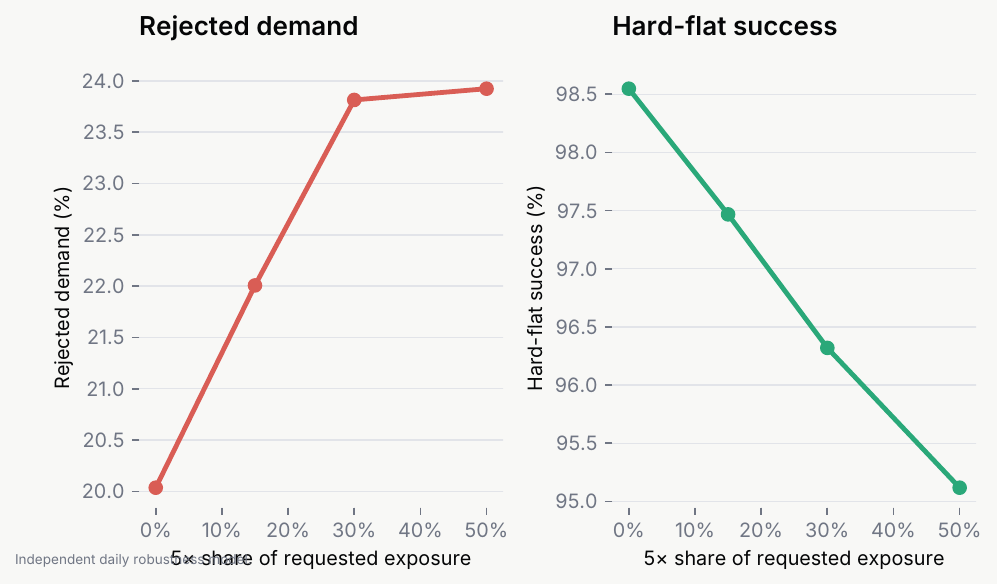}
\caption{Independent daily leverage-mix trade-off.}
\label{fig:daily-leverage-mix}
\end{figure}

\subsection{Cross-model interpretation}
The two agent implementations differ too substantially for pooled numerical inference. Their qualitative agreement is narrower:
\begin{enumerate}[leftmargin=*]
\item layered protection reduces registered Senior credit loss but does not eliminate provider liquidity or mark risk;
\item higher $5\times$ weight worsens capital efficiency and hard-flat conditions;
\item withdrawal stability depends jointly on behavioral amplification and immediately serviceable cash;
\item shared protection can reduce local credit loss while retaining common-factor and queue contagion; and
\item LBP capital is a priced, impairable provider layer rather than a guaranteed insurance fund.
\end{enumerate}
These are robustness observations across two synthetic mechanisms. They remain hypotheses for empirical calibration.

\section*{Acknowledgments}
The author acknowledges the ForesightFlow research programme for the empirical and theoretical foundation on event-linked perpetuals, informed-flow detection, manipulation risk, and prediction-market microstructure on which this work builds. This paper is a standalone Axient protocol-design and agent-based validation study and is not numbered within the four-paper ForesightFlow Event-Linked Perpetuals programme. The manuscript text is intended for distribution under CC BY 4.0 and the accompanying code under the MIT License.

\paragraph{Computational assistance disclosure.}
Automated language and code-assistance tools, including generative language models, were used during manuscript preparation for editorial review, consistency checks, LaTeX typesetting support, and development and testing of deterministic and synthetic verification scripts. The author specified and reviewed the research question, protocol design, assumptions, formal statements, proofs, numerical configurations, interpretation, and final manuscript, and takes responsibility for the complete work and accompanying materials. These tools were not treated as evidentiary sources and are not credited with authorship.

\bibliographystyle{plainnat}
\bibliography{references}

@article{brunnermeier2009,
  author={Brunnermeier, Markus K. and Pedersen, Lasse Heje}, title={Market Liquidity and Funding Liquidity}, journal={Review of Financial Studies}, volume={22}, number={6}, pages={2201--2238}, year={2009}
}

@article{diamond1983,
  author={Diamond, Douglas W. and Dybvig, Philip H.}, title={Bank Runs, Deposit Insurance, and Liquidity}, journal={Journal of Political Economy}, volume={91}, number={3}, pages={401--419}, year={1983}
}

@article{eisenberg2001,
  author={Eisenberg, Larry and Noe, Thomas H.}, title={Systemic Risk in Financial Systems}, journal={Management Science}, volume={47}, number={2}, pages={236--249}, year={2001}
}

@article{hanson2003,
  author={Hanson, Robin}, title={Combinatorial Information Market Design}, journal={Information Systems Frontiers}, volume={5}, pages={107--119}, year={2003}
}

@article{wolfers2004,
  author={Wolfers, Justin and Zitzewitz, Eric}, title={Prediction Markets}, journal={Journal of Economic Perspectives}, volume={18}, number={2}, pages={107--126}, year={2004}
}

@article{manski2006,
  author={Manski, Charles F.}, title={Interpreting the Predictions of Prediction Markets}, journal={Economics Letters}, volume={91}, number={3}, pages={425--429}, year={2006}
}

@misc{compound2019,
  author={Leshner, Robert and Hayes, Geoffrey}, title={Compound: The Money Market Protocol}, year={2019}, howpublished={White paper}
}

@misc{aave2026,
  author={{Aave}}, title={Aave Protocol Technical Documentation}, year={2026}, howpublished={\url{https://aave.com/docs}}, note={Accessed 2026-07-14; background for pooled lending, utilization rates, reserve factors, and liquidation design}
}

@inproceedings{werner2022,
  author={Werner, Sam M. and Perez, Daniel and Gudgeon, Lewis and Klages-Mundt, Ariah and Harz, Dominik and Knottenbelt, William J.}, title={SoK: Decentralized Finance (DeFi)}, booktitle={Proceedings of the 4th ACM Conference on Advances in Financial Technologies}, year={2022}
}

@inproceedings{gudgeon2020,
  author={Gudgeon, Lewis and Perez, Daniel and Harz, Dominik and Livshits, Benjamin and Gervais, Arthur}, title={The Decentralized Financial Crisis}, booktitle={2020 Crypto Valley Conference on Blockchain Technology}, year={2020}
}

@misc{perez2020,
  author={Perez, Daniel and Werner, Sam M. and Xu, Jiahua and Livshits, Benjamin}, title={Liquidations: DeFi on a Knife-edge}, year={2020}, eprint={2009.13235}, archivePrefix={arXiv}, primaryClass={cs.CR}, note={arXiv:2009.13235}
}

@misc{qin2023miqado,
  author={Qin, Kaihua and Ernstberger, Jens and Zhou, Liyi and Jovanovic, Philipp and Gervais, Arthur}, title={Mitigating Decentralized Finance Liquidations with Reversible Call Options}, year={2023}, eprint={2303.15162}, archivePrefix={arXiv}, primaryClass={cs.CR}, note={arXiv:2303.15162}
}

@misc{bastankhah2024agilerate,
  author={Bastankhah, Mahsa and Nadkarni, Viraj and Wang, Xuechao and Viswanath, Pramod}, title={AgileRate: Bringing Adaptivity and Robustness to DeFi Lending Markets}, year={2024}, eprint={2410.13105}, archivePrefix={arXiv}, primaryClass={cs.CR}, note={arXiv:2410.13105}
}

@misc{bastankhah2024fastslow,
  author={Bastankhah, Mahsa and Nadkarni, Viraj and Wang, Xuechao and Jin, Chi and Kulkarni, Sanjeev and Viswanath, Pramod}, title={Thinking Fast and Slow: Data-Driven Adaptive DeFi Borrow-Lending Protocol}, year={2024}, eprint={2407.10890}, archivePrefix={arXiv}, primaryClass={cs.CR}, note={arXiv:2407.10890}
}

@misc{sadeghi2026,
  author={Sadeghi, Agathe and Feinstein, Zachary}, title={Liquidation Dynamics in DeFi and the Role of Transaction Fees}, year={2026}, eprint={2602.12104}, archivePrefix={arXiv}, primaryClass={q-fin.MF}, note={arXiv:2602.12104}
}

@misc{nechepurenko2026axient,
  author={Nechepurenko, Maksym}, title={Axient: Debt-Free Finality for Leveraged Binary Event Markets}, year={2026}, note={Axient research manuscript, version r0.3.1}
}

@misc{nechepurenko2026resolutionaware,
  title={Resolution-Aware Perpetual Futures on Binary Prediction Markets: An Empirical Risk-Design Framework Using {Polymarket} Data}, author={Nechepurenko, Maksym}, year={2026}, eprint={2605.10400}, archivePrefix={arXiv}, primaryClass={q-fin.TR}, doi={10.48550/arXiv.2605.10400}, note={arXiv:2605.10400}
}

@misc{nechepurenko2026taxonomy,
  title={A Taxonomy of Event-Linked Perpetual Futures: Variant Designs Beyond the Single-Market Binary Case}, author={Nechepurenko, Maksym}, year={2026}, eprint={2605.10428}, archivePrefix={arXiv}, primaryClass={q-fin.TR}, doi={10.48550/arXiv.2605.10428}, note={arXiv:2605.10428}
}

@misc{nechepurenko2026manipulation,
  title={Manipulation, Insider Information, and Regulation in Leveraged Event-Linked Markets}, author={Nechepurenko, Maksym}, year={2026}, eprint={2605.10486}, archivePrefix={arXiv}, primaryClass={q-fin.TR}, doi={10.48550/arXiv.2605.10486}, note={arXiv:2605.10486}
}

@misc{nechepurenko2026fillside,
  title={Fill-Side Non-Retail Trading on {Polymarket}: An Empirical Study of Behavioral Tiers and Microstructure Signatures Under Quote-Attribution Constraints}, author={Nechepurenko, Maksym}, year={2026}, eprint={2605.11640}, archivePrefix={arXiv}, primaryClass={q-fin.TR}, doi={10.48550/arXiv.2605.11640}, note={arXiv:2605.11640}
}

@incollection{tesfatsion2006,
  author={Tesfatsion, Leigh}, title={Agent-Based Computational Economics: A Constructive Approach to Economic Theory}, booktitle={Handbook of Computational Economics}, volume={2}, pages={831--880}, publisher={Elsevier}, year={2006}, doi={10.1016/S1574-0021(05)02016-2}
}

@article{farmer2009,
  author={Farmer, J. Doyne and Foley, Duncan}, title={The Economy Needs Agent-Based Modelling}, journal={Nature}, volume={460}, pages={685--686}, year={2009}, doi={10.1038/460685a}
}

@article{brock1998,
  author={Brock, William A. and Hommes, Cars H.}, title={Heterogeneous Beliefs and Routes to Chaos in a Simple Asset Pricing Model}, journal={Journal of Economic Dynamics and Control}, volume={22}, number={8--9}, pages={1235--1274}, year={1998}
}

@article{lux1999,
  author={Lux, Thomas and Marchesi, Michele}, title={Scaling and Criticality in a Stochastic Multi-Agent Model of a Financial Market}, journal={Nature}, volume={397}, pages={498--500}, year={1999}, doi={10.1038/17290}
}

@article{fagiolo2007,
  author={Fagiolo, Giorgio and Moneta, Alessio and Windrum, Paul}, title={A Critical Guide to Empirical Validation of Agent-Based Models in Economics}, journal={Computational Economics}, volume={30}, pages={195--226}, year={2007}, doi={10.1007/s10614-007-9106-0}
}

@article{goldstein2005,
  author={Goldstein, Itay and Pauzner, Ady}, title={Demand-Deposit Contracts and the Probability of Bank Runs}, journal={Journal of Finance}, volume={60}, number={3}, pages={1293--1327}, year={2005}, doi={10.1111/j.1540-6261.2005.00762.x}
}

@article{gai2010,
  author={Gai, Prasanna and Kapadia, Sujit}, title={Contagion in Financial Networks}, journal={Proceedings of the Royal Society A}, volume={466}, number={2120}, pages={2401--2423}, year={2010}, doi={10.1098/rspa.2009.0410}
}

@article{acemoglu2015,
  author={Acemoglu, Daron and Ozdaglar, Asuman and Tahbaz-Salehi, Alireza}, title={Systemic Risk and Stability in Financial Networks}, journal={American Economic Review}, volume={105}, number={2}, pages={564--608}, year={2015}, doi={10.1257/aer.20130456}
}

@article{glosten1985,
  author={Glosten, Lawrence R. and Milgrom, Paul R.}, title={Bid, Ask and Transaction Prices in a Specialist Market with Heterogeneously Informed Traders}, journal={Journal of Financial Economics}, volume={14}, number={1}, pages={71--100}, year={1985}, doi={10.1016/0304-405X(85)90044-3}
}

@article{kyle1985,
  author={Kyle, Albert S.}, title={Continuous Auctions and Insider Trading}, journal={Econometrica}, volume={53}, number={6}, pages={1315--1335}, year={1985}, doi={10.2307/1913210}
}

@misc{chaudhary2022abm,
  author={Chaudhary, Amit and Pinna, Daniele},
  title={A Multi-Asset, Agent-Based Approach Applied to {DeFi} Lending Protocol Modelling},
  year={2022},
  eprint={2211.08870},
  archivePrefix={arXiv},
  primaryClass={q-fin.MF},
  note={arXiv:2211.08870}
}

@misc{lehalle2016adverse,
  author={Lehalle, Charles-Albert and Mounjid, Othmane},
  title={Limit Order Strategic Placement with Adverse Selection Risk and the Role of Latency},
  year={2016},
  eprint={1610.00261},
  archivePrefix={arXiv},
  primaryClass={q-fin.TR},
  note={arXiv:1610.00261}
}

@misc{fodra2012inventory,
  author={Fodra, Pietro and Labadie, Mauricio},
  title={High-Frequency Market-Making with Inventory Constraints and Directional Bets},
  year={2012},
  eprint={1206.4810},
  archivePrefix={arXiv},
  primaryClass={q-fin.TR},
  note={arXiv:1206.4810}
}

\end{document}